\documentclass[11pt,reqno]{amsart}

\makeatletter
\@ifclassloaded{siamart171218}{
  \newsiamremark{remark}{Remark}
  }{}
\@ifclassloaded{amsart}{
  \newtheorem{theorem}{Theorem}[section]
  \newtheorem{corollary}[theorem]{Corollary}
  \newtheorem{lemma}[theorem]{Lemma}
  \theoremstyle{definition}
  \newtheorem{definition}[theorem]{Definition}
  \newtheorem{assumption}[theorem]{Assumption}
  \newtheorem{remark}[theorem]{Remark}
  }{}
\makeatother

\usepackage{amsmath,amsfonts,amssymb,cancel,textgreek,mathrsfs,accents,bm,soul}
\usepackage{a4wide,xcolor,soul}
\usepackage[greek,english]{babel}
\usepackage{graphicx,color,pinlabel}
\usepackage{longtable}
\usepackage{booktabs}

\def\Expectation{\mathbb E}

\def\R{\mathbb R}

\def\N{\mathbb N}
\def\P{\mathcal P}
\def\bbP{\mathbb P}
\def\bbQ{\mathbb Q}
\def\X{\mathsf X}
\def\Y{\mathsf Y}

\def\cZ{\mathcal Z}

\def\Lebesgue{\mathscr L}
\def\weakto{\rightharpoonup}
\def\RelEnt{\mathcal H}
\def\FreeEnergy{\mathcal F}
\def\Energy{\mathcal E}
\def\hEnergy{\hat{\mathcal E}}
\def\InvMeas{\mathbb P^{\mathrm{inv}}}
\def\InvMeasf{\mathbb P^{\mathrm{inv},f}}
\def\InvMeasG{\mathbb P^{\mathrm{inv},\mathcal G}}
\def\InvMeasWiszero{\mathbb P^{\mathrm{inv,}W=0}}

\DeclareMathOperator\Prob{Prob}
\DeclareMathOperator*\argmin{arg\,min}
\DeclareMathOperator\supp{supp}
\DeclareMathOperator\grad{grad}
\DeclareMathOperator\Lip{Lip}
\DeclareMathOperator\BL{BL}
\DeclareMathOperator\EDP{GF}
\let\e\varepsilon
\let\ds\displaystyle

\usepackage{bbm} 
\usepackage{dsfont}
\def\Indicator{\mathds{1}}

\DeclareMathAccent{\widerhat}{\mathord}{largesymbols}{"62}

\def\hFreeEnergy{\hat\FreeEnergy}
\DeclareMathSymbol{\widehatsym}{\mathord}{largesymbols}{"62}
\newcommand\lowerwidehatsym{%
  \text{\smash{\raisebox{-1.3ex}{%
    $\widehatsym$}}}}
\newcommand\fixwidehat[1]{%
  \mathchoice
    {\accentset{\displaystyle\lowerwidehatsym}{#1}}
    {\accentset{\textstyle\lowerwidehatsym}{#1}}
    {\accentset{\scriptstyle\lowerwidehatsym}{#1}}
    {\accentset{\scriptscriptstyle\lowerwidehatsym}{#1}}
}
\def\Ent{{\mathcal E}\mathrm{nt}}
\def\hEnt{\fixwidehat\Ent}

\newcommand\spX{\R_x}
\newcommand\spY{\R_y}

\def\clap#1{\hbox to 0pt{\hss#1\hss}}

\usepackage{ifthen}
\newlength{\leftstackrelawd}
\newlength{\leftstackrelbwd}
\def\leftstackrel#1#2{\settowidth{\leftstackrelawd}%
{${{}^{#1}}$}\settowidth{\leftstackrelbwd}{$#2$}%
\addtolength{\leftstackrelawd}{-\leftstackrelbwd}%
\leavevmode\ifthenelse{\lengthtest{\leftstackrelawd>0pt}}%
{\kern-.5\leftstackrelawd}{}\mathrel{\mathop{#2}\limits^{#1}}}
\usepackage[bookmarks=false,hidelinks]{hyperref}
\begin{document}

\title{Large deviations and gradient flows for the Brownian one-dimensional hard-rod system}
\author{Nir Gavish\and Pierre Nyquist\and Mark Peletier}
\date{\today}

\begin{abstract}
We study a system of hard rods of finite size in one space dimension, which move by Brownian noise while avoiding overlap. We consider a scaling  in which the number of particles tends to infinity while the volume fraction of the rods remains constant; in this limit the empirical measure of the rod positions converges almost surely to a deterministic limit evolution. 
We prove a large-deviation principle on path space for the empirical measure, by exploiting a one-to-one mapping between the hard-rod system and a system of non-interacting particles on a contracted domain. The large-deviation principle naturally identifies a gradient-flow structure for the limit evolution, with clear interpretations for both the driving functional (an `entropy') and the dissipation, which in this case is the Wasserstein dissipation. 
%This large-deviation principle identifies the limit evolution as the Wasserstein-$2$ gradient flow of a non-standard entropy functional, and explains how this structure relates to the underlying particle system.

This study is inspired by recent developments in the continuum modelling of multiple-species interacting particle systems with finite-size effects; for such systems many different modelling choices appear in the literature, raising the question how one can understand such choices in terms of  more microscopic models. 
The results of this paper give a clear answer to this question, albeit for the simpler one-dimensional hard-rod system. 
For this specific system this result provides a clear understanding of the value and interpretation of different modelling choices, while giving hints for more general systems.
\end{abstract}


\maketitle

\section{Introduction}

\subsection{Continuum modelling of systems of interacting particles}
Systems of interacting particles can be observed in physics (e.g.\ gases, liquids, solutions~\cite{Ruelle69,Liggett06}), biology (e.g.\ populations of cells~\cite{TindallMainiPorterArmitage08-II}),  social sciences (e.g.\ animal swarms~\cite{Okubo86}), engineering (e.g.\ swarms of robots~\cite{BrambillaFerranteBirattariDorigo13}), and various other fields.  Such systems are routinely described with different types of models: \emph{particle-} or \emph{individual-based} models characterize the position and velocity of every single particle, while \emph{continuum} models characterize the behaviour of the system in terms of (often continuous) \emph{densities} or \emph{concentrations}. While particle-based models contain more information and may well be more accurate, continuum models are easier to analyze and less demanding to simulate, and there is a natural demand for continuum models that describe such systems in as much accuracy as possible. 

In this paper we focus on systems of particles for which  the \emph{finite size} of the particles has a prominent effect on the larger-scale, continuum-level behavior of the system. For such systems a wide range of continuum models has been postulated (see below) but very few of these continuum models have been rigorously justified, and  particularly the dynamics of these continuum models have rarely been rigorously justified. We present here a rigorous derivation of the continuum equation that describes such a class of  interacting particles, but restricted to one space dimension. The proof uses large-deviation theory, and this method identifies not only the limiting equation but also the gradient-flow structure of the limit. In this way it gives a rigorous derivation of the \emph{Variational Modelling} structure of the  limit.

%In this paper we focus on the case of \emph{sterically interacting particles}.  We adopt a variational modeling approach which follows a basic paradigm of material science, roughly speaking `a material evolves in a way that dissipates its free energy in the fastest way'.  Accordingly, the evolution of an (over-damped) system is defined by a driving functional (free energy) and a dissipation mechanism that describes how the system dissipates its free energy.  The existing literature on systems of sterically interacting particles had focused solely on the effect of the finite size of the particles on the free energy of the system, while assuming that the dissipation mechanism is not influenced.  In this paper, we focus on the effect of the finite size of the particles on the dissipation mechanism.  In particular, we consider a very simple system, and perform a rigorous upscaling based on large deviations and gradient flows to derive the dissipation mechanism. By doing this we develop a \emph{variational-modelling} understanding of this system that is founded on a rigorous upscaling of the continuum equations from the microscopic particle system.

\subsection{Sterically interacting particles}\label{sec:Sterically_interacting_particles}

In many applications, the finite size of the particles  can be witnessed in various ways, such as in the natural upper bound on the density of such particles, in the way particles `push away' other particles (\emph{cross-diffusion}, possibly leading to \emph{uphill diffusion}), and in the striking oscillations of ion densities near charged walls (see e.g.~\cite{GrootFaberVan-der-Eerden87,HyonFonsecaEisenbergLiu12,Gillespie15}).

Characterizing the macroscopic, continuum-level behaviour of such `sterically interacting' particles (from the Greek \foreignlanguage{greek}{stere\'os} for `hard, solid') is a major challenge, and various communities have addressed this challenge. In the mathematical  community, in one-dimensional systems the behaviour of single, `tagged' particles has been characterized~\cite{LebowitzPercus67,MonPercus02,LizanaAmbjornsson09}, the equilibrium statistical  properties of the ensemble were determined by Percus~\cite{Percus76},  and the dynamic continuum limit was first derived by Rost~\cite{Rost84}.
In higher dimensions cluster expansions have opened the door to accurate expansions of the free energy~\cite{JansenKonigMetzger15,Jansen15}. In a number of papers Bruna and Chapman~\cite{Bruna12TH,BrunaChapman12,BrunaChapman12a,BrunaChapman14} have given asymptotic expressions for the continuum-limit partial differential equation in the limit of low volume fraction. 
A  related line of research focuses on strongly interacting particle systems with \emph{soft} interaction; Spohn characterized the central-limit fluctuations in an infinite system of interacting Brownian particles~\cite{Spohn86}, and Varadhan proved the continuum limit as $n\to\infty$ for a system of $n$ interacting particles~\cite{Varadhan91,Uchiyama94} in one dimension.

In the chemical-engineering community there has been a strong interest in the case of finite-size particles that are \emph{charged}. For such systems the classical stationary Poisson-Boltzmann theory  (e.g.~\cite[Ch.~6]{PoonAndelman06}) and time-dependent Poisson-Nernst-Planck equations (e.g.~\cite{Rubinstein90}) both give unsatisfactory predictions, such as unphysically high concentrations near charged walls. Starting with the early work of Bikerman~\cite{Bikerman42} various authors have modified the static Poisson-Boltz\-mann theory by incorporating the entropy of solvent molecules, thus limiting the concentration of the 
ions~\cite{grimley1947general,grimley1950contact,dutta1950onthe,bagchi1950new,eigen1951theorie,wicke1952einfluss,eigen1954thermodynamics,dutta1954theory,redlich1955solutions,wiegel1993distribution,strating1993effects,strating1993distribution,kralj1994influence,kralj1996simple,BorukhovAndelmanOrland97,borukhov2000adsorption,bohinc2001thickness,bohinc2002charged,borukhov2004charge,kilic2007steric,wiegel2013physical}. The Boublik-Mansoori-Carnahan-Starling-Leland theory~\cite{Boublik70,MansooriCarnahanStarlingLeland-Jr71,DiCaprioBorkowskaStafiej03} further modifies this by adding higher-order concentration dependencies, and later works~\cite{Tresset08,liu2014poisson} generalize this to the case of ions of different sizes.  Other approaches include modelling the solvent as polarizable spheres~\cite{Lopez-GarciaHornoGrosse11}, and the addition to the free energy of convolution integrals with various kernels, such as Lennard-Jones-type kernels~\cite{EisenbergHyonLiu10,HyonFonsecaEisenbergLiu12,HsiehHyonLeeLinLiu15,gavish2018poisson} or step functions and their derivatives~\cite{Rosenfeld89,Roth10}. See~\cite{bazant2009towards} for further review and references.

Despite all this  activity, however, the main question for this paper is still open: Which continuum-level partial differential equation describes the evolution of systems of many finite-size particles, and what is the corresponding gradient-flow structure? Before describing the answer of this paper we first comment on the philosophy of \emph{Variational Modelling}, which underlies both this paper and some of the work in this area.

\subsection{Variational Modelling}
Many strongly damped continuum systems can be modelled by gradient flows; they are then fully characterized by a driving functional (e.g., a Gibbs free energy) and a dissipation mechanism that describes how the system dissipates its free energy.  
By choosing these two components one fully determines the model, and the model equations are readily derived  as an outcome of the these two inputs.
We call this way of working \emph{Variational Modelling}; recent examples can be found in e.g.~\cite{Ziegler83,Doi11,ArroyoWalaniTorres-SanchezKaurin18}, and the lecture notes~\cite{PeletierVarMod14TR} describe this modelling philosophy and its foundations in detail. 

The quality of a variational-modelling derivation rests on the quality of the two choices, the choice of the driving functional  and the choice of the dissipation (e.g., drift-diffusion or Wasserstein gradient flow).  Different combinations of choices, however, can lead to the same equation (see e.g.~\cite{PeletierRedigVafayi14} or~\cite[Eq.~(2.1)]{DondlFrenzelMielke18TR}).  Therefore, it is not possible to assess the quality of the independent modelling choices, nor to deduce that the combined choices are right, based on comparison of the model predictions to particle-based simulations, or to experimental data.  Accordingly, there is great importance in systematically determining the driving functional and dissipation mechanism from `first principles'. In recent years it has been discovered that not only the free energy, but also the dissipation mechanism can be rigorously deduced from an upscaling of the underlying particle system, by determining the large-deviation rate functional in the many-particle limit. In this way various free energies and dissipation mechanisms have been placed on a secure foundation~\cite{AdamsDirrPeletierZimmer13,MielkePeletierRenger14,PeletierRedigVafayi14,BonaschiPeletier16,MielkePeletierRenger16}.

In a series of works, Hyon, Horng, Lin, Liu, and Eisenberg applied a special case of the variational modelling approach, called `Energetic Variational Approach', to derive evolution equations for hard-sphere ions~\cite{horng2012pnp,lin2014new,hyon2011mathematical}.
However, the authors solely focused on the effect of the finite size on the free energy. The impact of the finite size on the dissipation mechanism, and therefore the dynamics,  was not considered, and implicitly taken to be the same as for zero-size particles.

Instead, in this work we develop a systematic derivation of both the driving functional and the dissipation mechanism for the system of this paper: a one-dimensional system of Brownian hard rods.

\subsection{The model: One-dimensional hard rods}
\label{subsec:model}
The system that we consider is a collection of $n$ hard rods of length $\alpha/n$, for $\alpha\in (0,1)$, that are free to move  along the real line~$\R$, except that they may not overlap. The position of each rod is given by its left-hand point~$Y_i^n$, i.e.\ the rod occupies the space $[Y_i^n,Y_i^n+\alpha/n)$. Since the rods can not overlap, the state space is the `swiss cheese' space 
\begin{equation}
\label{def:Omega-n}
\Omega_n := \Bigl\{y\in \R^n:\; \forall i,j, \ i\not=j,\; |y_i-y_j|\geq \alpha/n\Bigr\}.
\end{equation}
Note that the length $\alpha/n$ is scaled such that the total volume fraction of the rods is $O(1)$. 

The evolution of the rods is that of Brownian motion in a potential landscape with the non-overlap constraint, where we additionally allow for mean-field interaction between the particles. For this paper we choose as potentials an on-site potential $V$ and a two-particle interaction potential $W$, and both are assumed to be sufficiently smooth.

In the interior of $\Omega_n$ we therefore solve
\begin{equation}
\label{eq:SDE}
dY_i^n(t) = -V'(Y_i^n(t))\, dt - \frac1n \sum_{j=1}^n W'(Y_i^n(t)-Y_j^n(t))\, dt + dB_i(t), 
\end{equation}
where $B_i$ are independent one-dimensional standard Brownian motions. On the boundary $\partial\Omega_n$ we assume reflecting boundary conditions. 
%The function $V$ has the interpretation of an on-site potential, while the function $W$ characterizes the non-hard-sphere interactions between two particles.
%We can make this precise by specifying the generator
%\[
%\mathcal L = \frac12 \Delta + \mathrm b\cdot \nabla,\qquad\text{with }
%\mathrm b_i(x) = -\nabla V(y_i) -\frac1n \sum_{j=1}^n W'(y_i-y_j),
%\]
%and specifying the domain as $D(\mathcal L) = \{f\in C_b^2(\Omega_n): \partial f/\partial n = 0 \text{ on } \partial \Omega_n\}$. 

\bigskip

This system has been studied before. For the case $V\equiv W\equiv 0$, Percus~\cite{Percus76} calculated various distribution functions for finite $n$. Also for $V\equiv W\equiv 0$, Rost~\cite{Rost84} proved that in the the $n\to\infty$ limit, the empirical measures 
\[
\hat \rho_n(t) := \frac1n\sum_{i=1}^n \delta_{Y_i^n(t)}
\]
converge almost surely (the \emph{continuum limit}) to a solution of the nonlinear parabolic equation
\[
\partial_t \rho = \frac12 \partial_y \frac{\partial_y \rho}{(1-\alpha \rho)^2} \qquad \text{on $\R$.}
\]
Bodnar and Velazguez generalized this convergence by allowing for $n$-dependent $W_n$ that shrinks to a Dirac delta function as $n\to\infty$, leading to an additional term $\partial_y (\rho\partial_y\rho)$ in the equation above~\cite{BodnarVelazquez05}. Bruna and Chapman studied the related case of fixed $n$ in the limit $\alpha\to 0$, in arbitrary dimensions, and calculated the approximate limit equation up to order $O(\alpha)$~\cite{BrunaChapman12,BrunaChapman12a,BrunaChapman14,Bruna12TH}.

Based on analogy with continuum limits in other interacting particle systems (see e.g.~\cite{Oelschlager84,DawsonGartner87}), one would expect that the continuum limit for the case of non-zero~$V$ and~$W$ is 
\begin{equation}
\label{eq:HL}
\partial_t \rho = \frac12 \partial_y \frac{\partial_y \rho}{(1-\alpha \rho)^2}
+ \partial_y \Bigl[\rho \partial_y \bigl( V + W*\rho\bigr)\Bigr].
\end{equation}
Here $(W*\rho)(y) = \int_\R W(y-y')\rho(dy')$ is the convolution of $W$ with $\rho$.

The aims of this paper are (a) to prove the limit equation above rigorously, and (b) show that it has a variational, gradient-flow structure that is generated by large deviations in a canonical way. 
%The large-deviation origin of the variational structure places it on a secure footing. 

%
%As mentioned in the Introduction, the results of this paper will follow from connecting the stochastic particle system~\eqref{eq:SDE} with its deterministic limit~\eqref{eq:HL} at the level of \emph{large deviations} and \emph{gradient flows}. Connections between large-deviation principles and gradient-flow structures have successfully been used in a number of problems in recent years~\cite{AdamsDirrPeletierZimmer13,MielkePeletierRenger14,PeletierRedigVafayi14,BonaschiPeletier16,MielkePeletierRenger16}, and we now describe this connection at a high level of abstraction.

\subsection{Main result I: Large deviations of the invariant measure}
\label{subsec:var-structure}

The first step in probing the variational structure of equation~\eqref{eq:HL}  is to derive the `free energy' that will drive the gradient-flow evolution. For reversible stochastic processes, such as the particle system $Y^n$, it is well known (see e.g.~\cite{MielkePeletierRenger14}) that this driving functional is given by the large-deviation rate functional of the invariant measure.

\medskip

Under our conditions on $V$ and $W$, the system of particles $Y^n$ has an invariant measure 
\begin{equation}
\label{def:InvMeasn}
\InvMeas_n := \frac1{\cZ_n} \exp \biggl[\,-2 \sum_{i=1}^n V(y_i) - \frac1{n} \sum_{i,j=1}^n  W(y_i-y_j)\biggr] \, \Lebesgue^n\Big|_{\Omega_n},
\end{equation}
where $\Lebesgue^n|_{\Omega_n}$ is $n$-dimensional Lebesgue measure restricted to $\Omega_n$, and $\mathcal Z_n$ is 
the normalization constant
\begin{equation}
\label{def:norm-const-Zn-intro}
\cZ_n := \int_{\Omega_n}\exp \biggl[\,-2 \sum_{i=1}^n V(y_i) - \frac1{n} \sum_{i,j=1}^n  W(y_i-y_j)\biggr] \, dy.
\end{equation}
Our first main result identifies the large-deviation behaviour of these invariant measures. In this paper, $\P(\R)$ is the space of probability measures on $\R$.

\medskip

\begin{theorem}[Large-deviation principle for the invariant measures]
\label{th:LDP-invmeas}
Assume that the functions $V$ and $W$ satisfy Assumption~\ref{ass:VW}. 
For each $n\in\N$, let $Y^n \in \R^n$ have law  $\InvMeas_n$, and let $\rho_n := \frac1n \sum_{i=1}^n\delta_{Y^n_i} \in\P(\R)$ be the corresponding empirical measure. Then the measures $\rho_n$ satisfy a large-deviation principle with good rate function  $2\hFreeEnergy$:
\begin{equation}
\label{def:LDP-static-formal}
\Prob(\rho_n\approx  \nu) \sim e^{-n2\hFreeEnergy(\nu)} \qquad \text{as }n\to\infty.
\end{equation}
Here $\hFreeEnergy:\P(\R)\to[0,\infty]$ is given by
\begin{equation}
\label{def:F}
\hFreeEnergy (\rho)
:= \begin{cases}
\ds \int_{\R} \rho \biggl[\frac12\log \rlap{$\ds\frac{\rho}{1-\alpha \rho} + V\biggr] + \frac12 \int_{\R}\int_{\spY} W(y-y')\rho(dy)\rho(dy') + c$}\\[4\jot]
\qquad\qquad&\text{if $\rho$ is Lebesgue-absolutely-continuous and $\rho(y)<1/\alpha$ a.e.},\\[2\jot]
+\infty &\text{otherwise}
\end{cases}
\end{equation}
%In other words, 
%\begin{subequations}
%\label{def:LDP-static}
%\begin{alignat}2
%\limsup_{n\to\infty} \frac1n \log \Prob(\rho_n\in O) &\geq -\inf_{\nu\in O} 2\hFreeEnergy(\nu) &\qquad& \text{for all open $O\subset \P(\R)$;}\\
%\liminf_{n\to\infty} \frac1n \log \Prob(\rho_n\in C) &\leq -\inf_{\nu\in C} 2\hFreeEnergy(\nu) && \text{for all closed $C\subset \P(\R)$.}
%\end{alignat}
%\end{subequations}
The constant $c$ in the definition~\eqref{def:F} is chosen such that $\min \{\hFreeEnergy(\rho): \rho\in \P(\R)\} = 0$. 
\end{theorem}%

\medskip

The functional $\hFreeEnergy$ is non-negative by definition;  our assumptions on $V$ and $W$  imply that~$\hFreeEnergy$ has at least one minimizer at value zero, and possibly more than one. 

The large-deviation principle~\eqref{def:LDP-static-formal} gives a characterization of the behaviour of the empirical measures $\rho_n$  that can be split into two parts:
\begin{enumerate}
\item With probability one, and along a subsequence, $\rho_n$ converges to a minimizer of $\hFreeEnergy$. (This can be proved using the Borel-Cantelli lemma; see e.g.~\cite[Th.~A.2]{PeletierSchlottke19TR}.)
\item The event that $\rho_n\approx\nu$ where $\nu$ is \emph{not} a minimizer of $\hFreeEnergy$ becomes increasingly unlikely as $n$ tends to infinity; in fact, it is  exponentially unlikely in $n$, with a prefactor $2\hFreeEnergy(\nu)$ that depends on $\nu$. 
Large values of $\hFreeEnergy(\nu)$ correspond to `even more unlikely' behaviour of $\rho_n$ than smaller values. 
\end{enumerate}

\subsection{Gradient flows}
\label{sec:intro-GF-LDP}

We now turn to the evolution. A gradient-flow structure is defined by a state space, a driving functional, and a dissipation metric~\cite{PeletierVarMod14TR,Mielke16}. The driving functional was identified above as $\hFreeEnergy$; the large-deviation principle that we prove below will indicate that the state space for this gradient-flow structure is the metric space given by the set $\P_2(\R)$ of probability measures of finite second moment (i.e.\ $\int_\R y^2\rho(dy)<\infty$) equipped with the  Wasserstein metric. 

We describe the  Wasserstein metric and Wasserstein gradient-flow structures in more detail in Section~\ref{s:prelims}; here we only summarize a few aspects. 
The Wasserstein distance $W_2$ is a measure of distance between two probability measures on physical space. When modelling particles embedded in a viscous fluid, the appearance of the Wasserstein distance in gradient-flow structures can be traced back to the  drag force experienced by the particles when moving through the  fluid. This is illustrated by the property that if $n$ particles are dragged from positions $y_1,\dots,y_n$ to positions $\overline y_1,\dots,\overline y_n$ in time $\tau$, then the minimal viscous dissipation as a result of this motion is given by the Wasserstein distance between the two empirical measures:
\begin{equation}
\label{def:Wasserstein-intro}
\frac {cn} \tau \,W_2\biggl(\frac1n\sum_{i=1}^n \delta_{y_i} , \frac1n \sum_{i=1}^n \delta_{\overline y_i}\biggr)^2 = 
\frac c\tau \min \biggl\{ \sum_{i=1}^n \bigl| y_{\sigma(i)} - \overline y_i\bigr|^2: \; \sigma\text{ permutation of }1,\dots,n\; \biggr\}.
\end{equation}
The drag-force parameter $c$  depends on the size of the particles and the viscosity of the fluid. This connection between the Wasserstein distance and viscous dissipation is described in depth in the lecture notes~\cite[Ch.~5]{PeletierVarMod14TR}. 
%In the rest of this paper we disregard physical dimensions, however. 

For the functional $\hFreeEnergy$ one can formally define a `Wasserstein gradient' $\grad_W \hFreeEnergy(\rho)$ for each~$\rho$ as a real-valued function on $\R$ given by
\begin{equation}
\label{def:gradW}
\grad_W \hFreeEnergy(\rho)(y) := -\partial_y \bigl[ \rho \partial_y \xi\bigr](y), \qquad 
\xi := \frac{\delta \hFreeEnergy}{\delta\rho} = \frac12 \log \frac{\rho}{1-\alpha\rho} + \frac{1}{2(1-\alpha\rho)} + V + W*\rho.
\end{equation}
Therefore~\eqref{eq:HL} can be rewritten abstractly as the Wasserstein gradient flow of $\hFreeEnergy$, 
\begin{equation}
\label{eq:GF}
\partial_t \rho =  -\grad_W\hFreeEnergy(\rho).
\end{equation}

\medskip

In this context there are two natural solution concepts for equation~\eqref{eq:HL}. The first is the more classical, distributional defintion.
\begin{definition}[Distributional solutions of~\eqref{eq:HL}]
\label{def:distr-solutions}
A Lebesgue measurable  function $\rho:[0,T]\to \P(\R)$ is a \emph{distributional} solution of~\eqref{eq:HL} if it is a solution in the sense of distributions on $(0,T)\times \R$ of the (slightly rewritten) equation
\begin{equation}
\label{eq:HL-distributional}
\partial_t \rho = \frac1{2\alpha} \partial_{yy}\Bigl(\frac 1{1-\alpha\rho} \Bigr)
  + \partial_y \Bigl [ \rho\partial_y \bigl( V + W*\rho\bigr)\Bigr].
\end{equation}
\end{definition}

We also use a second solution concept that is more adapted to the gradient-flow structure. The monograph~\cite{AmbrosioGigliSavare08} formulates  a number of alternative solution concepts for the general idea of a `metric-space gradient flow'; in this paper we focus on the following one, called Curve of Maximal Slope in~\cite{AmbrosioGigliSavare08} and Energy-Dissipation Principle in~\cite{Mielke16a}, and attributed originally to De Giorgi~\cite{DeGiorgiMarinoTosques80}. 
%We formulate it for the metric space $(\P_2(\R),W_2)$.

\begin{definition}[Gradient-flow solutions in the Energy-Dissipation Principle formulation]
\label{def:metric-GFs}
A curve $\rho\in AC^2([0,T];\P_2(\R))$ with $\hFreeEnergy(\rho(0))<\infty$ is called a \emph{solution of the gradient flow} of~$\hFreeEnergy$ if for all $t\in[0,T]$,
\begin{equation}
\label{eq:EDP}
0= \hFreeEnergy(\rho(t)) - \hFreeEnergy(\rho(0)) + \frac12 \int_0^t \Bigl[ |\dot \rho|^2(s) + |\partial\hFreeEnergy|^2(\rho(s))\Bigr]\, ds.
\end{equation}
Here
\begin{itemize}
\item $AC^2([0,T];\P_2(\R))$ is the space of absolutely continuous functions $\rho:[0,T]\to \P_2(\R)$ (see Definition~\ref{def:AC2});
\item The metric derivative $|\dot \rho|$ of a curve $\rho\in AC^2([0,T];\P_2(\R))$ is defined as
\begin{equation}
\label{def:metric-derivative}
|\dot \rho|(t):= \lim_{h\to0} \frac {W_2(\rho(t+h),\rho(t))}h, \qquad \text{for } 0<t<T;
\end{equation}
\item The  local slope is 
\begin{equation}
\label{def:local-slope-abstract}
|\partial \hFreeEnergy|(\rho) := \limsup_{\nu\to \rho} \frac{(\hFreeEnergy(\rho)-\hFreeEnergy(\nu))_+}{W_2(\rho,\nu)}.
\end{equation}
\end{itemize}
\end{definition}

One can calculate that the metric velocity $|\dot \rho|(t)$ and the metric slope $|\partial\hFreeEnergy|(\rho)$ are formally given by the expressions (see Sections~\ref{ss:Wasserstein} and~\ref{s:Wasserstein-functionals}):
\begin{align*}
|\dot \rho|^2(t) &:= \int_\R v^2(t,y)\, \rho(t,dy), \qquad 
  v(t,y) := -\frac1{\rho(t,d y)}\int_{-\infty}^y \partial_t \rho(t,d\tilde y),\\
|\partial\hFreeEnergy|^2(\rho) &:=  \int_\R |\partial_y \xi(y)|^2 \rho(dy), \qquad \text{with $\xi$ given in~\eqref{def:gradW}.}
\end{align*}
Each gradient-flow solution also is a distributional solution, and for given initial datum $\rho(0)$ gradient-flow solutions are unique (see Lemma~\ref{l:GF-sol-is-distr-sol}).

\medskip
The definition~\eqref{eq:EDP} is inspired by the smooth Hilbert-space case, in which $|\cdot|$ is a Hilbert norm, and an expression of the form of equation~\eqref{eq:EDP} for a curve $x$ and a functional $\Phi$ in Hilbert space can be rewritten as
\begin{equation}
\label{eq:EDP2}
0=  \frac12 \int_0^t  |\dot x(s) + D\Phi(x(s))|^2\, ds,
\end{equation}
where $D\Phi(x)$ is the Hilbert gradient (Riesz representative) of $\Phi$ at $x$. The right-hand side in~\eqref{eq:EDP2} is non-negative, its minimal value is zero, and this minimal value is achieved when $\dot x(t) = -D \Phi(x(t))$ for almost all $t$. For metric spaces, under some  conditions, the same is true for~\eqref{eq:EDP}: the right-hand side is non-negative, its minimal value is zero, and this minimal value is achieved in exactly one curve $x$ among all curves $\tilde x$ with given initial datum $\tilde x(0)$. In other words, {equations} such as~\eqref{eq:EDP} and~\eqref{eq:EDP2} both define a flow in the state space, and this is what we call a `gradient flow' in this paper.

In recent years it has become clear that expressions of the type of~\eqref{eq:EDP} and~\eqref{eq:EDP2} arise naturally as large-deviation rate functions associated with stochastic processes, typically in a many-particle limit; we describe this in detail below for the system of this paper, and the general scheme can be found in~\cite[Prop.~3.7]{MielkePeletierRenger14}. Through such connections, gradient-flow structures of various partial-differential equations can be understood as a natural consequence of the upscaling from a more microscopic system of which the PDE is a scaling limit~\cite{AdamsDirrPeletierZimmer11,AdamsDirrPeletierZimmer13,DuongPeletierZimmer13,DuongLaschosRenger13,MielkePeletierRenger14,PeletierRedigVafayi14,ErbarMaasRenger15}. In addition, this connection provides a natural way to derive and understand new gradient-flow structures for equations in the long term. In this paper we use this method to investigate the gradient-flow structure that arises in this simple one-dimensional, hard-rod system.

\subsection{Main result II: Large deviations of the stochastic evolutions}

The second main theorem of this paper then describes the large-deviation behaviour of the empirical measures $\rho_n(t) = \frac1n \sum_{i=1}^n \delta_{Y^n_i(t)}$ as functions of time, i.e.\ in the state space $C\bigl([0,T];\P(\R)\bigr)$. 

The choice of initial data for the process $Y_i^n$ requires some care. From the point of view of  equation~\eqref{eq:HL} we would like to fix a measure $\rho^\circ\in \P(\R)$ and then select initial data $Y^n_i(0)$ such that the empirical measures $\frac1n\sum_{i=1}^n \delta_{Y^n_i(0)}$ converge to $\rho^\circ$ as $n\to\infty$.

However, not all $\rho^\circ\in \P(\R)$ are admissible, since initial data for~\eqref{eq:HL} should have Lebesgue density bounded by $1/\alpha$. This is a natural consequence of the fact that each particle occupies a section of length $\alpha/n$, and it is also visible in the degeneration of the denominators in~\eqref{eq:HL} and~\eqref{def:F}.

Given some $\rho^\circ$ satisfying this restriction, one might try to draw initial data $Y_i^n(0)$ i.i.d.\ from~$\rho^\circ$, since then with probability one we have $ n^{-1}\sum_{i=1}^n \delta_{Y_i^n(0)}\weakto \rho^\circ$. This is still problematic, since the strong interaction between the rods implies that the initial data for $Y_i^n$ can never be chosen independently. Instead, in the theorem below, we choose initial data for the~$Y^n_i$ by modifying a version of the invariant measure $\InvMeas_n$ instead.  

\medskip
Let $f\in C_b(\R)$, and define the tilted, `$W=0$' invariant measure $\InvMeasf_n\in\P(\R)$ by
\begin{equation}
\label{def:InvMeasf}
\InvMeasf_n(dy) := \frac1{\mathcal Z_n^f} \exp \biggl[\,-\sum_{i=1}^n f(y_i) - 2 \sum_{i=1}^n V(y_i) \biggr] \, \Lebesgue^n\Big|_{\Omega_n}(dy).
\end{equation}
Under this measure the particles are i.i.d.
Also define the 
tilted free energy
\begin{equation}
\label{def:Ff}
\hFreeEnergy^f(\rho) := \begin{cases}
\ds\frac12 \int_{\R} \rho \log \rlap{$\ds\frac{\rho}{1-\alpha \rho} + \int_{\R}\Bigl[\frac12 f +V\Bigr]\rho \; + \;  C_f$}\\[4\jot]
\qquad\qquad&\text{if $\rho$ is Lebesgue-absolutely-continuous and $\rho(y)<1/\alpha$ a.e.},\\[2\jot]
+\infty &\text{otherwise}
\end{cases}
\end{equation}
where the constant $C_f$ is chosen such that $\inf \hFreeEnergy^f=0$. The functional $\hFreeEnergy^f$ is strictly convex and coercive, and we write $\rho^{\circ,f}$ for the unique minimizer of $\hFreeEnergy^f$:
\[
\rho^{\circ,f} = \argmin\limits_{ \P(\R)}\ \hFreeEnergy^f.
\]
(It is not hard to verify that any $\rho$ can be written this way, provided it satisfies $\rho<1/\alpha$ a.e.\ and $\log (\rho/e^{-2V})\in C_b(\R)$.)

\begin{theorem}[Large-deviation principle on path space]
\label{th:LDP-path}
Assume that $V,W$ satisfy Assumption~\ref{ass:VW}. 
For each $n$, let the particle system $t\mapsto Y^n(t)\in\R^n$ be given by~\eqref{eq:SDE}, with initial positions drawn from the tilted invariant measure $\InvMeasf_n$.

The random evolving empirical measures $\rho_n(t) =\frac1n\sum_{i=1}^n \delta_{Y^n_i(t)} $ then satisfy a large-deviation principle on $C\bigl([0,T];\P(\R)\bigr)$ with good rate function $\hat I^f$:
\[
\Prob\Bigl(\rho_n|_{t\in[0,T]} \,\approx \,\nu|_{t\in[0,T]}\Bigr) \sim 
e^{-n\hat I^f(\nu)}\qquad\text{as }n\to\infty.
\]
%. That is,
%\begin{alignat*}2
%\limsup_{n\to\infty} \frac1n \log \Prob(\rho_n \in O) &\geq -\inf_{\nu\in O} \hat I^f(\nu) &\qquad& \text{for all open $O\subset C\bigl([0,T];\P(\R)\bigr)$;}\\
%\liminf_{n\to\infty} \frac1n \log \Prob(\rho_n \in C) &\leq -\inf_{\nu\in C} \hat I^f(\nu) && \text{for all closed $C\subset C\bigl([0,T];\P(\R)\bigr)$.}
%\end{alignat*}

If in addition $\rho$ satisfies  $\rho(0)\in \P_2(\R)$ and  $\hFreeEnergy(\rho(0)) + \hat I^f(\rho) < \infty$, then we have $\rho\in C([0,T];\P_2(\R))$ and $\hat I^f(\rho)$ can be characterized as 
\begin{align}
\hat I^f(\rho) := 2\, &\hFreeEnergy^f(\rho(0)) + \hFreeEnergy( \rho(T))-  \hFreeEnergy(\rho(0))
+\frac12\int_0^T |\dot {\rho}|^2(t)\, dt  
+ \frac12\int_0^T |\partial \hFreeEnergy|^2(\rho(t))\, dt.
\label{def:RF-ind-part-intro}
\end{align}
Here $|\dot {\rho}|$ and $|\partial \hFreeEnergy|$ are the metric derivative and the local slope defined in Definition~\ref{def:metric-GFs}, for the Wasserstein metric space $\mathcal X=(\P_2(\R),W_2)$ (see Section~\ref{ss:Wasserstein}).
\end{theorem}

\subsection{Consequences: the limit equation as a Wasserstein gradient flow}

The large-deviation rate functional $\hat I^f$ in~\eqref{def:RF-ind-part-intro} can be decomposed as 
\[
\hat I^f(\rho) =  2\, \hFreeEnergy^f(\rho(0)) 
+ \EDP[\hFreeEnergy,W_2](\rho),
\]
where $\EDP[\hFreeEnergy,W_2](\rho)$ is shorthand for the right-hand side in the gradient-flow definition in~\eqref{eq:EDP}, with driving functional $\hFreeEnergy$ and dissipation metric $W_2$.  
Both terms are non-negative, and they represent different aspects of the large-deviation behaviour of the sequence of particle systems~$Y^n$.

The first term, $2 \hFreeEnergy^f(\rho(0)) $, characterizes the probability of deviations of the initial empirical measure $\rho_n(0) = \frac1n \sum_{i=1}^n \delta_{Y^n_i(0)}$ from the minimizer $\rho^{\circ,f}$ of $\hFreeEnergy^f$. The second term $\EDP[\hFreeEnergy,W_2](\rho)$ measures deviations of the time course $t\mapsto \rho_n(t)$ from `being a solution of the gradient flow~\eqref{eq:HL}' (or~\eqref{eq:GF}). For minimizers both terms are zero, implying the following

\begin{corollary}
Minimizers $\rho$ of the rate function $\hat I^f$ are solutions of the Wasserstein gradient flow equation~\eqref{eq:HL} (in the gradient-flow sense), with initial datum $\rho(0) = \rho^{\circ,f}$. Therefore minimizers of $\hat I^f$ are unique.
\end{corollary}

Minimizers of $\hat I^f$ describe the typical behaviour of empirical measures $\rho_n$, by the Borel-Cantelli argument that was already mentioned above:
\begin{corollary}
\label{cor:convergence}
The curve of empirical measures $t\mapsto \rho_{n}(t)$ converges almost surely in $C([0,T];\P(\R))$  to a (unique) solution $\rho$ of~\eqref{eq:HL} with initial datum $\rho(0)=\rho^{\circ,f}$.
\end{corollary}

Although the $\lambda$-convexity of $\hFreeEnergy$ already guarantees existence of gradient-flow solutions by~\cite{AmbrosioGigliSavare08}, Corollary~\ref{cor:convergence} trivially gives the same:
\begin{corollary}
\label{cor:existence-of-solution}
Equation~\eqref{eq:HL} with initial datum {$\rho^{\circ,f}$} has a gradient-flow solution.
\end{corollary}

\subsection{Ingredients of  the proofs}

% Ingredients of the proof:
% * compression map, isometry in W_2
% * generalized Varadhan Lemma by Hoeksema et al
% * contraction principle to get from dynamic to static
% * Sanovs theorem

%The proof of the two large-deviation theorems combines a number of ingredients. 
As in many proofs of large-deviation principles, the core of the argument is Sanov's theorem, which provides a large-deviation principle for independent particles. 

In the case of this paper, however, the particles are not only correlated, but the hard-core interaction is a very strong one.
A central step in the proof is to replace this strong interaction by a weaker one. 
This step is done by the second main ingredient, a mapping from the hard-rod particle system $Y^n$ to a system of weakly-interacting zero-length particles called $X^n$. This map appears to have been known at least to Lebowitz and Percus~\cite{LebowitzPercus67} and was used to prove the many-particle limit by Rost~\cite{Rost84} and later by Bodnar and Velazguez~\cite{BodnarVelazquez05}.

\begin{figure}[ht]
\centering
\labellist
\pinlabel $y_1$ [b] at 24 70
\pinlabel $y_2$ [b] at 56 70
\pinlabel $y_3$ [b] at 78 70
\pinlabel $y_4$ [b] at 130 70
\pinlabel $x_1$ [t] at 24 1
\pinlabel $x_2$ [t] at 41 1
\pinlabel $x_3$ [t] at 53 1
\pinlabel $x_4$ [t] at 90 1
\pinlabel $\spY$ [l] at 183 63
\pinlabel $\spX$ [l] at 183 7
\endlabellist
\includegraphics[width=7cm]{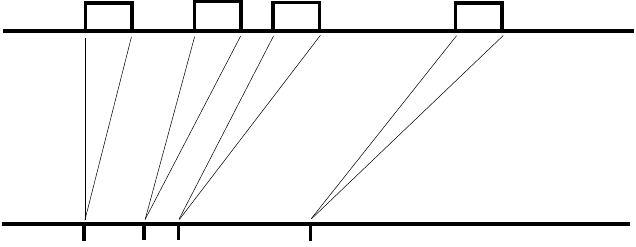}
\caption{The discrete compression and expansion maps (see Section~\ref{ss:compression_map}). Particles of length $\alpha/n$ at positions $y_i$ are mapped to zero-length particles at positions $x_i$ (the compression map $A_n^{-1}$) and vice versa (the expansion map~$A_n$). }
\label{fig:compression-map}
\end{figure}
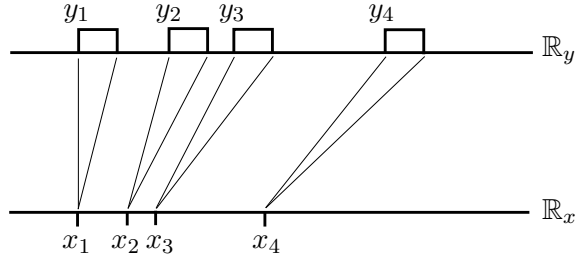

The idea behind this mapping is to map the original collection of rods of length $\alpha/n$  to a collection of zero-length particles by `collapsing' or `compressing' them to zero length and moving  the rods on the right-hand side up towards the left (see Figure~\ref{fig:compression-map}). Two particles $Y^n_i$ and $Y^n_{i+1}$ that collide at some time $t_0$ are mapped by this transformation to two particles $X^n_i$ and $X^n_{i+1}$ that occupy the same point $x$ at time $t_0$. While the compressed particles $X^n_i$ and $X^n_{i+1}$ remain ordered for all time ($X^n_i(t)\leq X^n_{i+1}(t)$), the distribution of the empirical measures remains the same if the two particles {are} allowed to pass each other instead. Mapping the length of the particles to zero therefore allows us to remove the non-passing restriction, and by removing this restriction we eliminate the strong interaction between particles. 

The price to pay is that after transformation the effects of the on-site potential $V$ and the interaction potential $W$ come to depend on the whole particle system. This happens because the amount that particle $X^n_i$ should be considered `shifted to the right' is equal to $\alpha/n$ times the number of particles $X^n_j$ that are---at that moment---to the left of $X^n_i$, and that therefore the force exerted by the on-site potential $V$ (for instance) is equal to 
\[
-V'\biggl( X^n_i(t) + \frac\alpha n\#\Bigl\{j\in 1,\dots, n: X^n_j(t)< X^n_i(t)\Bigr\}\biggr).
\]
This  force on the particle $X^n_i$  depends in a discontinuous manner on the positions of all particles. Had this force been smooth, a standard application of Varadhan's Lemma would   convert Sanov's theorem into a large-deviation principle for the particle system $X^n$, as done by e.g. Dai Pra and Den Hollander (see~\cite{Dai-PraHollander96} or~\cite[Ch.~X]{DenHollander00}).
Since it is not smooth, however, we use a recent result by Hoeksema, Maurelli, Holding, and Tse~\cite{HoeksemaMaurelliHoldingTse20TR}, that generalizes Varadhan's Lemma to mildly singular and discontinuous forcings (Theorem~\ref{th:JasperMarioOliver} below).

Finally, a fortuitous property of the expansion and compression maps is that they are isometries for the Wasserstein metric. This implies that the metric structure of the large-deviation rate functional $\hat I^f$---in terms  of the metric velocity $|\dot\rho|$ and the metric slope $|\partial\hFreeEnergy|$---transforms transparently from the $X^n$ to the $Y^n$ particle system.

\subsection{Conclusion and discussion, part I: Mathematics}
\label{sec:discussion-I}

We have proved a large-deviation principle on path space for a one-dimensional system of hard rods, in the many-particle limit. This large-deviation principle characterizes the entropy of the system as a function of the density, and identifies the limit evolution as a Wasserstein gradient flow of the entropy.

\medskip
From a mathematical point of view, this result can be  interpreted in different ways:
\begin{enumerate}
\item It rigorously establishes equation~\eqref{eq:HL} as the continuum limit of the particle system, in the sense that the empirical measures $\rho_n$ converge to a solution of~\eqref{eq:HL}. While this result was proved for the case $V=W=0$ by Rost in~\cite{Rost84}, it is new for the case of non-zero $V$ and $W$.
\item In addition, it establishes the functional $\hFreeEnergy$ as the driving functional and the metric $W_2$ as the dissipation of the gradient-flow structure for equation~\eqref{eq:HL}. This result is  new, also for the case $V=W=0$. 
\end{enumerate}

\subsubsection*{The difference between  $W_2$- and narrow topology.} 
Hidden in the notation of the two large-deviation theorems  is a subtlety concerning  topology. The $W_2$-topology is central to the gradient-flow structure, and we argue here that this structure arises from the large deviations. On the other hand, the two large-deviation principles themselves are proved in the narrow topology on $\P(\R)$, which is weaker. 

The large-deviation theorems themselves probably do not hold in the stronger $W_2$-topology.  For independent particles this can be recognized in the characterization of the validity of Sanov's theorem in Wasserstein metric spaces by Wang, Wang, and Wu~\cite{WangWangWu10}. These authors show that Sanov's theorem is invalid without exponential moments on the underlying distribution, and this condition is much stronger than the first-moment condition induced by $V$ in the case of this paper. 

This begs the question how the $W_2$-topology is generated by the large-deviation rate function while not being part of the large-deviation principle. The answer is that if $I^f(\rho)$ is finite and if the  initial datum $\rho(0)$ is in $\P_2(\R)$, then $\rho(t)\in \P_2(\R)$ for all time $t$; this is shown in Lemma~\ref{lemma:I-to-v-formulation}. However,  $\rho(t)\in \P_2(\R)$ is a much weaker property than finiteness of exponential moments of $\rho(t)$, which is necessary for exponential tightness in $W_2$ of the underlying particle system.

\subsection{Conclusion and discussion, part II: Consequences for modelling}

This large-deviation result also gives rise to a rigorous \emph{Variational-Modelling} derivation of the limit equation~\eqref{eq:HL}. It  explains and motivates the choice of the modified entropy $\hFreeEnergy$ as the driving functional and the Wasserstein distance as the dissipation.

The appearance of  the driving functional $\hFreeEnergy$ is  expected. The first integral in $\hFreeEnergy$ arises as a measure of `free space' after taking into account the finite length of the particles; this becomes apparent in the discussion of the `compression' map in Section~\ref{sec:part-systems}. The second and third integrals are relatively standard contributions from on-site and interaction potentials.

On the other hand, the appearance of the Wasserstein distance as the dissipation metric is unexpected. This is the same metric as for non-interacting particles~\cite{DawsonGartner87,KipnisOlla90,AdamsDirrPeletierZimmer11,PeletierVarMod14TR}, and the result therefore shows that incorporating steric interactions does not change the dissipation metric, a fact that is surprising at first glance.

This fact can be understood from the proof, however. It is related to the property that the compression and expansion maps are isometries for the Wasserstein distance. The central observation is the following: the total travel distance between a set of initial points $y_1,\dots,y_n$ and final points $\overline y_1,\dots,\overline y_n$ is the same as the total travel distance between the corresponding compressed set of  initial points $x_1,\dots,x_n$ and final points $\overline x_1,\dots,\overline x_n$. This is true because in the minimization problem~\eqref{def:Wasserstein-intro} the optimal permutation of the particles is such that particles preserve their ordering, and therefore the compression mapping moves the points $y_{\sigma(i)}$ and $\overline y_i$ to the left by the same amount.
 
This result therefore is intrinsically limited to the one-dimensional setup of this paper. In higher dimensions there is no such compression map, but one can still wonder whether the dissipation of particles with finite and with zero size might be both respresented by the Wasserstein distance. This appears not to be the case: we illustrate this in Figure~\ref{fig:mobility}. In addition, in the case of multiple species the metric can certainly not be Wasserstein, since particles moving in opposite directions will be forced to move around each other. 
\begin{figure}[ht]
\includegraphics[width=0.8\hsize]{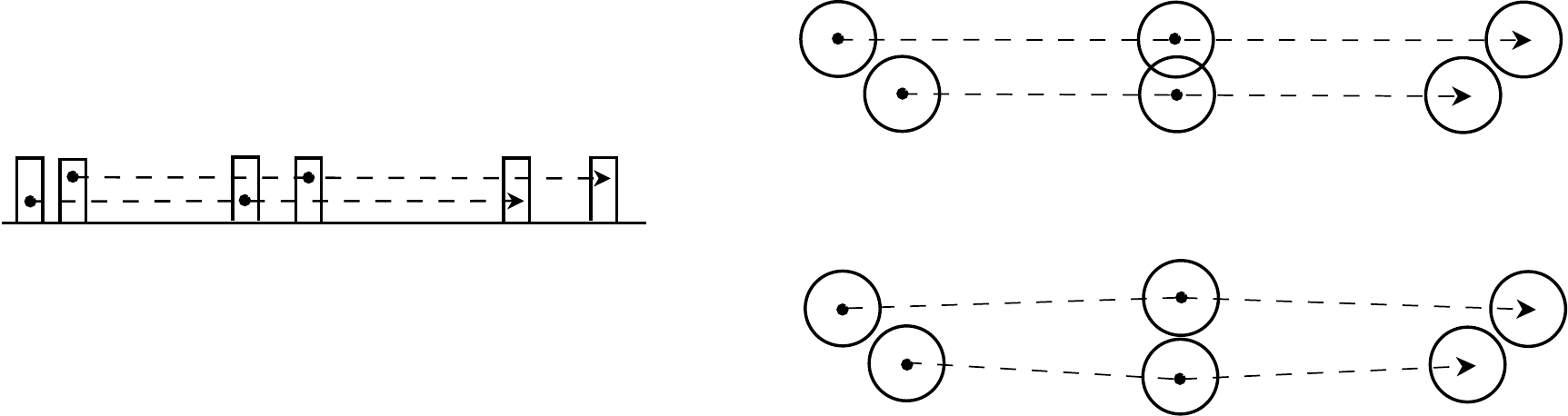}
\caption{In higher dimensions the metric will not be Wasserstein.
In one dimension (left), linear interpolation of particle positions preserves admissibility: if the initial and final positions do not overlap, then the intermediary positions also do not overlap. In higher dimensions, this is false: two spheres arranged in admissible configurations may collide under linear interpolation (top right). We expect that the metric in higher dimensions therefore will be non-Wasserstein, since it will have to accommodate particles `moving around' each other (bottom right). 
}
\label{fig:mobility}
\end{figure}

%The limitation to a setup of a single species of hard rods in one spatial dimension naturally reduces the applicability of the result; at the same time, however, it indicates how the extension of the variational structure to higher dimensions and to multiple species may be challenging to establish, and probably different from the current setup. We now comment on these and other consequences. 

\subsubsection*{Comparison with Bruna \& Chapman's approximate equation.} 
In a series of publications~\cite{Bruna12TH,BrunaChapman12,BrunaChapman12a,BrunaChapman14}, Bruna and Chapman analyze systems of hard spheres with Brownian noise in the limit of small volume fraction. Their approach is to apply a singular-limit analysis to the Fokker-Planck equation associated with the particles, and this allows them to address this issue in all dimensions and for finite numbers of particles. For the setup of this paper with $W=0$, Bruna finds an approximate equation in the small-$\alpha$ limit~\cite[App.~D]{Bruna12TH}
\begin{equation}
\label{eq:Bruna-Chapman}
\partial_t \rho = \partial_y\Bigl[ \frac12 \partial_y \rho + \alpha \rho \partial_y \rho
 + V'(y)\rho\Bigr] + O(\alpha^2).
\end{equation}
This equation is also found by a Taylor development of the denominator in~\eqref{eq:HL}. Similarly applying a formal Taylor development to~$\hFreeEnergy$ in~\eqref{def:F}, we find  that equation~\eqref{eq:Bruna-Chapman} has a formal `approximate' gradient flow structure
\[
\text{driving functional }\hFreeEnergy_{\mathrm{BC}}(\rho) = \int \Bigl[ \frac12 \rho\log \rho + \alpha\rho^2 + V\rho \Bigr] + O(\alpha^2),
\qquad \text{and metric $W_2$.}
\]
Bruna, Burger, Ranetbauer, and Wolfram study the concept of approximate  gradient-flow structures in more detail in~\cite{BrunaBurgerRanetbauerWolfram17,BrunaBurgerRanetbauerWolfram17TR}.

\subsubsection*{Comparison with Poisson-Nernst-Planck type models with steric effects.}
As described in the introduction, a wide family of generalized Poisson-Nernst-Planck models has been derived by modelling the effect of the finite particle size on the driving functional (free energy) of the system, while assuming that the dissipation mechanism is the same as for systems with point ions. Our work shows that the last assumption is valid for the case of a single species of hard rods in one spatial dimension, and is therefore consistent with the current literature.  

As illustrated above, however, in the case of multiple species in higher dimensions a form of cross-diffusion is to be expected. We present an example of such a system for charged particles in~\cite{gavish2017solvent}.  Here the mobility matrix is nonlinear and degenerate, in that transport of particles of species A  to a region diminishes with increasing concentration of that species in that region.  The mobility matrix is also non-diagonal, reflecting inter-diffusion, i.e., the movement of an ionic species must involve counter movement of water and other ionic species.  (This also is observed in limits of lattice models with exclusion; see e.g.~\cite{BurgerSchlakeWolfram12}). Furthermore, while in classical Poisson-Nernst-Planck theory, the diffusivity of the ions is proportional to their concentration, the modified equation show a super-linear increase of diffusivity with ionic concentration.  This increase reflects the solvent tendency to diffuse to the regions of high ionic concentration and  may be a significant effect since the entropy per volume of many small particles is larger than the entropy of a fewer larger particles and the solvent molecules are typically significantly smaller than the ions. This work should be considered a step towards the study of such systems.
%
%{\revNir The prediction of these theories, however, deviate from empirical observations for systems with rather low ionic concentrations, see the text of classical textbook due to Barthel. et al., [33] who say on p. 325, with slight paraphrase, 
%\cite{barthel1998physical} {\em ``Theories with point ions are restricted to such low concentrations that their experimental verification often proves to be an unsolvable task.''}.  Accordingly, extensive effort has been invested during the past centennial in extending these theories to account for the finite-size of the ions, see~\cite{bazant2009towards} and references within for review.}
%

\subsection{Overview of the paper}

In Section~\ref{s:prelims} we introduce the Wasserstein distance, Wasserstein gradient flows,  and inverse cumulative distribution functions, which play a central role in the analysis. In Section~\ref{s:ldp} we introduce large-deviation principles. In Section~\ref{sec:part-systems} we formally define the systems that we study and the compression and expansion maps that we mentioned above. In Section~\ref{s:functionals} we formally define various functionals that appear in the analysis, and prove a number of properties. In Sections~\ref{s:ldp-iid}, \ref{s:ldp-Y-special}, and~\ref{s:proof-dynamic-ldp} we prove Theorems~\ref{th:LDP-invmeas} and~\ref{th:LDP-path} in three stages, while Section~\ref{s:ldp-estimates} is devoted to a number of estimates used in Section~\ref{s:ldp-Y-special}.

\subsection{Notation}

We sometimes write $\spX$ and $\spY$ to distinguish state spaces for particle systems of `compressed' particles (usually called $X^n$, sometimes $Z^n$) and `expanded' particles~$Y^n$.
For measures $\mu$ on $\R$ we  write $\mu(x)$ for the Lebesgue density and $\mu(dx)$ for the measure inside an integral. For time-dependent measures $\mu(t,dx)$ we write both $\mu(t)$ and $\mu_t$ for the measure $\mu(t,\cdot)$, and correspondingly $\mu_0$ and $\mu(0)$ both indicate the measure $\mu(0,\cdot)$.
%\medskip
%\noindent Other notation:

\begin{center}
\newcommand{\specialcell}[2][c]{%
  \begin{tabular}[#1]{@{}l@{}}#2\end{tabular}}
%\begin{small}
\begin{longtable}{lll}
$|\dot \rho|(t)$ & Metric derivative &\eqref{def:metric-derivative}, \eqref{eq:metric-derivative-v}\\
$\|\cdot\|_{BL}$ & Bounded-Lipschitz norm on continuous and bounded functions &Sec.~\ref{ss:narrow-convergence}\\
$\alpha$ & Rods have length $\alpha/n$ \\
$ A$, $A_n$ & Expansion maps &Def.~\ref{def:AAn}\\
$\gamma(\cdot)$ & Correction term in entropy & \eqref{def:gamma}\\
$d_{BL}$ & Dual bounded-Lipschitz metric on $\P(\R)$ &Sec.~\ref{ss:narrow-convergence}\\
$\Ent_V$, $\hEnt_V$
 & Entropies  in compressed and expanded coordinates & Sec.~\ref{s:functionals}\\
$\Energy_W$, $\hEnergy_W$
 & Interaction energies in compressed and expanded coordinates & Sec.~\ref{s:functionals}\\
$\FreeEnergy$, $\hFreeEnergy$ & Free energies in compressed and expanded coordinates & \eqref{def:F}, Sec.~\ref{s:functionals}\\
$\hFreeEnergy^f$ & Tilted free energy  & \eqref{def:Ff}\\
$\partial F(\rho)$ & Fr\'echet subdifferential of $F$ & Def.~\ref{def:subdiff}\\
$\partial^\circ F(\rho)$ & Element of $\partial F(\rho)$ of minimal norm & Def.~\ref{def:subdiff}\\
$|\partial F|(\rho)$ & Metric slope of $F$ &\eqref{def:local-slope-abstract}\\
$\eta_n(\cdot)$ & Empirical measure map & \eqref{def:eta_n}\\
$\RelEnt(\,\cdot\,|\,\cdot\,)$ & Relative entropy & \eqref{e:def:RelEnt}\\
icdf & Inverse cumulative distribution function & Def.~\ref{def:icdf}\\
$I^f$ & Rate functional for pathwise large-deviation principle &\eqref{def:RF-ind-part-intro}\\
$\mathfrak I_\xi$ & Dynamic rate function for i.i.d. initial data &\eqref{def:I-time-dependent}\\
$\mathcal L_X, \mathcal L_Y$ & Generators for $X^n$ and $Y^n$ stochastic particle systems & Def.~\ref{def:particle-systems}\\
$\Omega_n$ & State space for particle system $Y^n$ & \eqref{def:Omega-n}\\
$\P(\R)$ & Probability measures on $\R$, with metric $d_{BL}$& Sec.~\ref{ss:narrow-convergence}\\
$\P_2(\R)$ & Probability measures  with finite second moments and $W_2$-metric &Sec.~\ref{ss:Wasserstein}\\
$\P^n(\R)$ & Empirical measures of $n$ points on $\R$ & \eqref{def:Pn} \\
$\InvMeas_n$ & Invariant measure for $Y^n$ & \eqref{def:InvMeasn}\\
$\InvMeasf_n$ & Tilted, $W=0$ invariant measure for $Y^n$ & \eqref{def:InvMeasf}\\
$\bbQ^\nu$ & Single-particle tilted measure on $\spX$ & \eqref{def:Qnu}\\
$\bm t_\mu^\nu$ & Transport map from $\mu$ to $\nu$ &Lemma~\ref{l:props-W2}\\
$T_\mu$ & Auxiliary expansion map &Lemma~\ref{lemma:A-is-pushforward}\\
$V$ & On-site  potential & Ass.~\ref{ass:VW}\\
$W$ & Interaction potential & Ass.~\ref{ass:VW}\\
$W_2$ & Wasserstein metric of order $2$ & Sec.~\ref{ss:Wasserstein}\\
$\mathcal Z_n$ & Normalization constant for $\InvMeas_n$ & \eqref{def:norm-const-Zn-intro}
\end{longtable}
%\end{small}
\end{center}

\section{Measures and the Wasserstein metric}
\label{s:prelims}

The Wasserstein gradient of a functional $\hFreeEnergy$, and the corresponding gradient flow, was informally defined in Section~\ref{subsec:var-structure}. There is an extensive literature on the Wasserstein metric and its properties~\cite{Villani03,AmbrosioGigliSavare08,Villani09,Santambrogio15}, but for the discussion of this paper we only need a number of facts, which we summarize in this section.

\subsection{Preliminaries on one-dimensional measures}

The concept of push-forward will be used throughout this work:
\begin{definition}[Push-forwards]
Let $f:\R\to\R$ be Borel measurable, and $\mu\in\P(\R)$. The \emph{push-forward} $f_\#\mu\in \P(\R)$ is the measure $\mu\circ f^{-1}$, and has the equivalent characterization
\[
\int_\R \varphi(y)\,(f_\#\mu)(dy) = \int_\R \varphi(f(x))\,\mu(dx),
\qquad\text{for all Borel measurable $\varphi:\R\to\R$.}
\]
\end{definition}

The Wasserstein distance and the energy functionals in this paper have convenient representations in terms of inverse cumulative distribution functions.
\begin{definition}[Inverse cumulative distribution functions]
\label{def:icdf}
Let $\mu\in \P(\R)$. Let $F:\R\to[0,1]$ be the right-continuous cumulative distribution function:
\[
F(x) := \mu((-\infty,x]).
\]
Then the \emph{inverse cumulative distribution function} $\X$ of $\mu$ is the generalized (right-continuous) inverse of $F$,
\[
\X(m) := \inf\{x\in \R: F(x)> m\}.
\]
\end{definition}

The following lemma collects some well-known properties of inverse cumulative distribution functions.
\begin{lemma}
\label{lemma:icdf-transformation}
Let $\mu\in \P(\R)$ and let $\X$ be the inverse cumulative distribution function of $\mu$.
\begin{enumerate}
\item $\X$ is non-decreasing and right-continuous;
\item If $\mu$ is absolutely continuous, then $F(\X(m)) = m$, $\X'(m)$ exists for Lebesgue-almost-every $m\in[0,1]$, and for those $m$ we have 
\[
\X'(m) = 
 1/\mu(\X(m));
\]
\item 
For all Borel measurable $\varphi:\R\to\R$ we have
\begin{equation}
\label{eq:transformation-X}
\int_\R \varphi(x) \,\mu(dx) = \int_0^1 \varphi(\X(m))\, dm.
\end{equation}
\end{enumerate}
\end{lemma}

\begin{proof}
The function $\X$ is obviously non-decreasing, and the right-continuity is a direct consequence of the definition. To characterize  $\X'(m)$ for absolutely-continuous $\mu$, first note that $F$ then is an absolutely continuous function; by~\cite[Prop.~1]{Winter97} we then have $F(\X(m)) = m$ for all $m\in [0,1]$. Since $\X$ is monotonic, it is differentiable at almost all $m\in [0,1]$. Let $M$ be the set of such $m$; then for each $m\in M$, 
\[
\X'(m) = \lim_{\tilde m\to m} \frac{\X(\tilde m)-\X(m)}{\tilde m-m}
= \lim_{\tilde m\to m} \frac{\X(\tilde m)-\X(m)}{F(\X(\tilde m))-F(\X(m))}.
\]
First, assume that $\X'(m) = 0$. The identity above then implies that  $F$ is not differentiable at $x=\X(m)$; the set $\mathcal X$ of such $x$ is a Lebesgue null set of $\R$, and since the function $F$ has the `Lusin N' property~\cite[Def.~9.9.1]{Bogachev07.II} the corresponding set of values $F(\mathcal X)$ has Lebesgue measure zero as well. For all $m$ in the full-measure set $M\setminus F(\mathcal X)$ we therefore have that $\X'(m)$ exists and is non-zero, and by the calculation above  $\X'(m) =1/F'(\X(m)) =  1/\mu(\X(m))$.

Finally, the transformation rule~\eqref{eq:transformation-X} is proved in~\cite[Th.~2]{Winter97}.
\end{proof}

\subsection{Narrow topology and the dual bounded-Lipschitz metric}
\label{ss:narrow-convergence}

We will be using two topologies  on spaces of probability measures. The first type is the \emph{narrow} topology, often called the \emph{weak topology of measures}, which can be defined in various ways. For the purposes of this paper it is convenient to introduce it through the set $BL(\R)$ of \emph{bounded Lipschitz functions} on $\R$, with norm
\[
\|f\|_{BL} := \|f\|_\infty + \Lip(f), \qquad \Lip(f) := \sup_{x,y\in \R} \frac{|f(x)-f(y)|}{|x-y|}.
\]
The narrow convergence on $\P(\R)$ is metricised by duality with the set of bounded Lipschitz functions, leading to the \emph{dual bounded-Lipschitz metric}
\[
d_{BL}(\mu,\nu) := \sup_{\|f\|_{BL}\leq 1} \int_\R f(x)\bigl[\mu(dx)-\nu(dx)\bigr].
\]
Alternative ways of defining the same topology are by the L\'evy metric or through duality with continuous and bounded functions~\cite{Rachev91}. When we write $\P(\R)$, we implicitly equip the space with the dual bounded-Lipschitz metric $d_{BL}$.

\subsection{The Wasserstein metric}
\label{ss:Wasserstein}

We write $\P_2(\R)$ for the space of probability measures with finite second moments,
\[
\P_2(\R) := \biggl\{ \mu\in \P(\R): \int_\R x^2 \mu(dx) < \infty\biggr\}.
\]

\begin{definition}[Wasserstein distance]
The \emph{Wasserstein distance of order 2} between measures $\mu$ and $\nu$ in {$\P_2(\R)$} is defined by 
\begin{equation}
\label{def:W2}
W_2( \mu,\nu)^2 := \inf\Bigl\{ \int_\R |x-x'|^2 \, \gamma(dxdx'): \gamma\in \Gamma(\mu,\nu)\Bigr\},
\end{equation}
where $\Gamma(\mu,\nu)$ is the set of couplings (`transport plans') of $\mu$ and $\nu$, i.e.\ of measures $\gamma\in \P(\R\times\R)$ such that 
\[
\gamma(A\times \R) = \mu(A), \quad \gamma(\R\times A) = \nu(A), \qquad\text{for all Borel sets }A\subset \R.
\]
\end{definition}

\medskip
\noindent
In this paper we always consider $\P_2(\R)$ to be equipped with the metric $W_2$.

\bigskip

\begin{lemma}[Properties of the Wasserstein metric]
\label{l:props-W2}
\noindent
\begin{enumerate}
\item \label{lem:W2:inf-achieved}
The infimum in~\eqref{def:W2} is achieved and unique. 
\item \label{lem:W2:icdf}
We have the characterization
\begin{equation}
\label{eq:W2-in-icdf}
W_2^2(\mu,\nu) = \int_0^1 |\X_\mu(m)-\X_\nu(m)|^2 \, dm,
\end{equation}
where $\X_\mu$ and $\X_\nu$ are the inverse cumulative distribution functions of $\mu$ and $\nu$.
\item \label{lem:W2:map}
If $\mu$ is Lebesgue-absolutely-continuous, then the minimizer in~\eqref{def:W2} can be written as a transport \emph{map}: $\gamma = (id\times \bm t_\mu^\nu)_\#\mu$ where $\bm t_\mu^\nu: \R\to\R$ pushes forward $\mu$ to $\nu$, i.e. $\nu = (\bm t_\mu^\nu)_\# \mu$. 
In terms of the inverse cumulative distribution functions $\X_\mu$ and $\X_\nu$ of $\mu$ and $\nu$, the map $\bm t_\mu^\nu$ satisfies
\begin{equation}
\label{prop:t-map-icdf}
\bm t_\mu^\nu(\X_\mu(m)) = \X_\nu(m)\qquad\text{for all }m\in [0,1].
\end{equation}

\item \label{lem:W2:bnd-BL}
$d_{BL}(\mu,\nu)\leq W_2(\mu,\nu)$ for all $\mu,\nu\in \P_2(\R)$.
\end{enumerate}
\end{lemma}

\begin{proof}
Part~\ref{lem:W2:inf-achieved} is a consequence of the tightness of $\Gamma(\mu,\nu)$ and the strict convexity of the quadratic function. Parts~\ref{lem:W2:icdf} and~\ref{lem:W2:map} are proved in~\cite[Th.~2.18 and 2.12]{Villani03}.
To prove part~\ref{lem:W2:bnd-BL}, we use the Kantorovich formulation of the Wasserstein distance of order $1$ (e.g.~\cite[Th.~1.14]{Villani03}):
\begin{align*}
d_{BL}(\mu,\nu) = \sup_{\|f\|_\infty + \Lip(f)\leq 1} \int_\R  f\bigl[\mu-\nu\bigr]
&\leq \sup_{ \Lip(f)\leq 1} \int_\R  f\bigl[\mu-\nu\bigr]\\
&= \inf\Bigl\{ \int_\R |x-x'| \, \gamma(dxdx'): \gamma\in \Gamma(\mu,\nu)\Bigr\}\\
&\leq \inf\Bigl\{ \int_\R |x-x'|^2 \, \gamma(dxdx'): \gamma\in \Gamma(\mu,\nu)\Bigr\}^{1/2}\\
&= W_2(\mu,\nu).
\end{align*}

\end{proof}

\bigskip

\begin{definition}[$AC^2$-curves in the $W_2$-metric~{\cite[Ch.~8]{AmbrosioGigliSavare08}}]
\label{def:AC2}
Define the space $AC^2([0,T];\P_2(\R))$ as the space of curves $\mu\in C([0,T];\P_2(\R))$ such that there exists $w\in L^2(0,T)$ with the property
\[
W_2(\mu_s,\mu_t) \leq \int_s^t w(\sigma)\, d\sigma, \qquad \text{for all }0\leq s\leq t \leq T.
\]
\end{definition}

\begin{lemma}[Characterization of $AC^2$-curves in $\P_2$]
\label{lemma:char-AC2}
A curve $\mu\in C([0,T];\P_2(\R))$ is an element of $AC^2([0,T];\P_2(\R))$ if and only if there exists a Borel vector field $v:(x,t)\to v_t(x)$ such that $v_t \in L^2(\mu_t)$ for a.e. $t\in [0,T]$, and 
\begin{equation}
\label{cond:vt-in-L2}
 t\mapsto \|v_t\|_{L^2(\mu_t)}\in L^2(0,T),
\end{equation}
and the continuity equation
\begin{equation}
\label{eq:ct-eq}
\partial_t\mu_t + \partial_x (\mu_t v_t)  = 0
\end{equation}
holds in the sense of distributions on $(0,T)\times \R$. In this case, if two functions $(t,x)\mapsto v(t,x), \tilde v(t,x)$  satisfy~\eqref{cond:vt-in-L2} and~\eqref{eq:ct-eq}, then $\mu_t v_t = \mu_t \tilde v_t$ Lebesgue-almost everywhere in $(0,T)\times \R$,

In addition, if $\mu\in AC^2([0,T];\P_2(\R))$, then the metric derivative $|\dot \mu|$ defined in~\eqref{def:metric-derivative} exists at a.e.\ $t\in [0,T]$, and satisfies
\begin{equation}
\label{eq:metric-derivative-v}
|\dot \mu|^2(t) := \int_\R |v_t|^2 \mu_t, \qquad \text{for a.e. }t\in[0,T].
\end{equation}
\end{lemma}

\begin{proof}
The statement follows directly from~\cite[Th.~8.3.1]{AmbrosioGigliSavare08}; we only need to prove uniqueness of $v$. Assume that there exist two  functions $(t,x)\mapsto v(t,x), \tilde v(t,x)$ as in the Lemma. Then $\partial_x(\mu(v-\tilde v)) = 0$ in $(0,T)\times \R$ in the sense of distributions, and there exists a Borel measurable function $f:[0,T]\to\R$ such that $\mu_t(x)(v_t(x)-\tilde v_t(x)) = f(t)$ for Lebesgue-almost all $(t,x)\in (0,T)\times\R$. 

We then calculate for $a<b$ and $\varphi\in C_b([0,T])$,
\begin{align*}
(b-a)\int_0^T \varphi(t) f(t) \,dt
&= \int_0^T \varphi(t) \int_a^b (v_t(x)-\tilde v_t(x))\mu_t(dx) \,dt\\
&\leq \int_0^T \varphi(t) \biggl\{ \frac12\int_a^b (v_t(x)-\tilde v_t(x))^2\mu_t(dx) 
+  \frac12\int_a^b \mu_t(dx)\biggr\}\,dt\\
&\leq \|\varphi\|_\infty \int_0^T \Bigl\{\|v_t\|_{L^2(\mu_t)}^2+\|\tilde v_t\|_{L^2(\mu_t)}^2 + \frac12\Bigr\}
\, dt.
\end{align*}
Since the right-hand side does not depend on $(b-a)$, we find $\int_0^T\varphi(t)f(t)\, dt = 0$ for all $\varphi\in C_b([0,T])$, and therefore $f=0$. It follows that $\mu v$ and $\mu \tilde v$ are almost everywhere equal. 
\end{proof}

\subsection{Functionals on Wasserstein space}
\label{s:Wasserstein-functionals}

\begin{definition}[$\lambda$-convex functionals; {\cite[Ch.~9]{AmbrosioGigliSavare08}}]
Fix $\lambda\in \R$. The functional $F:\P_2(\R)\to\R\cup\{\infty\}$ is called \emph{$\lambda$-convex} if
\[
F(\mu_t^{1\to2}) \leq (1-t)F(\mu^1) + tF(\mu^2) - \frac\lambda2 t (1-t) W_2^2(\mu^1,\mu^2),
\]
where $\mu^{1\to2}_t$ is the constant-speed geodesic connecting $\mu^1$ to $\mu^2$ (see e.g.~\cite[Sec.~7.2]{AmbrosioGigliSavare08}).
\end{definition}

\begin{definition}[Fr\'echet subdifferentials; {\cite[Def.~10.1.1]{AmbrosioGigliSavare08}}]
\label{def:subdiff}
Let $F:\P_2(\R)\to \R\cup\{\infty\}$, and let  $\mu\in D(F) := \{\mu': F(\mu')<\infty\}$ be Lebesgue-absolutely-continuous. The \emph{Fr\'echet subdifferential} $\partial F(\mu)$ is the set of all $\xi\in L^2(\mu)$ such that
\[
F(\nu)-F(\mu)\geq \int_\R \xi(x) (\bm t_\mu^\nu(x)-x)\, \mu(dx) + o(W_2(\mu,\nu))
\qquad \text{as }\nu\to\mu.
\]
The subdifferential is a closed convex subset of $L^2(\mu)$; if it is non-empty, it therefore admits a unique element $\xi^\circ$ of minimal $L^2(\mu)$-norm. 
We write $\partial^\circ F(\mu) := \xi^\circ$ if this element exists. 
\end{definition}

\begin{lemma}[Subdifferentials and the chain rule; {\cite[Lemma~10.1.5 and Proposition~10.3.18]{AmbrosioGigliSavare08}}]
\begin{enumerate}
\item \label{lem:i:local-slope-minimal-element}
In the context of Definition~\ref{def:subdiff}, if the subdifferential is non-empty, then the local slope~\eqref{def:local-slope-abstract} is finite and satisfies
\begin{equation}
\label{char:metric-slope-subdifferential}
|\partial F|(\mu) = \|\xi^\circ\|_{L^2(\mu)} = \inf \{\|\xi\|_{L^2(\mu)}: \xi\in \partial F(\mu)\}.
\end{equation}
\item \label{lem:i:chain-rule}
The following \emph{chain rule} holds. Let $F:\P_2(\R)\to\R\cup\{\infty\}$ be $\lambda$-convex, and let $\mu\in AC^2([0,T];\P_2(\R))$ be such that 
\begin{enumerate}
\item $\mu_t$  is Lebesgue-absolutely-continuous and $\partial F(\mu_t) \not= \emptyset$ for almost all $t\in[0,T]$;
\item We have
\[
\int_0^T |\dot\mu|(t) |\partial F|(\mu_t)\, dt < \infty.
\]
\end{enumerate}
For any $0\leq s\leq t\leq T$ and any selection $\xi_\sigma\in \partial F(\mu_\sigma)$ we then have 
\begin{equation}
\label{eq:chain-rule}
F(\mu_t)-F(\mu_s) = \int_s^t \int_\R v_\sigma(x) \xi_\sigma(x) \,\mu_\sigma(dx) d\sigma,
\end{equation}
where $v_t$ is the velocity field given by Lemma~\ref{lemma:char-AC2}.
\end{enumerate}
\end{lemma}

\subsection{Wasserstein gradient flows}
\label{ss:Wass-GF}

Recall from Section~\ref{sec:intro-GF-LDP} the definition of `gradient flow' that we use here, applied to the case of the Wasserstein metric space $\P_2(\R)$ and a functional $F:\P_2(\R)\to\R\cup {+\infty}$: A function $\rho\in AC^2([0,T];\P_2(\R))$ is a gradient-flow solution if for all $t>0$, 
\begin{equation}
\label{def:Wass-GFs}
0 = F(\mu_t) -F(\mu_0) + \frac12 \int_0^t \Bigl[ |\dot\mu|^2(s) + |\partial F|^2(\mu_s)\Bigr]\, ds.
\end{equation}
By~\cite[Th.~11.1.3]{AmbrosioGigliSavare08}, if $F$ is proper, lower semicontinuous, and $\lambda$-convex, then solutions in this sense satisfy the pointwise property
\[
v_t = -\partial^\circ F(\rho_t) \qquad\text{for a.e. }t>0.
\]

In the case of the Wasserstein gradient flow of $\hFreeEnergy$, we show in Lemma~\ref{lemma:properties-of-the-functionals} that $\hFreeEnergy$ satisfies these properties, and that $\partial^\circ \hFreeEnergy(\rho)$ is 
$\partial_y \xi$, where the function $\xi$ was already introduced in~\eqref{def:gradW}, 
\[
\xi (\rho) :=\frac12 \log \frac{\alpha\rho}{1-\alpha\rho} + \frac{\alpha\rho}{2(1-\alpha\rho)} + V + W*\rho.
\]
If $\rho$ is a gradient-flow solution with $\hFreeEnergy(\rho_0)<\infty$, then writing~\eqref{def:Wass-GFs} as 
\[
\hFreeEnergy(\rho_t) + \frac12 \int_0^t \Bigl[|\dot\rho|^2(s) + |\partial\hFreeEnergy(\rho_s)|^2 \Bigr] \, dt  = \hFreeEnergy(\rho_0),
\]
it follows that $|\partial\hFreeEnergy(\rho_t)|<\infty$ and therefore $\partial_y \xi(\rho_t)\in L^2(\rho_t)$ for almost all $t$. 
Therefore solutions $\rho$ of the gradient flow of $\hFreeEnergy$ satisfy
\begin{equation}
\label{eq:GF-v-xi}
\partial_t\rho = \partial_y \bigl[\rho \partial_y \xi(\rho)\bigr],
\end{equation}
in the sense of distributions on $(0,T)\times \R$.

\begin{lemma}
\label{l:GF-sol-is-distr-sol}
Gradient-flow solutions of $\hFreeEnergy$ are unique, and a gradient-flow solution also is a distributional solution in the sense of Definition~\ref{def:distr-solutions}. 
%
%, and for given initial datum $\rho(0)$ gradient-flow solutions are unique.
\end{lemma}

\begin{proof}
The $\lambda$-convexity and lower semicontinuity properties of the functional $\hFreeEnergy$ (Lemma~\ref{lemma:properties-of-the-functionals}), in combination with e.g.~\cite[Th.~11.1.4]{AmbrosioGigliSavare08}, together imply that gradient-flow solutions are unique.

To prove the distributional-solution property, set 
\[
\psi(s) := \frac12 \log \frac{\alpha s}{1-\alpha s}
  + \frac{\alpha s}{2(1-\alpha s)}\qquad\text{and}\qquad
\widetilde \psi(s) := \frac1{2\alpha(1-\alpha s)} - \frac1{2\alpha},
\]
so that $\widetilde\psi'(s) = s\psi'(s)$.

Comparing equation~\eqref{eq:GF-v-xi} with~\eqref{eq:HL-distributional} it follows that $\rho$ satisfies~\eqref{eq:HL-distributional} in the sense of distributions if we prove that $\bigl[\partial_y\widetilde\psi(\rho_t)\bigr](y) = \rho_t(y)\bigl[\partial_y \psi(\rho_t)\bigr](y)$ in the sense of distributions on $(0,T)\times\R$.
This identity follows from the next Lemma and the fact that $\partial_y {\xi}(\rho_t)\in L^2(\rho_t)$. 
\end{proof}

\begin{lemma}
Let $u\in L^\infty(\R)\cap \P(\R)$ satisfy $\partial_y \psi(u)\in L^2_u$ and $u<1/\alpha$ a.e.\ on $\R$. Then $u\partial_y \psi(u) = \partial_y \widetilde\psi(u)$ in $L^1_{\mathrm loc}(\R)$. 
\end{lemma}

\begin{proof}
First note that since $\psi'(s) \geq (2s)^{-1}$, the property $\partial_y \psi(u)\in L^2_u$ implies that $\partial_y u\in L^2$ and that $u$ is continuous. 
For $\e>0$ set $A_\e := \{ y\in \R: \e< u(y) < 1/\alpha -\e\}$. Since $\psi$ and $\widetilde\psi$ are smooth on $[\e,1/\alpha-\e]$ and $u$ is continuous, on $A_\e$ we have $
u\partial_y \psi(u) = \partial_y \widetilde \psi(u)$. 

Take $\varphi\in C_c(\R)$. From the estimate $\varphi u \partial_y \psi(u)\leq \frac12 u\varphi^2 + \frac12 u(\partial_y \psi(u))^2$ and the Lebesgue dominated convergence theorem we find
\begin{align*}
\int_\R \varphi u \partial_y \psi(u) 
= \int_{u>0} \varphi u \partial_y \psi(u)
&= \lim_{\e\downarrow 0} \int_{A_\e} \varphi u \partial_y \psi(u)\\
&= \lim_{\e\downarrow 0} \int_{A_\e} \varphi  \partial_y \widetilde \psi(u)
= \int_{u>0} \varphi \partial_y \widetilde \psi(u)
= \int_\R \varphi \partial_y \widetilde \psi(u).
\end{align*}
This proves the assertion.
\end{proof}

\section{Large-deviation principles}
\label{s:ldp}
%\red{
%\begin{enumerate}
%\item Definition of LDPs
%\item rate functions have infimum zero
%\item contraction principle
%\end{enumerate}
%}
The theory of large deviations characterizes the probability of events that become exponentially small in an asymptotic sense. Consider a sequence of probability measures $\{ \gamma _n \} _{n=1} ^{\infty}$ on some space $\mathcal X$. Large-deviation theory describes exponentially small probabilities under the $\gamma _n$'s in the limit $n \to \infty$, in terms of a {\emph rate function} $I: \mathcal X \to [0, \infty]$, in the following (rough) sense: for $A \subset \mathcal X$,
\begin{align*}
	\gamma _n (A) \sim e^{- n \inf _{x \in A} I(x)}, \ \ \textrm{as } n \to \infty.
\end{align*} 
This is formalized by the notion of a large-deviation principle. Before giving the definition we define the type of functions $I$ of interest in this setting; $\mathcal X$ is here taken to be a complete separable metric space.
\begin{definition}
	A function $I: \mathcal X \to [0,\infty]$ is called a rate function if it is lower semicontinous. The function $I$ is called a good rate function if for each $\alpha \in [0, \infty)$, the sublevel sets $\{ x: I(x) \leq \alpha \}$ are compact.
\end{definition}
Note that for a good rate function lower semicontinuity follows from the compact sublevel sets.

We are now ready to state the definition of a large-deviation principle. The definition can be made more general, however the following form suffices for this paper.
\begin{definition}
\label{def:LDP}
	Let $\{ \gamma _n \}$ be a sequence of probability measures on a complete separable metric space $\mathcal X$. We say that the sequence $\{ \gamma _n \}$ satisfies a large-deviation principle with rate function $I: \mathcal{X} \to [0, \infty]$ if for every measurable set $A \subset \mathcal X$,
	\begin{align*}
		- \inf _{x \in A ^{\circ}} I(x) \leq \liminf _{n \to \infty} \frac{1}{n} \log \gamma _n  (A^\circ) \leq \limsup _{n \to \infty} \frac{1}{n}\log \gamma_n (\bar A) \leq - \inf _{x \in \bar A} I(x),
	\end{align*}
	where $A ^{\circ}$ and $\bar A$ denote the interior and closure, respectively, of the set $A$.
\end{definition}
This definition is also referred to as a {\it strong} large-devation principle and there is a related notion of a {\it weak large-deviation principle}: The sequence $\{ \gamma _n \}$ is said to satisfy a weak large-deviation principle, with rate function $I$, if the lower bound in the previous definition holds for all measurable sets, and the following upper bound holds for every $\alpha < \infty$:
\begin{align*}
	\limsup \frac{1}{n} \log \gamma _n (A) \leq -\alpha,
\end{align*}
for $A$ a compact subset of $\Psi_I (\alpha) ^c$, where $\Psi _I (\alpha)$ is the $\alpha$-sublevel set of $I$.

A weak LDP can be strengthened to a full, or strong, LDP by showing exponential tightness of $\{ \gamma _n \}$:
\begin{definition}
	The sequence $\{ \gamma _n \}$ is {\it exponentially tight} if for every $\alpha < \infty$, there exists a compact $K_{\alpha}$ such that 
	\begin{align*}
		\limsup_{n \to \infty} \frac{1}{n} \log \gamma _n (K_{\alpha} ^c) < - \alpha.
	\end{align*}
\end{definition}

If $\gamma _n \to \delta _x$ for some $x \in \mathcal X$, that is if the sequence of underlying random elements has a unique deterministic limit as $n \to \infty$, then $I(x) =0$ and $I(\tilde x) > I(x)$ for all $\tilde x \in \mathcal{X} \setminus \{ x\}$.

A useful result when dealing with large deviations is the so-called contraction principle, a continuous-mapping-type theorem for the large-deviation setting.
\begin{theorem}[Contraction principle for large-deviations \cite{DemboZeitouni98}]
	Let $\mathcal X$ and $\mathcal Y$ be two complete separable metric spaces and $f: \mathcal X \to \mathcal Y$ a continuous mapping. Suppose the sequence $\{\gamma _n\} \subset \mathcal{P} (\mathcal X)$ satisfies a large-deviation principle with good rate function $I : \mathcal{X} \to [0,\infty]$. Then the sequence of push-forward measures $\{  f_{\#} \gamma_n \}$ satisfies a large-deviation principle on $\mathcal Y$ with good rate function $\tilde I$ defined as
	\begin{align*}
		\tilde I (y) = \inf \{ I(x) : x\in \mathcal X, \  y= f(x) \}, \ \ y \in \mathcal Y.
	\end{align*}
\end{theorem}
This result can be extended to  `approximately continuous' maps (see \cite[Section 4.2]{DemboZeitouni98}). To prove the main theorems of this paper we will use both the standard contraction principle above and a version with $n$-dependent maps.

\bigskip

We will also use the following `mean-field' localization result two times.
\begin{lemma}[Simple mean-field large-deviations result]
\label{l:mean-field-LDP}
Let $\mathcal X$ be a metric space. For each~$n$ let $P_n\in \P(\mathcal X)$, and for each $n$ and each $y\in\mathcal X$ let $Q_n^y\in \P(\mathcal X)$. Assume that for each $y$, $Q_n^y$ satisfies a strong large-deviation principle with good rate function $I^y$. Let $f: \mathcal X\to\R$ be lower semi-continuous. 

Assume that 
\begin{enumerate}
\item \label{l:mfldp:1}
If $I^y(y) < \infty$, then 
\[
\lim_{\delta\downarrow 0}\limsup_{n\to\infty}\,
\Bigl| \frac1n \log P_n(B_\delta(y)) - \frac1n \log Q_n^y(B_\delta(y)) + f(y)\Bigr| = 0.
\]
\item \label{l:mfldp:2}
If $I^y(y) = \infty$, 
\[
\sup_{\delta>0}\,\sup_{n\geq 1}\,
\Bigl| \frac1n \log P_n(B_\delta(y)) - \frac1n \log Q_n^y(B_\delta(y)) + f(y)\Bigr| =: C <\infty.
\]
\end{enumerate}
Then $P_n$ satisfies a weak large-deviation principle with good rate function $x\mapsto I^x(x) + f(x)$.
\end{lemma}

\begin{proof}
By \cite[Th.~4.1.11]{DemboZeitouni98}, the sequence $P_n$ satisfies a weak large-deviation principle provided that 
for all $y\in \mathcal X$,
\begin{equation}
\label{eq:liminf-limsup-equal}
\lim_{\delta\downarrow 0} \liminf_{n\to\infty} \frac1{n} \log P_n(B_{\delta}(y))
= 
\lim_{\delta\downarrow 0} \limsup_{n\to\infty} \frac1{n} \log P_n(B_{\delta}(y)),
\end{equation}
in which case the common value of the two is the negative of the rate function at $y$. 

If $y$ is such that $I^y(y)<\infty$, then by condition~\ref{l:mfldp:1}, and using the lower semi-continuity of $I^y$, 
\begin{align*}
\lim_{\delta\downarrow 0}\liminf_{n\to\infty} \frac1{n} \log P_n(B_{\delta}(y))
&\geq\lim_{\delta\downarrow 0}\liminf_{n\to\infty}\frac1n \log Q_n^y(B_\delta(y)) - f(y)\\
&\geq\lim_{\delta\downarrow 0}\Bigl(-\inf_{x\in B_\delta(y)} I^y(x) \Bigr)- f(y)\\
&= -I^y(y) -f(y).
\end{align*}
Similarly, 
\begin{align*}
\lim_{\delta\downarrow 0}\limsup_{n\to\infty} \frac1{n} \log P_n(B_{\delta}(y))
&\leq\lim_{\delta\downarrow 0}\limsup_{n\to\infty}\frac1n \log Q_n^y(B_\delta(y)) - f(y)\\
&\leq\lim_{\delta\downarrow 0}\Bigl(-\inf_{x\in \overline{B_\delta(y)}} I^y(x) \Bigr)- f(y)\\
&= -I^y(y) -f(y).
\end{align*}
This proves~\eqref{eq:liminf-limsup-equal} for the case $I^y(y)<\infty$. If $I^y(y) = \infty$, then by condition~\ref{l:mfldp:2}, 
\begin{align*}
\lim_{\delta\downarrow 0}\limsup_{n\to\infty} \frac1{n} \log P_n(B_{\delta}(y))
&\leq\lim_{\delta\downarrow 0}\limsup_{n\to\infty}\frac1n \log Q_n^y(B_\delta(y)) - f(y) + C\\
&\leq\lim_{\delta\downarrow 0}\Bigl(-\inf_{x\in \overline{B_\delta(y)}} I^y(x) \Bigr)- f(y) + C\\
&= -\infty.
\end{align*}
This concludes the proof of the lemma. 
\end{proof}

\section{The particle system \(Y^n\) and the transformed particle system \(X^n\)}
\label{sec:part-systems}
As mentioned in the introduction, the proofs of the results of this paper are based on  a `compression' mapping that is very specific for this system, and which was already illustrated in Figure~\ref{fig:compression-map}. 

\subsection{The `compression' map}
\label{ss:compression_map}
The idea is to consider a collection of rods of length $\alpha/n$ in the one-dimensional domain $\spY$, described by their empirical measure, and map them to a collection of zero-length particles by `collapsing' them to zero length and moving  the rods on the right-hand side up towards the left. For notational reasons we prefer to define the inverse operation, which is to map zero-length particles in $\spX$ to particles of length $\alpha/n$ in $\spY$ by `expanding' each zero-length particle to length $\alpha/n$ and `pushing along' all the particles to the right. 

This mapping comes in two forms, one for the discrete case and one for the continuous case. For convenience we write $\P^n(E)$ for the set of empirical measures of $n$ points, i.e. 
\begin{equation}
\label{def:Pn}
\P^n(E) := \biggl\{ \frac1n \sum_{i=1}^n \delta_{z_i}: z_i\in E, \ i=1,\dots,n\biggr\}.
\end{equation}
%We will always assume that the $z_i$ are ordered: $z_i\leq z_{i+1}$.
\begin{definition}[Expansion maps]
\label{def:AAn}
\noindent
\begin{enumerate}
\item The operator $A_n$ maps empirical measures   of zero-length particles to the corresponding empirical measures of rods by expanding each particle by $\alpha/n$:
\begin{equation}
\label{eq:mapping-particles}
A_n: \P^n(\spX) \to \P^n(\spY), \qquad 
	\frac{1}{n} \sum _{i=1} ^n \delta _{x_i} \mapsto \frac{1}{n} \sum _{i=1} ^n \delta _{y_i}, \qquad  y_i = x_i  + (i-1)\frac\alpha n,
\end{equation}
where we assume that the $x_i$ are ordered ($x_i\leq x_{i+1}$).
\item The operator $A$ maps `particle densities' to `rod densities' in a similar way: if $\rho\in \P(\spX)$, then
\[
A: \P(\spX) \to \P(\spY), \qquad \mu \mapsto \rho,
\]
where $\rho$ is constructed as follows: let\/ $\X$ be the inverse cumulative distribution function  (icdf) of $\mu$, i.e. 
\begin{equation}
\label{def:icdf-X}
\X (m) := \inf \{ x\in \spX: F(x) > m\} , \quad F(x) = \mu((-\infty,x]),
\end{equation}
and set 
\begin{equation}
\label{eq:def:AAn-Y}
\Y(m) := \X(m) + \alpha m, \qquad\text{for all }m\in [0,1].
\end{equation}
Then $\rho\in \P(\spY)$ is defined to be the measure whose icdf is\/ $\Y$, i.e.\ we set $\rho$ to be the distributional derivative of the corresponding cumulative distribution function $G$,
\[
%\rho((-\infty,\Y(m)]) := m, \qquad\text{for all }m\in [0,1].
\rho := G', \qquad \text{with }G(y) := \inf\bigl\{ m\in [0,1]: \Y(m) > y\bigr\}
\quad \text{for }y\in \spY.
\]
\end{enumerate}
\end{definition}

\begin{lemma}[Wasserstein properties of the expansion maps]
\label{lemma:isometries}
\noindent
\begin{enumerate}
%\item 
%%\label{lemma:isometries:linfty}
%For all $\mu\in \P(\spX)$, $A\mu$ has Lebesgue density bounded by $1/\alpha$; \red{Do we use this?}
\item $A$ and $A_n$ are isometries for the Wasserstein-2 distance, i.e.\ $W_2(A\mu_1,A\mu_2) = W_2(\mu_1,\mu_2)$ and $W_2(A_n\mu_{n,1},A_n\mu_{n,2}) = W_2(\mu_{n,1},\mu_{n,2})$;
\item For $\mu \in \P (\spX)$ and  $\mu _n \in \P^n (\spX)$,
\begin{equation}
\label{est:W2-A-An}
\big|W_2(A\mu,A_n\mu_n) - W_2(\mu,\mu_n)\big| \leq \frac \alpha n;
\end{equation}
\item \label{lemma:isometries:smeared}
%Finally, if $\mu_n = \frac1n \sum_{i=1}^n \delta_{x_i}\in \P^n(\spX)$, then $A\mu_n$ is equal to $\frac1n \sum_{i=1}^n \delta^{\alpha,n}_{x_i+\alpha (i-1)/n}\in \P(\spY)$, where $\delta_x^{\alpha,n}$ is the smeared delta function $\frac n\alpha \One_{[x,x+\alpha/n]}$ of mass $1$.
If $\mu_n = \frac1n \sum _{i=1} ^n \delta _{x_i} \in \P^n(\spX)$, then $A\mu$ is equal to
\begin{align*}
	\frac{1}{n}\sum _{i=1}^n \delta ^{\alpha, n} _{x_i + (i-1)\alpha/n} := \frac{1}{n}\sum _{i=1}^n \frac{n}{\alpha}\Indicator_{[x_i+ (i-1)\alpha/n,\, x_i + i\alpha/n]}
\end{align*}
\end{enumerate}
\end{lemma}

\begin{proof} 
%The bound on the Lebesgue density of $A\mu$ follows from the properties of the inverse cumulative distribution function: by the definition~\eqref{eq:def:AAn-Y},  $\Y(b)-\Y(a)\geq \alpha(b-a)$ for all $0\leq a\leq b\leq 1$; since $\rho((\Y(a),\Y(b)]) = b-a$, it follows that the Lebesgue density of $\rho$ is bounded from above by $1/\alpha$.
%
The isometry of $A$ follows from writing the Wasserstein distance in terms of the icdf (see~\eqref{eq:W2-in-icdf}):
\begin{align*}
W_2(\mu_1,\mu_2)^2 &= \int_0^1 |\X_1(m)-\X_2(m)|^2\, dm
= \int_0^1 |\X_1(m)-\alpha m - (\X_2(m)-\alpha m)|^2\, dm\\
&= \int_0^1 |\Y_1(m) - \Y_2(m)|^2\, dm
= W_2(A\mu_1,A\mu_2)^2.
\end{align*}
For $A_n$ the isometry follows from observing that {monotone transport maps} in fact map $x_{1,i}$ to $x_{2,i}$ and   $y_{1,i}$ to $y_{2,i}$ (i.e.\ they preserve the order) and therefore
\begin{align*}
W_2(\mu_{n,1},\mu_{n,2})^2 
&= \frac1n \sum_{i=1}^n |x_{1,i}-x_{2,i}|^2 
= \frac1n \sum_{i=1}^n \left|x_{1,i}-i\frac\alpha m -\Bigl(x_{2,i}-i\frac\alpha m\Bigr)\right|^2 \\
&= \frac1n \sum_{i=1}^n |y_{1,i}-y_{2,i}|^2 
= W_2(A_n\mu_{n,1},A_n\mu_{n,2})^2.
\end{align*}
To estimate the difference $W_2(A\mu,A_n\mu_n) - W_2(\mu,\mu_n)$, the same formulation of the Wasserstein distance in terms of icdf's becomes
\begin{align*}
W_2(\mu,\mu_n) 
&= \left[\int_0^1 |\X(m) - x_{\lceil nm\rceil}|^2 dm\right]^{\frac12}
= \left[\int_0^1 \left|\Y(m) +\alpha m  - \Bigl(y_{\lceil nm\rceil} + \lceil nm\rceil\frac\alpha n\Bigr)\right|^2dm\right]^{\tfrac12}\\
&\leq \left[\int_0^1 \left|\Y(m) - y_{\lceil nm\rceil} \right|^2 \, dm\right]^{\frac12}
+ \left[\int_0^1 \left|\alpha m  -  \lceil nm\rceil\frac\alpha n\right|^2 \, dm\right]^{\frac12}\\
&\leq W_2(A\mu,A_n\mu_n) + \frac\alpha n.
\end{align*}
The opposite inequality follows similarly.

Finally, to prove part~\ref{lemma:isometries:smeared}, the fact that $A$ maps $\frac1n \sum_{i=1}^n \delta_{x_i}$ to $\frac1n \sum_{i=1}^n \delta^{\alpha,n}_{x_i+\alpha (i-1)/n}\in \P(\spY)$ follows from remarking that $\delta_{x_1}$ is mapped by $A$ to the left-most smeared delta function $\delta^{\alpha,n}_{x_1}$; the second one, $\delta_{x_2}$, to $\delta^{\alpha,n}_{x_2+\alpha/n}$; and so forth. 
\end{proof}

\subsection{Mapping particle systems}
\label{subsec:mapping_particle_systems}

The compression and decompression maps $A_n$ and $A_n^{-1}$ create a one-to-one connection between two stochastic particle systems, which is the basis for the proofs of the two main theorems.
We now make this connection explicit.

First, given a measure $\mu$ on $\spX$, the maps $A$ and $A_n$ induce corresponding maps from $\spX$ to $\spY$, made explicit by the following Lemma.

\begin{lemma}
\label{lemma:A-is-pushforward}
Let $\mu\in \P(\spX)$. Define the map 
\begin{equation}
\label{def:Tmu}
T_\mu:\spX\to \spY, \quad x\mapsto  x+ \alpha\, \mu\big((-\infty,x)\big).
\end{equation}
Then
\begin{enumerate}
\item \label{l:A:props:1}
If $\mu=\frac1n \sum_{i=1}^n \delta_{x_i}\in \P^n(\spX)$, and $x_i< x_{i+1}$ for all $i$, then $A_n\mu = (T_\mu)_\#\mu$.
\item If $\X$ is the icdf of an absolutely continuous $\mu$, then 
\begin{equation}
\label{eq:TrhoX}
T_\mu \X(m) = \X(m) + \alpha m.
\end{equation}
\item 
\label{i:l:A-Tmu-3}
If $\mu$ is Lebesgue-absolutely-continuous, and $\tilde \mu \in \P(\spX)$, then $A\mu  = (T_\mu)_\# \mu$  and $(\bm t_{A\mu}^{A\tilde \mu}-\mathrm{id})A\mu = (T_\mu)_\# \bigl[(\bm t_{\mu}^{\tilde \mu}-\mathrm{id})\mu\bigr]$.
\end{enumerate}
\end{lemma}

\begin{proof}
For $\mu = \frac1n \sum_{i=1}^n \delta_{x_i}$, with $x_i< x_{i+1}$, the claim $A_n\mu =(T_\mu)_\#\mu$ follows from observing that $T_\mu(x_i) = x_i + \alpha(i-1)/n$, and therefore
\[
(T_\mu)_\#\mu = \frac1n \sum_{i=1}^n \delta_{T_\mu x_i} = 
\frac1n \sum_{i=1}^n \delta_{x_i + \alpha(i-1)/n} = A_n\mu.
\]

Next, assume that $\mu\in \P(\spX)$ is Lebesgue-absolutely-continuous; then the cumulative distribution function $F(x) := \mu((-\infty,x])$ is continuous, and consequently the inverse cumulative distribution function $\X$  satisfies $F(\X(m)) = m$ for all $m$ (Lemma~\ref{lemma:icdf-transformation}). The expression~\eqref{eq:TrhoX} then follows from remarking that 
\[
T_\mu \X(m) \stackrel{\eqref{def:Tmu}}= \X(m) + \alpha\mu \big((-\infty,\X(m))\big)
= \X(m) + \alpha\mu \big((-\infty,\X(m)]\big)
= \X(m) + \alpha m.
\]

Turning to part~\ref{i:l:A-Tmu-3}, we have for any $\varphi\in C_b(\spY)$, writing $\Y$ for the icdf of $A\mu$ as in Definition~\ref{def:AAn},
\begin{eqnarray*}
\int_{\spY} \varphi(y) (A\mu)(dy) 
&\stackrel{\eqref{eq:transformation-X}}=& \int_0^1 \varphi(\Y(m))\, dm 
  \stackrel{\eqref{eq:def:AAn-Y}}= \int_0^1 \varphi(\X(m) + \alpha m)\, dm\\
&=& \int_0^1 \varphi(\X(m) + \alpha F(\X(m)))\, dm 
  \stackrel{\eqref{eq:transformation-X}}= \int_{\spX}\varphi(x+\alpha F(x))\, \mu(dx) \\
  &\stackrel{\text{$\mu$ a.c.}}=& \int_{\spX}\varphi\Big(x+\alpha \mu\big((-\infty,x)\big)\Big)\, \mu(dx)
  \\
&=& \int_{\spY} \varphi(y) \, (T_\mu)_\#\mu(dy).
\end{eqnarray*}
This proves $A\mu = (T_\mu)_\#\mu$ for absolutely-continuous $\mu$. 

Finally, to prove that $(\bm t_{A\mu}^{A\tilde \mu}-\mathrm{id})A\mu = (T_\mu)_\# \bigl[(\bm t_{\mu}^{\tilde \mu}-\mathrm{id})\mu\bigr]$, we write similarly, using~\eqref{eq:transformation-X} and~\eqref{prop:t-map-icdf},
\begin{align*}
\int_{\spY} \varphi(y) (\bm t_{A\mu}^{A\tilde \mu}(y)-y) \, (A\mu)(dy)  
&= \int_0^1 \varphi(\Y(m)) (\tilde\Y(m)-\Y(m))\, dm\\
&= \int_0^1 \varphi(\X(m) +\alpha m)  (\tilde\X(m)-\X(m))\, dm\\
&= \int_{\spX} \varphi(x + \alpha\mu((-\infty,x))) (\bm t_\mu^{\tilde \mu}(x)-x)\, \mu(dx)\\
&= \int_{\spX} \varphi(T_\mu x) (\bm t_\mu^{\tilde \mu}(x)-x)\, \mu(dx).
\end{align*}
\end{proof}

\subsection{The particle systems of this paper}

We now state the assumptions on $V$ and $W$ and define precisely the systems of particles that we consider in this paper.

\begin{assumption}[Assumptions on $V$ and $W$.]
\label{ass:VW}
Throughout the paper we make the following assumptions.
\begin{enumerate}
\item[($V$)] The function $V:\R\to\R$ is $C^2 (\R)$, globally Lipschitz, $V'$ is $C^1 _b (\R)$ and there exist constants $c_1>0$, $c_2>0$ such that
\begin{equation}
\label{cond:coercivity}
V(y)\geq c_1|y| - c_2 \qquad\text{for all }y\in \R.
\end{equation}
\item[($W$)] The function $W:\R\to\R$ is $C^2 (\R)$, bounded and even, and $W'$ is $C^1 _b (\R)$.
\end{enumerate}
\end{assumption}

We will use two consequences of these assumptions:
\begin{align}
\label{bound:V-VTmu}
&\sup_{\mu\in \P(\spX)} \sup_{x\in \R} |V(x) - V(T_\mu x)| < \infty;\\
\label{bound:VTmu-from-below}
&\exists C>0: \inf_{\mu\in \P(\spX)} V(T_\mu x) \geq C(|x|-1)
\end{align}

The first set of particles $Y^n_i$ was already informally defined in the introduction; the second set $X^n_i$ is a compressed version of $Y^n_i$. We will often use the notation $\eta_n$ for the empirical measure of a set of particles,
\begin{equation}
\label{def:eta_n}
\eta_n: \R^n \mapsto \P(\R), \qquad
\eta_n(x) := \frac1n \sum_{i=1}^n \delta_{x_i}.
\end{equation}

\begin{definition}
\label{def:particle-systems}
\begin{enumerate}
\item
\label{i:def:particle-systems:Y}
For each $n\in \N$, the system of particles $Y^n = (Y^n_i)_{i=1,\dots,n}\subset C([0,\infty);\Omega_n)$ is defined by the generator
\[
\mathcal L_Y = \frac12 \Delta + \mathrm b\cdot \nabla,\qquad\text{with }
\mathrm b_i(y) = -\nabla V(y_i) -\frac1n \sum_{j=1}^n W'(y_i-y_j),
\]
with domain
\[
D(\mathcal L_Y) = \Big\{f\in C_b^2(\Omega_n): \frac{\partial f}{\partial n} = 0 \text{ on } \partial \Omega_n\Big\}.
\]
\item 
\label{i:def:particle-systems:X}
For each $n\in \N$, the system of particles $X^n = (X^n_i)_{i=1,\dots,n}\subset C([0,\infty);\spX^n)$ is defined by the generator
\[
\mathcal L_X = \frac12 \Delta + \mathrm b\cdot \nabla,
%\qquad\text{with }
%\mathrm b_i(x) = -\nabla V(y_i) -\frac1n \sum_{j=1}^n W'(y_i-y_j),
%\]
\qquad\text{with domain}\qquad
%\[
D(\mathcal L_X) = C_b^2(\spX^n),
\]
where the drift $\mathrm b$ is now given by
\begin{equation}
\label{def:b:first-def}
\mathrm b_i(x) := b(x_i,\eta_n(x)) := -V'(T_{\eta_n(x)} x_i) - \int_{\spX} W'(T_{\eta_n(x)} x_i - T_{\eta_n(x)} x') \, \eta_n(x)(dx').
\end{equation}
\end{enumerate}
\end{definition}

\begin{lemma}
\label{l:ex-un-SDEs}
For these two particle systems, weak solutions exist and are unique, and at each $t>0$ the laws of $X^n(t)$ and $Y^n(t)$ are absolutely continuous with respect to the Lebesgue measure.
\end{lemma}

This result is more-or-less standard, and the proof is given in the Appendix.
In fact, throughout this paper, unless explicitly stated otherwise, whenever we speak of existence or uniqueness of a solution of a stochastic differential equation, we are referring to the existence of weak solutions and uniqueness in law \cite[Section 5.3]{KaratzasShreve98}. Henceforth, unless required for the argument at hand, we do not go into details (such as corresponding filtrations, or similar aspects) about the weak solutions under study.

%
%
%\medskip
%The system of particles $X_i^n$ is simpler than that of $Y_i^n$, in that the particles `have zero length' and need not preserve any ordering; on the other hand, the interaction between the particles, given by the drift $b$, has a non-trivial dependence on the empirical measure $\mu_n$, and this dependence makes $\mathrm b_i$  a bounded function of $X_i^n$ but no more regular than that. 
%
%
%The following lemma makes this statement precise. 
%
%Note that one of the several ways in which the particle systems differ is the fact that particles $Y^n_i$ preserve their ordering;  the particles $X^n_i$ need not, and indeed they will typically pass each other. Since we are only interested in the empirical measures $\rho_n= \frac1n \sum_{i=1}^n \delta_{Y^n_i}$ of the particles, which are invariant under renumbering, there would be no change in the description if particles $Y^n_i$ were allowed to pass each other. However, this would not solve the problem of the steric interaction, which generates a strong interaction between the particles. Indeed, the reason for introducing the compressed particle system $X^n$ is to transform this strong interaction (in the form of steric interaction) into a weak one (in the form of dependence of the drift $\mathrm b$ on the other particles). 

\medskip

The following lemma makes the relationship between the two particle systems precise.
\begin{lemma}[Equality of distributions]
\label{lemma:mapping-particle-systems}
Let $\rho_n(t)= \frac1n \sum_{i=1}^n \delta_{Y^n_i(t)}$ and $\mu_n(t) = \frac1n \sum_{i=1}^n \delta_{X^n_i(t)}$ be the empirical measures of the particle systems $Y^n$ and $X^n$. The stochastic processes $\rho_n$ and  $A_n\mu_n$ have the same distribution in $C\big([0,\infty);\P(\R_y)\big)$.
\end{lemma}

\begin{proof}
The idea of this property goes back to Rost~\cite{Rost84}, who used it for the particle system~$Y^n$ without potentials $V$ and $W$. Because of the additional complexity of the two potentials $V$ and $W$ we give an independent proof. 

Since every function of $\eta_n(x)$ maps one-to-one to a symmetric function of $x$ (that is, a function $f:\R^n\to\R$ such that $f(x_1,\dots, x_n) = f(x_{\sigma_1},\dots, x_{\sigma_n})$ for all permutations~$\sigma$), the martingale problem for the random measure-valued process $\rho_n=\eta_n(Y^n)$ can be reformulated as the property that 
\begin{equation}
\label{def:mart-sol-tildeY}
M_t := f(Y^n(t)) - \int_0^t\mathcal (L_Yf)(Y^n(s))\, ds 
\qquad \text{is a martingale for all symmetric }f\in D(\mathcal L_Y).
\end{equation}
Given the process $X^n$, consider the transformed process $\widehat Y^n$ that is given by the expression (at each time $t$)
\[
\widehat Y^n := \bigl( T_{\eta_n(X^n)} X^n_1,\, T_{\eta_n(X^n)} X^n_2,\,\dots,
T_{\eta_n(X^n)} X^n_n\bigr).
\]
Whenever all $X^n_i$ are distinct, we have $\eta_n(\widehat Y^n) = A_n\eta_n(X^n)$ by part~\ref{l:A:props:1} of Lemma~\ref{lemma:A-is-pushforward}. Since the $X^n$ are almost surely distinct at any time, we have proved the lemma if we show that  $\widehat Y^n$ satisfies~\eqref{def:mart-sol-tildeY}.

Note that for any $x\in\R^n$ without collisions, i.e.\ with $x_i\not= x_j$ for $i\not=j$, 
\begin{align*}
\bigl[\partial_{x_k} (f\circ T_{\eta_n})\bigr](x) 
&= \partial_{x_k} \Bigl[f\Bigl( T_{\eta_n(x)}x_1, \dots, T_{\eta_n(x)}x_n\Bigr)\Bigr]\\
&= \partial_{x_k} \Bigl[f\Bigl( x_1 + \frac\alpha n \#\{\ell:x_\ell<x_1\}, \; \dots, \;
x_n + \frac\alpha n \#\{\ell:x_\ell<x_n\}\Bigr)\Bigr]\\
&\leftstackrel{(*)}= (\partial_k f)\Bigl( x_1 + \frac\alpha n \#\{\ell:x_\ell<x_1\}, \; \dots, \;
x_n + \frac\alpha n \#\{\ell:x_\ell<x_n\}\Bigr)\\
&= \bigl[(\partial_k f)\circ T_{\eta_n}\bigr](x).
\end{align*}
The equality $(*)$ holds because each of the terms $\#\{\ell:x_\ell<x_j\}$ is  constant away from the set of collisions. With this expression we find that e.g.\ for each $k$, 
\begin{align*}
-V'(T_{\eta_n(x)}x_k) \partial_{x_k}  \Bigl[f\bigl( T_{\eta_n(x)}x\bigr)\Bigr]
= -V'(y) (\partial_k f)(y)\Big|_{y_j=T_{\eta_n(x)}x_j\, \forall j}
\end{align*}
and by collecting similar arguments we conclude that  
\begin{equation}
\label{eq:trans-generator}
\mathcal L_X(f\circ T_{\eta_n})(x) = \bigl[(\mathcal L_Yf)\circ T_{\eta_n}\bigr](x)
\qquad\text{at any non-collision point $x\in \R^n$}.
\end{equation}

Also note that the function $f\circ T_{\eta_n}$ is an element of $D(\mathcal L_X)$. This follows since at non-collision points $f\circ T_{\eta_n}$ is as smooth as $f$ (by the same constancy argument as above); at the collision set,  $f\circ T_{\eta_n}$ connects with regularity $C^2$  by the $C^2$--regularity of $f$ in $\Omega_n$, the boundary condition $\partial_n f=0$,  and the symmetry of $f$.

To conclude the proof, we show that $\widehat Y^n$ satisfies~\eqref{def:mart-sol-tildeY} by rewriting
\begin{align*}
f\big(\widehat Y^n(t)\big) - \int_0^t (\mathcal L_Yf) \big(\widehat Y^n(s)\big)\, ds
&= \big(f\circ T_{\eta_n}\big)(X^n(t)) - \int_0^t (\mathcal L_Yf)\circ T_{\eta_n}  (X^n(s))\, ds\\
&\leftstackrel{(**)}=\big(f\circ T_{\eta_n}\big)(X^n(t)) - \int_0^t \bigl[\mathcal L_X(f\circ T_{\eta_n})\bigr]  (X^n(s))\, ds,
\end{align*}
and this expression is a martingale by the properties of $X^n$.
Note that although the identity~\eqref{eq:trans-generator} holds only for non-collision points~$x$, the process $X^n$ spends zero time on the set of remaining points. Therefore the identity $(**)$ above  holds almost surely.
\end{proof}

\begin{lemma}[Transformed version of $\InvMeas_n$ and $\cZ_n$]
\label{l:transformed-inv-meas}
The particle system $X^n$ has invariant measure $\bbP_n\in\P(\spX^n)$, given by
\[
\bbP_n(dx) := \frac1{\cZ_n} 
\exp\biggl[\,-2 \sum_{i=1}^n V(T_{\eta_n(x)}x_i) - \frac1{n} 
\sum_{i,j=1}^n  W(T_{\eta_n(x)}x_i-T_{\eta_n(x)}x_j)\biggr] \, dx.
\]
The normalization constant~$\cZ_n$ is the same as in~\eqref{def:norm-const-Zn-intro} and can be written as
\begin{equation}
\label{eq:transformed-Zn}
\cZ_n = \int_{\R^n}  \exp\biggl[\,-2 \sum_{i=1}^n V(T_{\eta_n(x)}x_i) - \frac1{n} 
\sum_{i,j=1}^n  W(T_{\eta_n(x)}x_i-T_{\eta_n(x)}x_j)\biggr] \, dx.
\end{equation}
\end{lemma}

\noindent
This property follows from arguments very similar  to those of Lemma~\ref{lemma:mapping-particle-systems}, and we omit the proof.

\section{The functionals of this paper}
\label{s:functionals}

With the maps $A$ and $T_\rho$ defined in the previous section, we can also define the various functionals that we use in this paper. The functional~$\hFreeEnergy$ as defined in the introduction is one of these; in this section we review this definition and place it in a larger context. 

We define in total six functionals, three functionals $\hFreeEnergy$, $\hEnt_V$, and $\hEnergy_W$, on the set of 
``expanded'' measures $\P(\spY)$, and at the same time three transformed versions $\FreeEnergy$, $\Ent_V$, and $\Energy_{W}$ on the set of ``compressed'' measures $\P(\spX)$. We split the definition of $\hFreeEnergy$ of the introduction up into an entropic part $\hEnt_V$ and an interaction-energy part $\hEnergy_W$: 
\begin{align}
\notag
\hFreeEnergy &:= \hEnt_V + \hEnergy_W + C_\FreeEnergy,\qquad 
\hFreeEnergy,\ \hEnt, \ \hEnergy_{W}: \P(\spY) \to \R\cup\{\infty\},\\
\hEnt_V(\rho) &:= \begin{cases}
\ds \int_{\R} \rho \biggl[\frac12\log \ds\frac{\rho}{1-\alpha \rho} + V\biggr] + C_{\Ent}&\text{if $\rho$ is Lebesgue-a.c. and $\rho(y)<1/\alpha$ a.e.},\\[2\jot]
+\infty &\text{otherwise}
\end{cases}
\label{def:hEnt_V-Section-functionals}\\
\hEnergy_{W}(\rho) &:=  \frac12\int_{\spY}\!\!\int_{\spY} W(y-y')\, \rho(dy)\rho(dy'), \qquad \text{for }\rho\in \P(\spY).
\notag
\end{align}
The constants $C_\FreeEnergy$ and $C_{\Ent}$ are such that $\inf \hEnt_V = \inf \hFreeEnergy=0$. 
The integral in $\hEnergy_W$ is well-defined by the boundedness of $W$, and we show in Lemma~\ref{lemma:properties-of-the-functionals} below that the integral in $\hEnt_V$ is well-defined in $(-\infty,+\infty]$.

We then define the corresponding functionals on ``compressed'' space $\P(\spX)$ through the isometry $A$:
\begin{gather*}
\FreeEnergy, \, \Ent_V, \, \Energy_{W}: \P(\spX) \to \R\cup\{\infty\},\\
\FreeEnergy := \hFreeEnergy \circ A, \qquad
\Ent_V := \hEnt_V \circ A, \qquad \text{and}\qquad
\Energy_{W} := \hEnergy_{W}\circ A.
\end{gather*}
We also need the relative entropy: for two measures $\mu$ and $\nu$ on the same space, 
\begin{equation}
\label{e:def:RelEnt}
\RelEnt(\mu|\nu) := \begin{cases}
\int f \log f \, d\nu & \text{if } \mu \ll \nu, \ \mu = f\nu,\\
+\infty &\text{otherwise}.
\end{cases}
\end{equation}

\begin{lemma}[Properties of the functionals]
\label{lemma:properties-of-the-functionals}
\noindent
\begin{enumerate}
\item \label{i:l:props:F}
(Alternative formula for $\FreeEnergy$.)  We have
\begin{align}
\label{char:FreeEnergy}
&\FreeEnergy(\mu) = 
\begin{cases}
\ds \int_{\spX} \Big[\frac12 \log\mu(x) +  V(T_\mu x)\Big]\, \mu(dx)  \\
  \qquad\qquad \ds+ \frac12\int_{\spX}\!\int_{\spX} W(T_\mu x - T_\mu x')\, \mu(dx)\mu(dx') + C_\FreeEnergy&\text{if $\mu$ is Lebesgue-a.c.},\\
+\infty&\text{otherwise}.
\end{cases}
\end{align}

\item (Alternative formula for $\Ent_V$)
For given $\nu\in \P(\spX)$, define the measure $\bbQ^\nu\in \P(\spX)$,
\begin{equation}
\label{def:Qnu}
\bbQ^\nu(dx) := \frac{1}{\cZ^{\bbQ,\nu}} 
\exp\bigl[\,-2 V(T_{\nu}x)\bigr] \, dx, 
\quad\text{with}\quad
\cZ^{\bbQ,\nu} := \int_{\spX} \exp\bigl[\,-2  V(T_{\nu}x)\bigr] \, dx.
\end{equation}
%\red{Note that $e^{-2V(T_\nu x)}$ is integrable by the coercivity property~\eqref{bound:VTmu-from-below}.}

\label{i:l:props:F2}
We then have
\begin{align}
\label{char:FE-RelEnt}
&\Ent_V(\mu) := \frac12 \RelEnt(\mu|\bbQ^\mu)
+\frac12 \gamma(\mu), \qquad \text{for }\mu\in \P(\spX),\\
&\gamma(\mu) := C_\gamma - \log \cZ^{\bbQ,\mu},\label{def:gamma}\\[\jot]
&\text{where $C_\gamma$ is determined  by the property $\ds
\inf_{\mu\in \P(\spX)} \RelEnt(\mu|\bbQ^{\mu}) + \gamma(\mu) = 0.$}
\notag
\end{align}

%\begin{equation}
%\label{char:FE-RelEnt}
%\Ent_V(\mu) = \frac12 \RelEnt(\mu|\bbQ^\mu) + \frac12 \gamma(\mu).
%\end{equation}
%\\
%\label{char:hFreeEnergy}
%&\hFreeEnergy (\rho)
%= \begin{cases}
%\ds \int_{\spY} \Big[\frac12\log \rlap{$\ds\frac{\rho}{1-\alpha \rho} +  V\Big]\rho  \;+ \;\frac12 \int_{\spY}\int_{\spY} W(y-y')\rho(dy)\rho(dy') \;+\; c$}\\[4\jot]
%\qquad\qquad&\text{if $\rho$ is Lebesgue-a.c. and $\rho(y)<1/\alpha$ a.e.},\\[2\jot]
%+\infty &\text{otherwise},
%\end{cases}
%\end{align}
%where in both cases the constant $c$ is chosen such that $\inf\FreeEnergy=\inf\hFreeEnergy = 0$.

\item 
\label{i:l:basic-props-lambda-convexity}
The functionals $\FreeEnergy$ and $\hFreeEnergy$ are  lower semicontinuous and $\lambda$-convex for some $\lambda\in \R$.

%\item For each $\mu_0\in W_2(\R)$ with $\FreeEnergy(\mu_0)<\infty$, there exists a unique solution of the gradient flow~\eqref{} starting at $\mu_0$. Similarly, for each $\rho_0\in W_2(\R)$ with $\hFreeEnergy(\rho_0)<\infty$ there exists a unique solution of the gradient flow~\eqref{} starting at $\rho_0$. \red{Are we using this?}

\item (Subdifferential of $\FreeEnergy$.)
If $\mu$ is Lebesgue-a.c. and $\int_{\spX} |\partial_x\mu|^2/\mu < \infty$, then
\begin{equation}
\label{char:b-minimal-element-subdiff}
\frac{\partial_x \mu}{2\mu} -b(\cdot,\mu) \quad \text{is the element of minimal norm of }\partial\FreeEnergy(\mu),
\end{equation}
where $b$ was already given in~\eqref{def:b:first-def}:
\[
b(x,\mu) := -V'(T_\mu x) - \int_{\spX} W'(T_\mu x - T_\mu x') \, \mu(dx').
\]

\item (Subdifferential  of $\hFreeEnergy$.)
\label{i:l:properties:subdiff-hFE}
If $|\partial\hFreeEnergy(\rho)|<\infty$, then 
\[
\frac{\partial_y \rho}{2\rho(1-\alpha\rho)^2} + V' + W'*\rho
\]
is the element of minimal norm of $\partial\hFreeEnergy(\rho)$.
\end{enumerate}
\end{lemma}

\begin{remark}[Mean-field structure of the compressed rate functions]
\label{rem:mean-field}
The invariant-measure rate function $\hFreeEnergy$ and the dynamic rate function $\hat I$ that we introduce below both have a particular form. This is best observed in~\eqref{char:FE-RelEnt} and in~\eqref{def:I-time-dependent} below: the argument of the functional appears twice, first  as the first argument in the relative entropy,  and secondly as a parameter in the reference measure. This is a common structure in mean-field interacting particle systems (see e.g.~\cite{Leonard95} or~\cite[Ch.~X]{DenHollander00}). It reflects the fact that once the system has been `compressed' (i.e., transformed to $X^n$) the interaction between the particles has a `nearly-weakly-continuous' dependence on the empirical measure. The estimate~\eqref{est:RN-deriv} below illustrates this: while $\frac1n \log d\bbP_n/d\bbQ_n^\nu$ is not completely continuous in the empirical measure (the right-hand side does not vanish as $\delta\to0$ for finite $n$), the discontinuity does vanish in the limit $n\to\infty$. 

This structure is reflected in the fact that the entropic part of the free energy $\FreeEnergy$ is of the Gibbs-Boltzmann type $\int \mu\log\mu$. By contrast, the \emph{expanded} system has  a  different entropic term $\int \rho\log (\rho/(1-\alpha\rho)$, which reflects the fact that in the expanded system the particles have a strong interaction with each other.
\end{remark}

%
%Note how in the characterization of part~\ref{i:l:props:F2} above the argument $\mu$ appears both in the first argument of the relative entropy and as the parameter $\nu=\mu$ in $\bbQ^\nu$. This is a common structure in mean-field interacting particle systems; see also Theorem~\ref{th:JasperMarioOliver} and Remark~\ref{rem:mean-field}.
%

\begin{proof}
We first show that the integrals in~\eqref{char:FreeEnergy} and~\eqref{def:hEnt_V-Section-functionals} are well defined; since $\mu\mapsto\int_\R  \mu\log\mu $ is unbounded from below on the space of probability measures, this is not immediate. 
%Since $W$ is bounded, the integrals involving $W$ are well defined in both cases. 
For the first integral in~\eqref{char:FreeEnergy}, we write $\mu_V(dx) := e^{-2V(T_\mu x)}dx$, and use the inequality $s_- \leq (s+t)_- + |t|$ for the negative part $s_- := \max\{-s,0\}$ to estimate
\begin{align*}
\int_{\spX} \Big[\frac12 \log\mu(x) +  V(T_\mu x)\Big]_-\, \mu(dx) 
&=\frac12 \int _{\spX} \Big[\frac \mu{\mu_V} \log\frac \mu{\mu_V} \Big]_- \mu_V\\
&\leq \frac 12 \int_{\spX} \Big[\frac \mu{\mu_V} \log\frac \mu{\mu_V} - \frac \mu{\mu_V} +1  \Big]_- \mu_V + \frac12 \int_{\spX} \bigg| \frac \mu{\mu_V} - 1\bigg| \,\mu_V\\
&= \frac12 \int_{\spX} \bigl[\mu - \mu_V] < \infty.
\end{align*}
It follows that the first integral in~\eqref{char:FreeEnergy} is well defined in $(-\infty,\infty]$, and a similar  calculation shows the same for the first integral in~\eqref{def:hEnt_V-Section-functionals}. 
%This also implies that whenever $\FreeEnergy(\mu)$ is finite, the integral $\int\mu\log\mu$ is well-defined and finite. 

We next prove the formula~\eqref{char:FreeEnergy} for the functional $\FreeEnergy$. 
Since $\FreeEnergy(\mu)$ is defined as $\hEnt_V(A\mu) + \hEnergy_{W}(A\mu) + C_\FreeEnergy$, finiteness of $\FreeEnergy(\mu)$ implies that $\mu$ is absolutely continuous. The  second integral in~\eqref{char:FreeEnergy} then follows by Lemma~\ref{lemma:A-is-pushforward}:
\[
\int_{\spY}\int_{\spY} W(y-y')(A^{-1}\mu)(dy)(A^{-1}\mu)(dy') =
\int_{\spX}\int_{\spX} W(T_\mu x - T_\mu x') \mu(dx)\mu(dx').
\]
We turn to the first integral in~\eqref{char:FreeEnergy}. 
Again let $\FreeEnergy(\mu)$ be finite, which implies that there exists $\rho$  such that $\rho = A\mu$ and $\hEnt_V(\rho) = \Ent_V(\mu)<\infty$. This  implies that $\rho$ is Lebesgue-absolutely-continuous and satisfies $\rho(y)< 1/\alpha$ for almost all $y$.  
By Lemma~\ref{lemma:icdf-transformation} the icdf $\Y$ of $\rho$ is monotonic, and its derivative $\Y'(m)$ exists at almost all $m \in (0,1)$ and is equal to $1/ \rho (\Y(m))$. We then calculate
\begin{align*}
\hEnt_V(\rho) &= \int_{\spY} \rho(y) \bigg[\frac12\log\frac{\rho(y)}{1-\alpha\rho(y)}  + V(y)\bigg] \, dy
\stackrel{\eqref{eq:transformation-X}}= \int_0^1 \bigg[\frac12\log\frac{\rho(\Y(m))}{1-\alpha\rho(\Y(m))}  + V(\Y(m))\bigg] \, dm\\
&= \int _0 ^1 \left[\frac12 \log \Big( \frac{1}{\Y' (m) - \alpha } \Big)
    + V(\Y(m))\right ] dm.
\end{align*}
Since this integral is assumed to be finite, $\Y'(m)>\alpha$ for Lebesgue-a.e.\ $m$, and therefore  $\X'(m) = \Y'(m) - \alpha>0$ for almost all $m$. Inserting this into the expression above yields
\begin{align*}
\hEnt_V(\rho) &= \int _0 ^1 \left[\frac12\log \Big( \frac{1}{\X' (m)} \Big )  + V(\X(m) + \alpha m)\right] dm\\
&\leftstackrel{\eqref{eq:TrhoX}}=\int_0^1 \bigg[ \frac12 \log\mu(\X(m))  + V(T_\mu \X(m))\bigg]\, dm\\
&\leftstackrel{\eqref{eq:transformation-X}}= \int_{\spX} \mu(x) \bigg[\frac12\log\mu(x)  + V(T_\mu x)\bigg] \, dx.
\end{align*}

Writing
\[
\widetilde \Ent_V(\mu)
:= \begin{cases}
\ds \int_{\spX} \mu(x) \left[ \frac{1}{2}\log \mu(x) + V(T_\mu x) \right] dy 
    &\text{if $\mu$ is Lebesgue-a.c. },\\
\infty&\text{otherwise},
\end{cases}
\]
we therefore have proved 
\[
\Ent_V(\mu) < \infty \quad\Longrightarrow\quad  \widetilde\Ent_V(\mu)<\infty \text{ and } 
\hEnt_V(\mu) = \widetilde\Ent_V(\mu) .
\]
By reversing the argument we similarly show that
\begin{equation}
\label{prf:tidleEnt-to-hEnt}
\widetilde\Ent_V(\mu) < \infty \quad\Longrightarrow\quad  \Ent_V(\mu)<\infty \text{ and } 
\hEnt_V(\mu) = \widetilde\Ent_V(\mu) ,
\end{equation}
which concludes the proof of~\eqref{char:FreeEnergy}.

\medskip

Turning to the characterization~\eqref{char:FE-RelEnt}, assume that $\Ent_V(\mu)<\infty$, which implies by part~\ref{i:l:props:F} that $\mu$ is absolutely continuous and $x\mapsto \tfrac12 \log \mu(x) + V(T_\mu x) = \tfrac12\log \bigl[\mu(x)/\bbQ^\mu(x)\bigr] -\tfrac12 \log \mathcal Z^{\bbQ,\mu}$ is an element of $L^1(\mu)$. Therefore
\[
\frac12 \RelEnt(\mu|\bbQ^\mu) = \frac12\int_{\spX} \mu(dx) \log \frac{\mu(x) }{\bbQ^\mu(x)} 
= \int_{\spX} \mu(dx) \bigg[\frac12 \log \mu(x) + V(T_\mu x) \biggr] - \frac12 \log \mathcal Z^{\bbQ,\mu},
\]
which proves the formula~\eqref{char:FE-RelEnt} for the case $\Ent_V(\mu)<\infty$. On the other hand, if $\RelEnt(\mu|\bbQ^\mu)<\infty$, then we can reverse the argument above and obtain $\Ent_V(\mu)<\infty$ and equality. This proves part~\ref{i:l:props:F2}.

\medskip

To prove part~\ref{i:l:basic-props-lambda-convexity}, note that the lower semicontinuity in $\P_2(\R)$ of $\hEnt_V$ is a consequence of its convexity, and the lower semicontinuity of $\hEnergy_{W}$ follows from the  boundedness and continuity of~$W$. The isometries $A$ and $A^{-1}$ then transfer the same properties to $\Ent_V$ and $\Energy_{W}$. The $\lambda$-convexity of $\hFreeEnergy$ also is a standard result for functionals of this type; see e.g.~\cite[Sec.~5]{CarrilloMcCannVillani06}.

\bigskip

We next turn to the calculation of the element of minimal norm in the subdifferential of~$\FreeEnergy$, evaluated  at a  $\mu\in \P(\spX)$ with the properties that $\mu$ is Lebesgue absolutely continuous and $\partial_x \mu  = w \mu $ with $w\in L^2(\mu)$.

Take $\phi\in  C_c^\infty(\R)$ and set $r_\e(x) := x + \e\phi(x)$; note that for small $\e$, $r_\e$ is strictly increasing. Set $\mu_\e := (r_\e)_\#\mu$, and note that $\mu_\e$ also is absolutely continuous. From~\cite[Theorems~10.4.4 and~10.4.6]{AmbrosioGigliSavare08} (or an explicit calculation) we deduce that
\begin{equation}
\label{eq:directional-derivative-Ent}
\frac d{d\e} \frac12\int_{\R}\mu_\e(x)\log \mu_\e(x)\, dx \Big|_{\e=0} = 
  -\frac12 \int_{\R} \mu\partial_x \phi .
\end{equation}
Setting 
\[
\Energy_V (\nu) := \int_\R V(T_\nu x)\, \nu(dx),
\]
 we calculate 
\begin{align*}
\Energy_{V}((r_\e)_\#\mu) 
&=  \int_{\R} V(T_{(r_\e)_\#\mu} (x)) \,((r_\e)_\#\mu) (dx)
  = \int_{\R} V\Bigl(T_{(r_\e)_\#\mu} \bigl(r_\e (x)\bigr)\Bigr) \,\mu(dx)\\
&= \int_{\R} V\Bigl(r_\e(x) +\alpha \,(r_\e)_\#\mu\bigl((-\infty,r_\e(x)\bigr)\Bigr) \,\mu(dx) \\
&= \int_{\R} V\Bigl(r_\e(x) +\alpha\mu\bigl((-\infty,x)\bigr)\Bigr) \,\mu(dx).
\end{align*}
Therefore 
\[
\frac d{d\e} \Energy_{V}(\mu_\e) \Big|_{\e=0} = 
\int_{\R} V'\Bigl(x +\alpha\mu\bigl([0,x)\bigr)\Bigr) \, \phi(x) \,\mu(dx)
= \int_{\R} V'(T_\mu x) \, \phi(x) \,\mu(dx).
\]
Combining these expressions with a similar one for $\Energy_W$, and using $\int_{\spX} |\partial_x\mu|^2/\mu < \infty$, we find
\begin{align*}
\frac d{d\e} \FreeEnergy(\mu_\e)\Big|_{\e=0} &= 
- \int_{\R} \Bigl[\frac12 \partial_x \phi(x)  + b(x,\mu)\phi(x)\Bigr]\mu(dx)\\
&= \int_{\R} \Bigl[\frac{\partial_x\mu(x)}{2\mu(x)}  - b(x,\mu)\Bigr]\phi(x)\, \mu(dx).
\end{align*}
By an argument as in the proof of \cite[Th.~10.4.13]{AmbrosioGigliSavare08} it follows that $\partial_x\mu/2\mu - b(\cdot,\mu)$ is the element of minimal norm in the subdifferential $\partial\FreeEnergy(\mu)$. 

Finally, the proof of part~\ref{i:l:properties:subdiff-hFE} follows along very similar lines as the previous part, and we omit it.
\end{proof}

\newpage

\section{Pathwise large deviations for \texorpdfstring{$X^n$}{Xn} with i.i.d. initial data}
\label{s:ldp-iid}

The aim of this section is to prove the following large-deviations principle for the compressed particle system $X^n$. In this theorem  we start the evolution with i.i.d.\ initial data, which is different from the situation of Theorem~\ref{th:LDP-path};  we use the  name $Z^n$ in order to distinguish this case from the case we study in the proof of Theorem~\ref{th:LDP-path}.

\begin{theorem}
\label{th:LDP-dynamic-X}
Let $\xi\in \P(\R)$, and for each $n\in \N$ let $Z^n$ be the particle system defined in~Definition~\ref{def:particle-systems}(\ref{i:def:particle-systems:X}), with initial data $Z^n(0)$ drawn from $\xi^{\otimes n}$. 

%\begin{enumerate}
%\item \label{i:th:LDP-X-a}
The random variable $t\mapsto  \eta_n(Z^n(t))$ satisfies a large-deviation principle on $C([0,T];\P(\spX))$ with good rate function
\begin{equation}
\label{def:I-time-dependent}
\mathfrak I_{\xi}(\mu) := \inf\Bigl\{ \RelEnt(P|\mathbb W_{\xi}^P): P_t = \mu_t \text{ for all }t\text{ and } \RelEnt(P|\mathbb W_{\xi})<\infty \Bigr\}.
\end{equation}
%\item \label{i:th:LDP-X-b} We have
%\begin{equation}
%\label{ineq:RF-dyn-RF-initial}
%\mathfrak I_\xi (\mu) \geq \RelEnt(\mu_0|\xi).
%\end{equation}
%\red{Do we need this?}
%%and equality is achieved at some $\mu\in C([0,T];\P(\spX))$.
%\end{enumerate}

Here $\mathbb W_{\xi}\in \P\big(C([0,T];\R)\big)$ is the law of a Brownian particle with initial position drawn from~$\xi$, and for any $P\in \P(C([0,T];\R))$, the measure $\mathbb W_{\xi}^P\in \P\big(C([0,T];\R)\big)$ is the law of the process $Z^P$ satisfying the SDE in $\R$,
\begin{equation}
\label{def:W^rho}
dZ^P(t) = b\bigl(Z^P(t),P_t\bigr) \,dt + dB_t, \quad Z^P(0) \sim \xi.
\end{equation}
The notation  $P_t\in\P(\R)$ represents the time-slice marginal of the measure~$P$ at time $t$.

\end{theorem}

\begin{proof}
The assertion is a direct translation of the following theorem from~\cite{HoeksemaMaurelliHoldingTse20TR}:

\begin{theorem}[{\cite[Prop.~4.15 and Rem.~4.16]{HoeksemaMaurelliHoldingTse20TR}}]
\label{th:JasperMarioOliver}
Let $\Psi:\R^4\to\R$ and $\varphi_1:\R^2\to\R$ be bounded and globally Lipschitz continuous, and let $\varphi_2\in L^p(\R)$. Set $\varphi(x_1,x_2) := \varphi_1(x_1,x_2) + \varphi_2(x_1-x_2)$. Let $\xi\in \P(\R)$.

Let $Z^n = (Z^n_1,\dots,Z^n_n)$ solve the system of interacting SDEs in $\R^n$
\begin{multline}
\label{sde:JasperMarioOliver}
dZ^n_i(t) = \frac1n \sum_{j=1}^n \Psi\biggl(Z^n_i(t),Z^n_j(t), 
  \frac1n \sum_{\substack{\ell=1\\\ell\not=i}}^n \varphi(Z^n_i(t),Z^n_\ell(t)),
  \frac1n \sum_{\substack{\ell=1\\\ell\not=j}}^n \varphi(Z^n_j(t),Z^n_\ell(t))
  \biggr) \, dt  + dB_i,\\
\qquad Z^n_i(0)\sim \xi \text{ i.i.d.,}
\end{multline}
where $B_i$ are independent standard Brownian motions. Let $\widehat P^n$ be the corresponding  \emph{empirical process} 
\[
\widehat P^n := \frac1n \sum_{i=1}^n \delta_{Z^n_i} \in \P(C([0,T];\R)).
\]
Then $\widehat P^n$ satisfies a large-deviation principle on $\P(C([0,T];\R))$, with {good} rate function
\begin{equation}
\label{def:Rf-J}
J(P) := \begin{cases}
\RelEnt(P|\mathbb W_{\xi}^P) & \text{if }\RelEnt(P|\mathbb W_{\xi})<\infty,\\
+\infty &\text{otherwise.}
\end{cases}
\end{equation}
\end{theorem}

To prove Theorem~\ref{th:LDP-dynamic-X}, we apply Theorem~\ref{th:JasperMarioOliver} to the particle system $Z^n$ of  Theorem~\ref{th:LDP-dynamic-X}. Let $H := \chi^{}_{(0,\infty)}$ be the lower semicontinuous Heaviside function; define $\varphi_1$ and $\varphi_2$ by
\[
\varphi_1(x_1,x_2) := H(x_1-x_2) - \varphi_2(x_1-x_2),
\qquad 
\varphi_2(s) := \begin{cases}
0 & s\leq 0,\\
1 & 0<s\leq 1,\\
\text{smooth interpolation} & 1\leq s\leq 2,\\
0 & s\geq 2,
\end{cases}.
\]
We also set 
\[
\Psi(x,y,s,t) := -V'(x+\alpha s) - W'(x-y + \alpha(s-t)).
\]
Then the functions $\varphi(x_1,x_2) := \varphi_1(x_1,x_2) + \varphi_2(x_1-x_2) = H(x_1-x_2)$ and $\Psi$ satisfy the conditions of Theorem~\ref{th:JasperMarioOliver}. For these choices of $\varphi$ and $\Psi$, we have 
\begin{align*}
\Psi&\biggl(Z^n_i,Z^n_j, 
  \frac1n \sum_{\substack{\ell=1\\\ell\not=i}}^n \varphi(Z^n_i,Z^n_\ell),
  \frac1n \sum_{\substack{\ell=1\\\ell\not=j}}^n \varphi(Z^n_j,Z^n_\ell)
  \biggr)\\
&= \Psi\biggl(Z^n_i,Z^n_j, 
  \frac1n \sum_{\substack{\ell=1\\\ell\not=i}}^n H(Z^n_i - Z^n_\ell),
  \frac1n \sum_{\substack{\ell=1\\\ell\not=j}}^n H(Z^n_j-Z^n_\ell)
  \biggr)\\
&= \Psi\biggl(Z^n_i,Z^n_j, 
  \frac1 n \#\{Z^n_\ell < Z^n_i),
  \frac1 n \#\{Z^n_\ell < Z^n_j)\biggr)\\
&= -V'\Bigl(Z^n_i + \frac\alpha n \#\{Z^n_\ell < Z^n_i\}\Bigr)
  - W'\Bigl(Z^n_i -Z^n_j + \frac\alpha n \#\{Z^n_\ell < Z^n_i\} - \frac\alpha n \#\{Z^n_\ell < Z^n_j\}\Bigr),
\end{align*}
which equals $\mathrm b$ in~\eqref{def:b:first-def}. Therefore 
the particle system~$Z_i$  is a weak solution of~\eqref{sde:JasperMarioOliver}. Theorem~\ref{th:JasperMarioOliver} then implies that the empirical process $\hat P^n = n^{-1} \sum_{i=1}^n \delta_{Z^n_i}$ satisfies a large-deviation principle with rate function~\eqref{def:Rf-J}.

Since the mapping $\mathcal T: \P(C([0,T];\R))\to C([0,T];\P(\R))$ given by 
\[
\langle (\mathcal TP)_t,\phi\rangle  := \int_{C([0,T];\R)} \phi(x(t))\, P(dx),
\qquad\text{for $P\in \P(C([0,T];\R))$ and $\phi\in C_b(\R)$,}
\]
is continuous, the contraction principle (e.g.~\cite[Sec.~4.2.1]{DemboZeitouni98}) implies that $\mu_n = \mathcal T\hat P^n$ satisfies a large-deviation principle on $C([0,T];\P(\spX))$ with good rate function~\eqref{def:I-time-dependent}. 
\end{proof}

\newpage

\section{Preliminary estimates for the pathwise large deviations}
\label{s:ldp-estimates}
In the previous section we established a large-deviation principle for the particle system~$Z^n$, which starts at initial positions $Z^n(0)$ drawn i.i.d.\ from some distribution $\nu\in \P(\spX)$. The particle system $Z^n$ is situated in the compressed (`$X$') setup. After mapping to the expanded setup, the evolutions $t \mapsto Z^n(t)$ are solutions of the `correct' SDE~\eqref{eq:SDE} (or Definition~\ref{def:particle-systems}\eqref{i:def:particle-systems:Y}). However, the expansion causes the initial data to be distributed in a  convoluted and unnatural way. 

In this section and the following two ones we therefore adapt the large-deviation principle for $Z^n$ of Theorem~\ref{th:LDP-dynamic-X} to  the more natural initial distribution of Theorem~\ref{th:LDP-path}. In this section we establish a number of estimates.

\bigskip

In the proof of Theorem~\ref{th:LDP-pathwise-weakversion} below of the large-deviation principle for $Y^n$, the initial data  $Y^n(0)$ will be distributed according to a version of the invariant measure $\InvMeas_n$ with $W=0$: 
\[
\InvMeasWiszero_n(dy) := \frac1{\cZ_n^{W=0}} \exp \biggl(\,-2 \sum_{i=1}^n V(y_i) \biggr) \, \Lebesgue^n\Big|_{\Omega_n}(dy).
\]
The compressed version of this measure is  (see Lemma~\ref{l:transformed-inv-meas})
\begin{equation}
\label{def:tildePn}
\widetilde \bbP_n (dx) := \frac1{\cZ_n^{W=0}} \exp \biggl(\,-2 \sum_{i=1}^n V(T_{\eta_n(x)}x_i) \biggr) \, \Lebesgue^n(dx).
\end{equation}

On the other hand, in the auxiliary particle system~$Z^n$ the initial positions $Z^n_i(0)$ will be i.i.d.\ distributed with common law $\xi := \bbQ^\nu$; recall from Lemma~\ref{lemma:properties-of-the-functionals} that for given $\nu\in \P(\spX)$ the single-particle measure $\bbQ^\nu\in \P(\spX)$ is defined as
\[
\bbQ^\nu(dx) := \frac{1}{\cZ^{\bbQ,\nu}} 
\exp\bigl[\,-2 V(T_{\nu}x)\bigr] \, dx, 
\qquad\text{where}\qquad
\cZ^{\bbQ,\nu} := \int_{\spX} \exp\bigl[\,-2  V(T_{\nu}x)\bigr] \, dx.
\]
Therefore the vector $Z^n(0)$ has as law the $n$-particle tensor product $\bbQ_n^\nu\in \P(\spX^n)$,
\begin{equation}
\label{def:Qnu_n}
\bbQ_n^{\nu}(dx_1\cdots dx_n)  := (\bbQ^\nu)^{\otimes n} = \frac1{(\cZ^{\bbQ,\nu})^n}
 \int_{\spX^n} \exp\biggl[\,-2 \sum_{i=1}^n V(T_{\nu}x_i)\biggr] \, dx_1\cdots dx_n
.
\end{equation}

The following lemma is the main result of this section, and establishes some estimates that connect these two particle systems; recall that the metric $d_{BL}$ on $\P(\R)$, appearing in part (3), is defined in duality with bounded Lipschitz functions (see Section~\ref{ss:narrow-convergence}).

\begin{lemma}[Basic estimates]
For $\nu\in \P(\spX)$, let $\bbQ_n^\nu\in \P(\spX^n)$ and $\widetilde \bbP_n$ be defined as above in~\eqref{def:tildePn} and~\eqref{def:Qnu_n}. Recall that the function $\gamma$ and the constant $C_\gamma$ were defined in Lemma~\ref{lemma:properties-of-the-functionals}.
\begin{enumerate}
\item \label{l:basic-est:Z}
We have
\begin{equation}
\label{eq:basic-est:Z:conv_rn}
r_n :=  \left|\frac1n\log {\cZ^{W=0}_n} - C_\gamma\,\right| \qquad\text{satisfies} \qquad r_{n}\stackrel{n\to\infty} \longrightarrow 0.
\end{equation}

\item\label{l:basic-est:PQ1}
For any $\nu\in \P(\R)$, 
\begin{equation}
\label{est:PQ1}
\sup_{n\geq 1} \left\|\frac1n \log \frac{d\widetilde \bbP_n}{d\bbQ_n^\nu}\right\|_{L^\infty(\spX^n)} <\infty.
\end{equation}

\item \label{l:basic-est:PQ2}
There exists a function $R:[0,\infty)\times \P(\spX)\to [0,\infty)$ such that 
for any $\delta>0$ and $\nu\in \P(\spX)$, 
\begin{equation}
\label{est:RN-deriv}
\left\| \frac1n \log \frac{d\widetilde \bbP_n}{d\bbQ_n^\nu} + \gamma(\nu)\right\|_{L^\infty(B_{n,\delta}(\nu))} \leq r_n + R(\delta,\nu).
\end{equation}
Here
\[
B_{n,\delta}(\nu) := \Bigl\{ x\in \R^n: d_{BL}(\eta_n(x),\nu)< \delta\Bigr\},
\]
and for all $\nu\in \P(\spX)$, $R(\cdot,\nu):[0,\infty)\to[0,\infty)$ is non-decreasing. If $\nu$ is Lebesgue-absolutely-continuous, then $\lim_{\delta\downarrow0} R(\delta,\nu) = 0$. 
%\item \label{l:basic-est:exp-tight}
%The sequence $(\eta_n)_{\#} \bbP_n$ is exponentially tight in the metric $d$. 
\end{enumerate}
\end{lemma}

\begin{proof}
We first show that $r_n$ is bounded as $n\to\infty$ for any $\nu$.
Fix $\nu\in \P(\spX)$. By the estimate~\eqref{bound:V-VTmu} there exists $C>0$ such that 
\begin{align*}
\frac1n\log \big(\cZ^{W=0}_n\big) &= \frac1n \log \int_{\R^n}  \exp\biggl[\,-2 \sum_{i=1}^n V(T_{\eta_n(x)}x_i) \biggr]\, dx\\
&\leq \frac1n \log \int_{\R^n}  \exp\biggl[\,-2 \sum_{i=1}^n V(x_i) \biggr]\, dx + 2C
\leq \log \cZ^{\bbQ,\nu} + 4C.
%= \int_\R \exp\bigl[-2V(x)\bigr]\, dx + C.
\end{align*}
In combination with the analogous estimate from the other side,
\[
\frac1n\log \big(\cZ^{W=0}_n\big) \geq  \log \cZ^{\bbQ,\nu} - 4C,
\]
this proves that $r_n$ is bounded.  It also follows  that there exists a subsequence $(n_k)_k$ and a constant $c\in \R$ such that
\begin{equation}
\label{conv:rnk-c}
\tilde r_{n_k} := \left|\frac1{n_k} \log \big(\cZ^{W=0}_{n_k}\big)-c\ \right| \;\stackrel{k\to\infty} \longrightarrow \;0.
\end{equation}
At the end of this proof we will show that $c=C_\gamma$, and therefore $r_n=\tilde r_n$,  and that the whole sequence converges. 

Part~\ref{l:basic-est:PQ1} will follow from part~\ref{l:basic-est:PQ2}, since we are able to take the function $R$ to be bounded. 
To show part~\ref{l:basic-est:PQ2}, take $\nu\in \P(\R)$ and estimate for $x\in \R^n$
\begin{align}
\notag
\left|\frac1{n_k} \log \frac{d\widetilde \bbP_{n_k}}{d\bbQ_{n_k}^\nu}(x) + c-\log \cZ^{\bbQ,\nu}\right|
&= \left| -\frac1{n_k} \log \cZ^{W=0}_{n_k} + c
 + \frac2{n_k}\sum_{i=1}^{n_k} \Bigl( V(T_\nu x_i) - V(T_{\eta_{n_k}(x)}x_i)\Bigr)\right|\\
\notag
&\leq \tilde r_{n_k} + \frac2{n_k} \Lip(V) \sum_{i=1}^{n_k} |T_\nu x_i - T_{\eta_{n_k}(x)}x_i|\\
&\leq \tilde r_{n_k} + 2\Lip(V) \,\sup_{\xi\in \R} \big| \nu((-\infty,\xi)) - (\eta_{n_k}(x))((-\infty,\xi))\big|.
\label{ineq:basic-est-RND-gamma}
\end{align}
If $\nu$ is not absolutely continuous, then we simply take $R(\delta,\nu) := 2\Lip(V)$, by which we satisfy the assertion of the Lemma. If $\nu$ is absolutely continuous, then by Lemma~\ref{l:Polya} below we can further estimate the right-hand side above by
\[
\leq \tilde r_{n_k} + 2\Lip(V) \,\omega_\nu(d_{BL}(\nu,\eta_{n_k}(x))).
\]
Setting $R(\delta,\nu) := 2\Lip(V)\, \omega_\nu(\delta)$, we now have proved the slightly modified version of~\eqref{est:RN-deriv},
\begin{equation}
\label{est:PQ-Linfty}
\left\| \frac1{n_k} \log \frac{d\widetilde \bbP_{n_k}}{d\bbQ_{n_k}^\nu} + c-\log \cZ^{\bbQ,\nu}\right\|_{L^\infty(B_{n,\delta}(\nu))} \leq \tilde r_{n_k} + R(\delta,\nu).
\end{equation}
 
\medskip
We now come back to the property $c=C_\gamma$, which we prove using Lemma~\ref{l:mean-field-LDP}. We set $\mathcal X := \P(\spX)$ with the bounded-Lipschitz metric and  $P_k := (\eta_{n_k})_\#\widetilde \bbP_{n_k}$. For $\nu\in \P(\spX)$ we set $Q_k^\nu := (\eta_{n_k})_\#\bbQ_{n_k}^\nu$; by Sanov's theorem $Q_k^\nu$ satisfies a (strong) large-deviation principle with good rate function $\mu\mapsto \RelEnt(\mu|\bbQ^\nu)$.

From~\eqref{est:PQ-Linfty} we deduce that for any $\nu \in \P(\spX)$, writing $B_\delta(\mu)$ for the $d_{BL}$-ball in $\P(\mathcal X)$, 
\begin{align}
\bigg|\frac1{n_k} &\log P_k(B_{\delta}(\nu))
- 
\frac1{n_k} \log Q_k^\nu(B_{\delta}(\nu))+\big[\gamma(\nu)-C_\gamma + c\big]
\bigg| = \notag\\
\label{est:PQ-ball}
&= \bigg|\frac1{n_k} \log \widetilde \bbP_{n_k}(B_{n_k,\delta}(\nu))
- 
\frac1{n_k} \log \bbQ_{n_k}^\nu(B_{n_k,\delta}(\nu))+\big[\gamma(\nu)-C_\gamma + c\big]
\bigg|\leq  \tilde r_{n_k} + R(\delta,\nu).
\end{align}
Note that if $\nu$ is such that $\RelEnt(\nu|\bbQ^\nu)<\infty$, then $\nu$ is Lebesgue absolutely continuous, and the right-hand side of~\eqref{est:PQ-ball} vanishes as $k\to\infty$ and then $\delta\to0$; and if $\RelEnt(\nu|\bbQ^\nu)=\infty$, then the right-hand side of~\eqref{est:PQ-ball} is bounded. Therefore the two conditions of Lemma~\ref{l:mean-field-LDP} are satisfied, and it follows that $P_k$ satisfies a weak large-deviation principle with rate function $\nu \mapsto \RelEnt(\nu|\bbQ^\nu) +\gamma(\nu)-C_\gamma + c$. Since the infimum of this is zero, we find $c=C_\gamma$. 
\end{proof}

The lemma below gives a quantitative version of a well-known result attributed to Poly\=a (see e.g.~\cite[Th.~9.1.4]{AthreyaLahiri06}).

\begin{lemma}[Quantitative Poly\=a Lemma]
\label{l:Polya}
Let $\nu\in \P(\R)$ be Lebesgue-absolutely continuous. There exists a non-decreasing function $\omega_\nu :[0,\infty)\to[0,\infty)$ with $\lim_{s\downarrow 0} \omega_\nu(s)=0$ such that 
\[
\forall \overline \nu\in \P(\R): \qquad
\sup_{x\in \R} |\nu\big((-\infty,x)\big) -\overline\nu\big((-\infty,x)\big)|\leq \omega_\nu\big(d_{BL}(\nu,\overline\nu)\big),
\]
where $d_{BL}$ is again the dual bounded-Lipschitz metric.
\end{lemma}

\begin{proof}
Write $F(x) := \nu\big((-\infty,x)\big)$ and $\overline F(x) := \overline\nu\big((-\infty,x)\big)$. Since $F$ is bounded and non-decreasing, it is uniformly continuous on $\R$; we write $\alpha$ for the modulus of continuity of~$F$, and we assume without loss of generality that $\alpha$ is non-decreasing. 

Fix $x_0$ and set $\e := \overline F(x_0) - F(x_0)$; for definiteness we assume that $\e>0$.   Since $\overline F$ is non-decreasing and $F$ has modulus of continuity $\alpha$, we estimate for $x\geq x_0$ that  
\[
\overline F(x)-F(x)\geq \overline F(x_0) - (F(x_0)+\alpha(x-x_0)) = \e-\alpha(x-x_0).
\]

Let $\delta_\e := \sup\{0< \delta\leq 1: \alpha(\delta) \leq  \e\}$, and define $\varphi_\e:\R\to\R$ by 
\[
\varphi_\e(x) := \begin{cases}
0 & x\leq x_0\\
x-x_0 & x_0\leq x \leq x_0+\delta_\e\\
\delta_\e & x_0+ \delta_\e \leq x.
\end{cases}
\]
We then calculate
\begin{align*}
-\int_\R \varphi_\e(x)(\overline\nu-\nu)(dx) 
&= \int_\R \varphi_\e'(x) (\overline F(x)-F(x))\, dx
= \int_{x_0}^{x_0+{\delta_\e}} (\overline F(x)-F(x))\, dx\\
&\geq \int_0^{\delta_\e} (\e-\alpha(y))\, dy =: \hat \alpha(\e),
\end{align*}
from which we deduce that 
\[
d_{BL}(\nu,\overline\nu) = \sup_{\varphi\in \BL(\R)} \|\varphi\|_{\BL}^{-1} \int_\R \varphi (\overline\nu-\nu)
\geq \frac{\hat\alpha(\e)}{\delta_\e + 1}
\geq \frac12 \hat\alpha(\e).
\] 
Taking the supremum over $x_0\in\R$ and inverting the relationship above we find
\[
\|\overline F-F\|_{L^\infty(\R)} \leq \omega_\nu(d_{BL}(\nu,\overline\nu)) 
\qquad\text{for}\qquad \omega_\nu(d) := \sup\{\e: \hat \alpha(\e)\leq 2d\}.
\]
Since $\hat \alpha(\e)$ is strictly positive for $\e>0$, $\lim_{d\downarrow 0} \omega_\nu(d) = 0$, implying that $\omega_\nu$ is a \emph{bona fides} modulus of continuity.
\end{proof}

\newpage

\section{Pathwise large deviations for \texorpdfstring{$Y^n$}{Yn} with special initial data}
\label{s:ldp-Y-special}

In this section we prove Theorem~\ref{th:LDP-pathwise-weakversion} below, which is a slightly weaker version of Theorem~\ref{th:LDP-path}. The difference  lies in the initial data $Y^n(0)$, which are not distributed by the measure $\InvMeasf_n$ as in Theorem~\ref{th:LDP-path}, but according to the ``$(W=0)$--version'' of the invariant measure $\InvMeas_n$ that we introduced in the previous section:
\[
\InvMeasWiszero_n := \frac1{\cZ^{W=0}_n} \exp \biggl(\,-2 \sum_{i=1}^n V(y_i) \biggr) \, \Lebesgue^n\Big|_{\Omega_n}.
\]

\begin{theorem}
\label{th:LDP-pathwise-weakversion}
Assume that $V,W$ satisfy Assumption~\ref{ass:VW}. 

For each $n$, let the particle system $t\mapsto Y^n(t)$ be given by Definition~\ref{def:particle-systems}(\ref{i:def:particle-systems:Y}), with initial positions $Y^n(0)$ drawn from the $W=0$ invariant measure $\InvMeasWiszero_n$. 

The random evolutions $t\mapsto \rho_n(t)=\eta_n(Y^n(t))$ then satisfy a large-deviation principle in $C\bigl([0,T];\P(\R)\bigr)$ with good rate function $\hat I$. If in addition $\rho(0)\in\P_2(\R)$ and $\rho$ satisfies $\hFreeEnergy(\rho(0)) + \hat I(\rho) < \infty$, then $\rho\in AC^2([0,T];\P_2(\R))$ and the functional $\hat I(\rho)$ can be characterized as 
\begin{equation}
\label{def:RF-ind-part}
\hat I(\rho) := 2\,\hEnt_V(\rho(0)) +  \hFreeEnergy( \rho(T))-  \hFreeEnergy(\rho(0))
+\frac12\int_0^T |\dot {\rho}|^2(t)\, dt  
+ \frac12\int_0^T |\partial \hFreeEnergy|^2(\rho(t))\, dt.
\end{equation}
\end{theorem}
Note that although the initial data $Y^n(0)$ are drawn from $\InvMeasWiszero_n$, the processes $Y^n$ evolve under the dynamics that includes $W$. 

\subsection{First part of the proof of Theorem~\ref{th:LDP-pathwise-weakversion}.}
As in the statement of the theorem, consider for each $n$ initial data $Y^n(0)$ drawn from $\InvMeasWiszero_n$. We transform these positions to initial data for the $X^n$-process, through 
\begin{equation}
\label{eq:th-ldp-path-initial-data-weak}
\eta_n(X^n(0)) = A_n^{-1} \eta_n(Y^n(0)).
\end{equation}
This identity only fixes the positions $X^n_i(0)$ up to permutation of $i$; this arbitrariness has no impact, however,  since all our results only depend on $\eta_n(X^n)$, which is invariant under such permutations.

By Lemma~\ref{l:transformed-inv-meas} the transformed initial data $X^n(0)$ have law
\[
\bbP_n(dx) := \frac1{\cZ^{W=0}_n} 
\exp\biggl[\,-2 \sum_{i=1}^n V(T_{\eta_n(x)}x_i)\biggr]\, dx.
\]
Let $t\mapsto X^n(t)$ solve the system of Definition~\ref{def:particle-systems}(\ref{i:def:particle-systems:X}) with initial data $X^n(0)$. 
\medskip

The following lemma uses the result of Section~\ref{s:ldp-iid} to give a large-deviation principle for this  particle system $X^n$.

\begin{lemma}
\label{l:LDP-W=0-initial-data}
Define the random time-dependent measures $\mu_n\in C([0,T];\P(\spX))$ by $\mu_n(t) = \eta_n(X^n(t))$. The sequence $\mu_n$  satisfies a large-deviation principle in the space $C([0,T];\P(\spX))$ with good rate function 
\begin{equation}
\label{def:LDP-RF-weak-X}
I(\mu) := \mathfrak I_{\bbQ^{\mu_0}}(\mu) + \gamma(\mu_0),
\end{equation}
where $\mathfrak I_\nu(\cdot)$ is defined in~\eqref{def:I-time-dependent}, $\bbQ^\nu$ in~\eqref{def:Qnu},  and $\gamma$  in~\eqref{def:gamma}. 
\end{lemma}

\begin{proof}[Proof of Lemma~\ref{l:LDP-W=0-initial-data}]
%Recall that $d_{BL}$ is the dual bounded-Lipschitz metric on $\P(\R)$; for $\xi\in \P(\R)$ we write $B_\delta(\xi)$ for the ball in $\P(\R)$ with radius $\delta$, and with a small abuse of notation we also write $B_\delta(\xi)$ for the ball in $C([0,T];\P(\R))$ around $\xi\in C([0,T];\P(\R))$ in the combined supremum-$d_{BL}$--metric.
We use Lemma~\ref{l:mean-field-LDP} to prove that the random measures $\mu_n$ satisfy a weak large-deviation principle; subsequently we upgrade this into a strong large-deviation principle by showing exponential tightness. 

\medskip

We write $\underline \bbP_n$ for the law of $X^n$ on $C([0,T];\spX^n)$. Set $\mathcal X = C([0,T];\P(\spX))$ and define the modified push-forward 
\[
P_n := (\eta_n)_\% \underline \bbP_n\in \P(\mathcal X)
\qquad
\text{by}\quad
P_n(A) := \underline \bbP_n \Big\{ \big(\eta_n(X^n(t))\big)_{t\in[0,T]}  \in A\Big\}.
\]

We construct the measure $Q_n$ as follows. Fix some $\nu\in \mathcal X$. Let the curves $Z^n$ solve the same equation in Definition~\ref{def:particle-systems}(\ref{i:def:particle-systems:X}) as $X^n$, but with initial data $Z^n(0)$ drawn from the independent measure $\bbQ^{\nu_0}_n = (\bbQ^{\nu_0})^{\otimes n}$ (see~\eqref{def:Qnu_n}). We write $\underline \bbQ^{\nu_0}_n$ for the law of $Z^n$ on $C([0,T];\spX^n)$, and we set $Q_k^\nu := (\eta_n)_\%\underline \bbQ^{\nu_0}_n\in\P(\mathcal X)$. 
By Theorem~\ref{th:LDP-dynamic-X}, the sequence of measures $Q_n^\nu$   satisfies a strong large-deviation principle with the good rate function $\mathfrak I_{\bbQ^{\nu_0}}$.

%We estimate
%\begin{align*}
%\biggl|\frac1n &\log P_n(B_\delta(\mu)) - \frac1n \log Q_n^\mu(B_\delta(\mu)) +\gamma(\mu_0)\biggr| = 
%&= 

Since the laws of $X^n$ and $Z^n$ are the same after  conditioning on the initial positions, we have for each $x\in C([0,T];\spX^n)$ that 
\begin{equation}
\label{id:RN-deriv-pathwise-initial}
\frac{d\underline \bbP_n}{d\underline \bbQ_n^{\nu_0}}(x) = 
\frac{d\bbP_n}{d\bbQ^{\nu_0}}(x(0))
%\qquad \text{and therefore}\qquad
%\frac{dP_n}{dQ_n} (\mu) = 
.
\end{equation}
We now estimate 
\begin{align}
\frac1n  \log &P_n(B_\delta(\mu))
= \frac1n \log \Expectation_{\underline \bbP_n} 
  \Bigl[\Indicator \Bigl\{\bigl(\eta_n(x(t))\bigr)_{t\in[0,T]} \in B_\delta(\mu) \Bigr\}\Bigr]\notag\\
&= \frac1n \log \Expectation_{\underline \bbQ_n^{\mu_0}}
  \Bigl[\frac{d\underline \bbP_n}{d\underline \bbQ_n^{\mu_0}} \Indicator \Bigl\{\bigl(\eta_n(x(t))\bigr)_{t\in[0,T]} \in B_\delta(\mu) \Bigr\}\Bigr]\notag\\
&\leftstackrel{\eqref{id:RN-deriv-pathwise-initial}}\leq \frac1n \log \Expectation_{\underline \bbQ_n^{\mu_0}} 
  \Bigl[ \Indicator \Bigl\{\bigl(\eta_n(x(t))\bigr)_{t\in[0,T]} \in B_\delta(\mu) \Bigr\}\Bigr] - \gamma(\mu_0)
  +\sup_{\eta_n^{-1}(B_\delta(\mu_0))}
    \biggl[ \frac1n \log \frac{d\bbP_n}{d\bbQ^{\mu_0}_n} + \gamma(\mu_0)\biggr]
    \label{ineq:LDP-pathwise-weak-DZ}\\
&\leftstackrel{\eqref{est:RN-deriv}}\leq \frac1n \log Q_n^{\mu} (B_\delta(\mu)) - \gamma(\mu_0)
  + r_n + R(\delta,\mu_0).
  \notag
\end{align}
By combining with the opposite inequality we find
\[
\biggl|\frac1n  \log P_n(B_\delta(\mu)) - \frac1n \log Q_n^{\mu} (B_\delta(\mu)) + \gamma(\mu_0)\biggr| 
\leq r_n + R(\delta,\mu_0).
\]
The properties of $r_n$ and $R$ now imply that the conditions of Lemma~\ref{l:mean-field-LDP} are satisfied, and it follows that $P_n$ satisfies a weak large-deviation principle with good rate function $\mu\mapsto \mathfrak I_{\bbQ^{\mu_0}}(\mu) + \gamma(\mu_0)$.

\medskip
Finally, to show that the weak large-deviation principle is in fact a strong principle, take an arbitrary $\nu\in C([0,T];\P(\spX))$ and construct the particle system $Z^n$ as above. By Theorem~\ref{th:LDP-dynamic-X} the random variables $t\mapsto \eta_n(Z^n(t))$ are exponentially tight in $C([0,T];\P(\spX))$; by~\eqref{id:RN-deriv-pathwise-initial} and the bound~\eqref{est:PQ1} the same follows for $t\mapsto \eta_n(X^n(t))$. 
\end{proof}

\subsection{Characterization of the rate function}

\begin{lemma}
\label{l:char-RF-path}
Let $\nu\in \P(\spX)$ and let $\mu\in C([0,T];\P(\spX))$ satisfy $\mathfrak I_{\nu}(\mu)<\infty$. 
%and define $\mathbb W^\rho$ by~\eqref{def:W^rho}. Let $P\in \P(C([0,T];\spX)$ satisfy $\RelEnt(P|\mathbb W^\rho)<\infty$, and write $P_t\in \P(\spX)$ for the single-time marginal at time~$t$. 
Then there exists a measurable function $u:[0,T]\times \spX\to\R$, such that 
\begin{equation}
\label{eq:I-in-terms-of-u}
\mathfrak I_\nu(\mu) = \RelEnt(\mu_0|\nu) + \frac12 \int_0^T \int_{\spX} u^2(t,x)\, \mu_t(dx) dt,
\end{equation}
and $\mu$ is a distributional solution of 
\begin{equation}
\label{eq:FP-u}
\partial_t\mu_t = \frac12 \partial_{xx}\mu_t - \partial_x \bigl[ (b(\cdot,\mu_t) + u(t,\cdot))\mu_t\bigr].
\end{equation} 
The function $u$ is unique in $L^2(0,T;L^2_{\mu_t}(\spX))$. 
%We have the alternative formulation
%\[
%
%\red{   adapted \red{Need to discuss filtrations} $\R$-valued process $\overline u = \overline u(t,\omega)$ such that 
%\[
%\frac{d P}{d\mathbb W^\rho} ((\omega_t)_{0\leq t\leq T}) = \Indicator_{\frac{dP}{d\mathbb W^\rho}>0}
%\exp\biggl(\int_0^T  u_t \, dB^P_t + \int_0^T |\overline u_t|^2 \, dt\biggr).
%\]
%}
\end{lemma}

\begin{proof}
%\red{It's possible that this existence proof is contained in Jasper-Mario-Oliver's result.}
We begin by showing existence and uniqueness in law of weak solutions to the SDE 
\begin{align}
\label{eq:SDE-single-particle}
	dX_t = b(X_t, \mu _t) dt + dB_t,
\end{align}
with initial positions distributed according to $\nu$. That is, we establish existence and uniqueness of the measure $\mathbb W_\nu ^{\mu}$ such that under $\mathbb W_\nu ^{\mu}$, and with respect to some filtration $\{ \mathcal F _t \} _{t \in [0,T]}$, $\{ B_T \} _{t\in [0,T]}$ is a Brownian motion and $\mathbb W_\nu ^{\mu}$ is the probability law of the solution $\{ X_t \} _{t \in [0,T]}$ of the SDE. To show this we first prove that $b$ satisfies, for any $x \in \R$ and $\rho, \eta \in \mathcal P (\R)$, 
\begin{align*}
	|b(x, \rho) - b(x, \eta)| \leq C _d{TV}(\rho, \eta),
\end{align*}
where $d_{TV}(\cdot, \cdot)$ is the total variation distance, and 
\begin{align*}
	|b(x, \rho)| \leq \tilde C (1 + |x|).
\end{align*}
With these estimates the assumptions of \cite{DjehicheHamadene18} are fulfilled and existence and uniqueness of weak solutions to the SDE hold.

From the inequality 
\[
|T_\rho (x) - T_\eta(x)| = \bigl|x + \alpha\rho((-\infty, x)) -x - \alpha\eta ((-\infty, x))\bigr|
\leq \alpha d_{TV}(\rho,\eta) \qquad\text{for all }x,
\]
we obtain with   the  assumption that $V'$ is Lipschitz the estimate
\begin{align*}
|V' (T_{\rho} (x)) - V'(T_{\eta} (x))| &\leq \alpha C_{V'} d_{TV} (\rho, \eta).
\end{align*}
For $W$ we split the difference as follows:
	\begin{align*}
	& \int _{\R _x} W' \left( T_{\rho} (x) - T_{\rho} (\tilde x) \right) \rho(d\tilde x)  - \int _{\R _x} W' \left( T_{\eta} (x) - T_{\eta} (\tilde x) \right) \eta (d\tilde x) \\
	& \quad = \int _{\R _x} \left[ W' \left( T_{\rho} (x) - T_{\rho} (\tilde x) \right) - W' \left( T_{\eta} (x) - T_{\eta} (\tilde x) \right) \right]\rho (d\tilde x) \\
	& \qquad + \int _{\R _x}  W' \left( T_{\eta} (x) - T_{\eta} (\tilde x) \right) \left( \rho (d \tilde x) - \eta (d \tilde x) \right).
\end{align*}
Because $W'$ is Lipschitz, the first term on the right-hand side can similarly be bounded:
\begin{align*}
	& \int _{\R _x} | W' \left( T_{\rho} (x) - T_{\rho} (\tilde x) \right) - W' \left( T_{\eta} (x) - T_{\eta} (\tilde x) \right) | \rho (d\tilde x) \\
	&  \quad \leq \int _{\R _x} C_{W'} | T_{\rho} (x) - T_{\rho} (\tilde x)  - T_{\eta} (x) + T_{\eta} (\tilde x)| \rho (d\tilde x) \\
	& \quad \leq \int _{\R _x} C_{W'} \left (| T_{\rho} (x) - T_{\eta} ( x) | + |T_{\rho} (\tilde x) - T_{\eta} (\tilde x)| \right) \rho(d\tilde x) \\
	& \quad \leq \alpha C_{W'} 2 d _{TV} (\rho, \eta).
\end{align*}
The second term, involving the integral with respect to the difference $\rho - \eta$, can be bounded from above using the characterization of the total variation distance in terms of Borel measurable functions $f\in \mathcal B(\spX)$:
	\begin{align*}
		& \Big| \int _{\R _x}  W' \left( T_{\rho} (x) - T_{\eta} (\tilde x) \right) \left( \rho (d \tilde x) - \eta (d \tilde x) \right) \Big| \\
		& \quad =  \lVert W' \rVert _{\infty} \Big| \int _{\R _x}  \frac{1}{\lVert W' \rVert _{\infty}} W' \left( T_{\eta} (x) - T_{\eta} (\tilde x) \right) \left( \rho (d \tilde x) - \eta (d \tilde x) \right) \Big|\\
		& \quad \leq \lVert W' \rVert _{\infty} \sup _{f \in \mathcal B (\spX), \lVert f \rVert _{\infty} \leq 1} \Big| \int _{\R} f(\tilde x) \left( \rho (d\tilde x) - \eta(d\tilde x) \right) \Big| \\
		& \quad = \lVert W' \rVert _{\infty} \,d_{TV} (\rho, \eta).
	\end{align*}
	Together the two upper bounds yield
	\begin{align*}
		& \Big|\int _{\R _x} W' \left( T_{\rho} (x) - T_{\rho} (\tilde x) \right) \rho (d\tilde x) - \int _{\R _x} W' \left( T_{\eta} (x) - T_{\eta} (\tilde x) \right) \eta (d\tilde x) \Big| \leq \tilde C _{W'} d_{TV} (\rho, \eta).
	\end{align*} 
	Combining this with the upper bound for the difference of $V'$-terms, we have
	\begin{align*}
		| b(x, \rho) - b(x, \eta) | &= \biggl|  V' (T_{\rho} (x)) + \int _{\R _x} W' \left( T_{\rho} (x) - T_{\rho} (\tilde x) \right) \rho (d\tilde x) \\
		&\qquad - V'(T_{\eta} (x)) - \int _{\R _x} W' \left( T_{\eta} (x) - T_{\eta} (\tilde x) \right) \eta (d\tilde x)\biggr| \\
		& \leq C d_{TV} (\rho, \eta),
	\end{align*}
	for some constant $C$.
	
	The linear growth condition follows from the assumption that $V'$ and $W'$ are Lipschitz and bounded, respectively. 
	
	With these estimates, the conditions of \cite{DjehicheHamadene18} are satisfied, implying that weak existence and uniqueness of solutions holds for the SDE~\eqref{eq:SDE-single-particle}; therefore the measure $\mathbb W_\nu^\mu$ is well-defined.

	\medskip
	Define the set
	\begin{align*}
		\mathcal A_\nu^\mu = \{  P \in \mathcal P (\R): \RelEnt (P | \mathbb W_\nu ^{\mu}) < \infty, \ P_t = \mu _t \ \forall t\},
	\end{align*}
	so that 
	\[
	\mathfrak I_\nu (\mu) = \inf_{A_\nu^\mu} \RelEnt(P|\mathbb W_\nu^\mu).
	\]
	Since by assumption $\mathfrak I_\nu(\mu)<\infty$, the set $\mathcal A_\nu^{\mu}$ is non-empty. By Theorem 3.1 of \cite{CattiauxLeonard94}, in the definition of $\mathfrak J _{\nu}(\mu)$ it is sufficient to minimize over (strongly) Markovian $P$ such that $P \in \mathcal A _\nu^{\mu}$. %satisfying $$$P_t = \rho _t$ for all $t$. 

Uniqueness in law corresponds to the uniqueness condition `U' in \cite{Leonard12} and by Theorems~1 and~2 therein, for each Markovian $P \in \mathcal A _\nu^{\mu}$ there exists a process $\{ \beta _t \} _{t \in [0,T]}$ adapted to the (augmented version of the) filtration $\{ \mathcal F_t \} _{t \in [0,T]}$ of the weak solution such that $\int _0 ^T \beta_t dt$ and $\int _0 ^T \beta _t ^2 dt$ are both finite $P$-a.s.\ 
	%\begin{align*}
	%	P \left( \int _0 ^T | \beta _t | ^2 dt < \infty \right) = 1
	%\end{align*}
	 and under $P$ there is a $P$-Brownian motion $B^P$ such that the process $\{ X_t \} _{t\in [0,T]}$ solves, under $P$,
	  \begin{align*}
	 	dX_t = (b(X_t, \mu _t) + \beta _t)dt + dB ^P _t.
	 \end{align*} 
	
	 By \cite[Thm 3.60]{CattiauxLeonard94} there is a $u: [0,T] \times \R _x \to \R$ such that $\int _0 ^T \int _{\R _x} u^2 (t,x) \mu _ t (dx) dt < \infty$ and the process $\beta$ can be expressed as $\beta _t = u(t, X_t)$; the function $u$ can be obtained via the Riesz representation theorem, see \cite{CattiauxLeonard94} for details. %Moreover, for each such $u \in L^2 _{\rho}$ there can be at most one $P \in A_{\rho}$ (see \cite[p.110]{CattiauxLeonard94} and references therein). Together these results yield the uniqueness of $u$. Moreover, it holds that
	 The Radon-Nikodym derivative between $P$ and $\mathbb W_\nu ^{\mu}$ satisfies
	 \begin{align*}
	 	\frac{dP}{d\mathbb W_\nu ^{\mu}} = \frac{d\mu_0}{d\nu} \exp \left\{ \int _0 ^T u(t, X_t) dB^P _t + \frac{1}{2} \int _0 ^T | u(t, X_t) | ^2 dt \right\},
	 \end{align*}
	 and it follows that 
	 \begin{align*}
	 	\RelEnt (P | \mathbb W_\nu ^{\mu}) = \RelEnt (\mu_0 | \nu) + \frac{1}{2} \Expectation _P \int _0 ^T |u(t, X_t) | ^2 dt. 
	 \end{align*}
	 This is precisely~\eqref{eq:I-in-terms-of-u}. 
	 
	 Replacing $\beta$ with $u$ in the SDE, we find that under $P$ the process $X$ solves
	 \begin{align*}
	 	dX_t = (b(X_t, \mu _t) + u(t, X_t))dt + dB^P _t.
	 \end{align*}
	 The Forward-Kolmogorov equation of this SDE for the single-time marginals $P_t=\mu_t$ is equal to~\eqref{eq:FP-u}. The uniqueness of $u$ is a direct consequence of the strict convexity of $\iint u^2\mu$ and the linear constraint~\eqref{eq:FP-u}.
\end{proof}

\begin{lemma}
\label{lemma:I-to-v-formulation}
Let $\nu\in \P(\spX)$  and let $\mu\in C([0,T];\P(\spX))$ satisfy $\mu_0\in \P_2(\spX)$ and {$\mathcal F(\mu_0) + \mathfrak I_{\nu}(\mu)<\infty$}. Then
\begin{enumerate}
\item\label{i:l:i-to-v:AC2}$\mu\in AC^2([0,T];\P_2(\spX))$. 
\item\label{i:l:i-to-v:RF} For almost all $t\in [0,T]$, $\mu_t$ is Lebesgue-absolutely-continuous, $\partial_x \mu_t \in L^1(\spX)$, and 
\[
\int_0^T \int_{\spX} \frac{|\partial_x \mu_t(x)|^2}{\mu_t(x)}\, dx dt < 
\infty
%\mathcal F(\rho_0) + I(\rho)
.
\]
\item \label{i:l:i-to-v:char-I} The  functional $\mathfrak I_{\nu}$ can be written as 
\begin{equation}
\label{eq:I-quadratic-form}
\mathfrak I_{\nu}(\mu) = \RelEnt(\mu_0|\nu) + \frac12 \int_0^T \int_{\spX} 
  \mu_t(x) \Bigl(v_t(x) + \frac{\partial_x \mu_t}{2\mu_t}(x) - b(x,\mu_t)\Bigr)^2\, dxdt,
\end{equation}
where $v_t$ is the velocity field associated with $\partial_t \mu_t$ (see Lemma~\ref{lemma:char-AC2}).
\end{enumerate}
\end{lemma}

\begin{proof}
We first show that $\int_{\R} x^2\mu_t(dx)$ is bounded uniformly in $t$. Formally this follows from multiplying equation~\eqref{eq:FP-u} by $x^2$ and integrating, giving
\begin{align*}
\frac d{dt} \frac12 \int_\R x^2 \mu_t(dx) &= \int_\R \mu_t\Bigl[ 1 + x(b_t+u_t)\Bigr]\\
&\leq 1 + \int_\R x^2\mu_t(dx) + \frac12 \int_\R b_t^2 \mu_t(dx) + \frac12 \int_\R u_t^2\mu_t(dx)\\
&\leq 1 + \int_\R x^2\mu_t(dx) + \frac12 \int_\R C(1+x^2)\mu_t(dx) + \mathfrak I_{\nu}(\mu).
\end{align*}
The second inequality follows from the assumptions~\eqref{cond:coercivity} on $V$. Gronwall's Lemma then yields boundedness of $\int x^2\mu_t(dx)$  for finite time. This argument can be made rigorous by approximating $x^2$ by smooth compactly supported functions. 

We next prove part~\ref{i:l:i-to-v:RF}.  Fix a function $\varphi\in C_c^{0,1}([0,T]\times \R)$,  and let $g\in C^{1,2}([0,T]\times \R)$ be a solution of the linear backward-parabolic equation
\begin{alignat*}2
&\partial_t g_t +\frac12 \partial_{xx} g_t + b_t\partial_x g_t = \frac18 g_t \Bigl(\frac12\varphi_t^2  - \partial_{x}\varphi_t\Bigr), & \qquad &(x,t)\in \R\times (0,T],\\
& g_T = 1, && x\in \R,\\
&g \text{ is bounded,}
\end{alignat*}
where we use the shorthand notation $b_t(x) := b(x,\mu_t)$. Existence of such a solution follows from standard linear parabolic theory; see e.g.~\cite[Th.~1.12]{Friedman64}. The constant initial datum and the compact support of the right-hand side imply that $g(x,t)\to 1$ for $x\to\pm\infty$ and for all $t$, and that all derivatives converge to zero as $x\to\pm\infty$; this can be recognized in the representation formula~\cite[(1.7.6)]{Friedman64}. 

Note that by the coerciveness bound~\eqref{cond:coercivity} on $V$ we have $\int_{\spX} g_T^2(x) e^{-2V(x)}\, dx = \int_{\spX}e^{-2V(x)}\, dx < \infty$. To obtain bounds on the same expression at earlier times $t$ we calculate, briefly suppressing the subscript $t$, 
\begin{align*}
 \frac d{dt} \int_{\spX} g_t^2(x) e^{-2V(x)}\, dx
&= \int ge^{-2V} \biggl(- \partial_{xx}g - 2b \partial_x g +\frac14g\Bigl( \frac12\varphi^2  -  \partial_x\varphi\Bigr)\biggr)\\
&= \int e^{-2V} \biggl(|\partial_{x}g|^2  -2g\partial_x g V'- 2b g \partial_x g +\frac18  g^2 \varphi^2 +\frac12  g\partial_x g \varphi - \frac12g^2\varphi V'\biggr)\\
&=\int e^{-2V}\biggl(\frac12 |\partial_{x}g - gV' +\frac12 g \varphi |^2  
  + \frac12  \bigl|\partial_x g - {2gb} - gV'\bigr|^2  
  - g^2(2b^2 +2bV' + {V'}^2) \biggr)\\
&\geq - 3(\|b\|_\infty^2 + \|V'\|_\infty^2) \int g_t^2 e^{-2V},
\end{align*}
after which Gronwall's Lemma  implies that 
$
\int g_0^2 e^{-2V}\leq C,
$
with a constant $C>0$ that does not depend on $\varphi$. 

The function $f := 2\log g$ then satisfies
\[
\partial_t f_t + \frac12 \partial_{xx} f_t + b_t \partial_x f_t + \frac14 |\partial_x f_t|^2 = \frac14  \Bigl(\frac12 \varphi^2_t-\partial_{x}\varphi_t \Bigr),
\]
with  final value $f_T = 0$.
Multiplying~\eqref{eq:FP-u} with  $f$ and integrating we find
\begin{align*}
0 &= \int_{\spX} \Bigl[ f_T\mu_T - f_0\mu_0\Bigr] 
- \int_0^T \int_{\spX} \mu_t \Bigl[\partial_t f_t + \frac12\partial_{xx} f_t + b_t \partial_x f_t \Bigr]\\
&= - \int_{\spX} f_0\mu_0 +  \int_0^T \int_{\spX} \mu_t \Bigl(u_t\partial_x f_t -\frac14 |\partial_x f_t|^2 - \frac14 \partial_{x}\varphi_t +\frac1{8}\varphi_t^2 \Bigr).
\end{align*}
Briefly writing $\mu_V(dx) := Z_V^{-1}e^{-2V}dx$, for which we have the estimate $\RelEnt(\mu_0|\mu_V)\leq 2\FreeEnergy(\mu_0) + C$, we then apply the entropy inequality to find
\begin{align*}
\frac14 \int_0^T \int_{\spX} \mu_t \Bigl(\partial_{x}\varphi_t -\frac12\varphi_t^2\Bigr)
&\leq  \int_0^T \int_{\spX} \mu_t \Bigl(u_t\partial_x f_t -\frac14 |\partial_x f_t|^2 \Bigr) - \int_{\spX} f_0\mu_0\\
&\leq \int_0^T \int_{\spX}\mu_t |u_t|^2 
+ \RelEnt(\mu_0|\mu_V) + \frac1{Z_V}\int_{\spX} e^{f_0-2V}\\
&\leq 2I(\mu) + 2\FreeEnergy(\mu_0) + C + \frac1{Z_V}\int_{\spX} g_0^2e^{-2V}.
\end{align*}
The right-hand side of this expression is bounded from above independently of $\varphi$, and by the dual characterization of Fisher Information (see e.g.~\cite[Lemma~D.44]{FengKurtz06}) it follows that 
\begin{align*}
\frac12 \int_0^T \int_{\spX} \frac{|\partial_x \mu_t(x)|^2}{\mu_t(x)}\, dx dt
&= \int_0^T \sup_{\psi\in C_c^1(\R)} \int_{\spX} \mu_t \Bigl(\partial_{x}\psi -\frac12\psi^2\Bigr)\\
&= \sup_{\varphi\in C_c^{0,1}([0,T]\times \R)}\int_0^T\int_{\spX} \mu_t \Bigl(\partial_{x}\varphi_t -\frac12\varphi_t^2\Bigr)\\
&<\infty,
\end{align*}
where the second identity follows from a standard argument involving the separability in $C_b(\R)$ of the subspace $C_c^1(\R)$. This proves part~\ref{i:l:i-to-v:RF}.

\medskip
To prove part~\ref{i:l:i-to-v:AC2}, i.e.\ to show that $\mu\in AC^2([0,T];\P_2)$, we write equation~\eqref{eq:FP-u} as
\begin{equation}
\label{eq:FP-u-with-v}
\partial_t \mu_t = -\partial_x (\mu_tv_t) \quad\text{with}\quad
v_t = -\frac12 \frac{\partial_x \mu_t}{\mu_t} + b_t + u_t.
\end{equation}
The function $v_t$ satisfies $\int_0^T\int_{\spX}\mu_t v_t^2 <\infty$ because each of the three terms in~\eqref{eq:FP-u-with-v} satisfies a similar bound: $u$ satisfies this property by~\eqref{eq:I-in-terms-of-u}, $b$ is uniformly bounded by Assumption~\ref{ass:VW}, and for $\partial_x\mu/\mu$ this follows from part~\ref{i:l:i-to-v:RF}. By the characterization of Lemma~\ref{lemma:char-AC2} it follows that $\mu\in AC^2([0,T];\P_2)$.

Finally, to show part~\ref{i:l:i-to-v:char-I} it suffices to substitute~\eqref{eq:FP-u-with-v} in~\eqref{eq:I-in-terms-of-u}.
\end{proof}

\subsection{End of proof of Theorem~\ref{th:LDP-pathwise-weakversion}.}
We now have constructed two particle systems as follows:
\begin{itemize}
\item The  particle system $Y^n$ of Definition~\ref{def:particle-systems}\eqref{i:def:particle-systems:Y} is  started at  initial positions drawn from $\InvMeasWiszero_n$;
\item The particle system $X^n$ of Definition~\ref{def:particle-systems}\eqref{i:def:particle-systems:X} is started from the transformed initial positions $X^n_i(0)$,  as described in~\eqref{eq:th-ldp-path-initial-data-weak}.
\end{itemize}
By Lemma~\ref{l:LDP-W=0-initial-data}, the  random time-dependent measures $\mu_n(t) = \eta_n(X^n(t))$ satisfy a large-deviation principle in $C([0,T];\P(\spX))$ with rate function~$I(\mu) =  \mathfrak I_{\bbQ^{\mu_0}}(\mu) + \gamma(\mu_0)$.

By Lemma~\ref{lemma:mapping-particle-systems}  the  random  measures $\rho_n(t) = \eta_n(Y^n(t))$ have the same distribution as $A_n\mu_n(t)$. We deduce the LDP for $\rho_n$ by applying a generalization of the contraction principle to $n$-dependent maps in the form of~\cite[Corollary~4.2.21]{DemboZeitouni98}.  In the case at hand  these maps will be the maps $\mathcal A^{-1}_n$ and $\mathcal A^{-1}$, which extend $A$ and $A_n$ to time-dependent measures:
\begin{align*}
&\mathcal A_n: C([0,T];\P_n(\spX)) \to C([0,T];\P_n(\spY)), \quad & (\mathcal A_n\mu)(t) &:= A_n(\mu(t)),\\
&\mathcal A: C([0,T];\P(\spX)) \to C([0,T];\P(\spY)), \quad & (\mathcal A\mu)(t) &:= A(\mu(t)).\\
\end{align*}

The only non-trivial condition to check for~\cite[Corollary~4.2.21]{DemboZeitouni98} is a convergence property of the maps $\mathcal A_n$ to $\mathcal A$. We temporarily write  $\bbP_{\mu_n}$ for the law of $\mu_n$ on $ \P(C([0,T];\spX))$.  Define for $n\in\N$ and $\delta>0$ 
\[
\Gamma_{n,\delta} := \{\nu\in \supp \bbP_{\mu_n}: d_{BL}(\mathcal A\nu,\mathcal A_n\nu) > \delta\},
\]
where we also write $d_{BL}$ for the metric on $C([0,T];\P(\spY))$ generated by the metric $d_{BL}$ on $\spY$. 
Corollary~4.2.21 of \cite{DemboZeitouni98} requires that this set has super-exponentially small $\bbP_{\mu_n}$-probability. In fact, for every $\delta>0$  the set is empty for sufficiently large $n$, since by part~\ref{lemma:isometries:smeared} of Lemma~\ref{lemma:isometries} we have for any $\nu_n = \frac1n \sum_{i=1}^n \delta_{x_i}\in \spX$ that 
\begin{align*}
W_2(A\nu_n, A_n\nu_n)^2  &= W_2\biggl(\sum_{i=1}^n \delta^{\alpha,n}_{x_i+\alpha (i-1)/n},\frac1n \sum_{i=1}^n \delta_{x_i+\alpha (i-1)/n} \biggr)^2 \\
&= \sum_{i=1}^n W_2\Bigl(\delta^{\alpha,n}_{x_i+\alpha (i-1)/n}, \frac1n \delta_{x_i+\alpha (i-1)/n}\Bigr)^2\\
&= \sum_{i=1}^n \frac12 \frac{\alpha^2}{n^2} =  \frac12 \frac{\alpha^2}{n}.
\end{align*}
Since $d_{BL}(\mu,\nu) \leq  W_2(\mu,\nu)$ by Lemma~\ref{l:props-W2}, the set $\Gamma_{n,\delta}$ is empty for $n\geq \alpha^2/2\delta^2$.
Applying \cite[Corollary~4.2.21]{DemboZeitouni98} we find that $\mathcal A_n\mu_n$ satisfies a large-deviation principle in $C([0,T];\P(\spY))$ with good rate function 
\begin{equation}
\label{def:hatI}
\hat I(\rho) := \begin{cases}
I(\mu), & \text{if $A\mu(t) = \rho(t)$ for all $t\in [0,T]$,} \\
+\infty & \text{otherwise},
\end{cases}
\qquad\text{for }\rho\in C\big([0,T];\P(\spY)\big).
\end{equation}
It remains to prove the form~\eqref{def:RF-ind-part} of $\hat I$; {this is the content of the next lemma. }

\begin{lemma}
\label{lemma:char-I-rho-hrho}
Let $\rho\in C([0,T];\P(\spY))$ satisfy $\rho_0\in \P_2(\spY)$ and $\hFreeEnergy(\rho_0)+ \hat I(\rho)<\infty$. Define $\mu\in C([0,T];\P(\spX))$ by $\rho_t := A\mu_t$ for all $t$ (such $\mu$ exists by~\eqref{def:hatI}). We then have $\rho\in AC^2([0,T];\P_2(\spY))$ and $\mu\in AC^2([0,T];\P_2(\spX))$, and 
\begin{align*}
\hat I(\rho) \stackrel{(a)}= I(\mu) &\leftstackrel{(b)}= 2\, \Ent_V(\mu_0) + \frac12\int_0^T |\dot \mu|^2(t)\, dt  
+ \frac12\int_0^T |\partial \FreeEnergy|^2(\mu_t)\, dt 
+ \FreeEnergy(\mu_T)- \FreeEnergy(\mu_0)\\
&\leftstackrel{(c)}= 2\, \hEnt_V(\rho_0) + \frac12\int_0^T |\dot {\rho}|^2(t)\, dt  
+ \frac12\int_0^T |\partial \hFreeEnergy|^2(\rho_t)\, dt 
+  \hFreeEnergy( \rho_T)-  \hFreeEnergy(\rho_0).
\end{align*}
\end{lemma}

\begin{proof} 
The $W_2$ absolute continuity of $\mu$ is given by Lemma~\ref{lemma:I-to-v-formulation}, and the corresponding property of $\rho$ follows from the $W_2$-isometry of $A$.

Identity~(a) above follows from the definition~\eqref{def:hatI}. 
Identity~(c) is a direct consequence of the isometry of the mapping~$A$: under this isometry, all metric-space objects on $\P_2(\spX)$ are mapped one-to-one to corresponding objects on $\P_2(\spY)$, and this holds in particular for $\Ent_V$ and $\hEnt_V$, $\FreeEnergy$ and $\hFreeEnergy$, $|\partial \FreeEnergy|$ and $|\partial \hFreeEnergy|$,  and $|\dot\mu|$ and $|\dot \rho|$. 

To show identity~(b), note that 
since $\FreeEnergy(\mu_0)+  \mathfrak I_{\bbQ^{\mu_0}}(\mu) + \gamma(\mu_0) = \hFreeEnergy(\rho_0)+ \hat I(\rho)<\infty$ we have by Lemma~\ref{lemma:I-to-v-formulation}
\begin{align*}
I(\mu) &= \mathfrak I_{\bbQ^{\mu_0}}(\mu) +\gamma(\mu_0)\\
&= \RelEnt(\mu_0|\bbQ^{\mu_0}) +\gamma(\mu_0)  +  \frac12 \int_0^T \int_{\spX} 
  \mu_t(x) \Bigl(v_t(x) + \frac{\partial_x \mu_t}{2\mu_t}(x) - b(x,\mu_t)\Bigr)^2\, dxdt.
\end{align*}
The first two terms are equal to $2\, \Ent_V(\mu_0)$; we  rewrite the remainder as
\[
I(\mu) = 2\,\Ent_V(\mu_0) + \frac12\int_0^T \!\!\int_{\spX} v_t^2 \mu_t 
+ \frac12 \int_0^T\!\!\int_{\spX} \Bigl(\frac{\partial_x \mu_t}{2\mu_t} -b_t\Bigr)^2\mu_t
+ \int_0^T\!\!\int_{\spX} v_t \Bigl(\frac{\partial_x \mu_t}{2\mu_t} -b_t\Bigr)\mu_t.
\]
The first integral equals $\frac12\int_0^T |\dot \mu|^2(t)\, dt$ by the characterization~\eqref{eq:metric-derivative-v} of the velocity of absolutely-continuous curves.  By the characterizations~\eqref{char:b-minimal-element-subdiff} and~\eqref{char:metric-slope-subdifferential} of the element of minimal norm in the subdifferential, the second integral equals $\frac12\int_0^T |\partial \FreeEnergy|^2(\mu_t)\, dt $, and by the chain rule~\eqref{eq:chain-rule} the third integral equals $\FreeEnergy(\mu_T)- \FreeEnergy(\mu_0)$. 
\end{proof}

\newpage

\section{Proofs of Theorems~\ref{th:LDP-invmeas} and~\ref{th:LDP-path}}
\label{s:proof-dynamic-ldp}

\subsection{Proof of Theorem~\ref{th:LDP-path}}

In this section we finalize the proofs of the two main theorems. We first prove a version of Theorem~\ref{th:LDP-path} that allows a little more freedom in the initial data. 

Let $\mathcal G:\P(\R)\to\R$ be continuous and bounded, and set
\[
\InvMeasG_n(dy) := \frac1{\mathcal Z_n^{\mathcal G}} \exp \biggl[\,
  -n\mathcal G\Bigl(\frac1n\sum_{i=1}^n \delta_{y_i}\Bigr) - 2 \sum_{i=1}^n V(y_i) \biggr] \, \Lebesgue^n\Big|_{\Omega_n}(dy).
\]
We also define the modified free energy
\begin{equation}
\label{def:FE-G}
\hFreeEnergy^{\mathcal G}(\rho) := \hEnt_V(\rho) + \frac12 \mathcal G(\rho) + C_{\mathcal G},
\end{equation}
where the constant $C_{\mathcal G}$ is such that $\inf \hFreeEnergy^{\mathcal G}=0$.

Then
\begin{itemize}
\item For $\mathcal G=0$ the distribution $\InvMeasG_n$ coincides with $\InvMeasWiszero_n$;
\item For $\mathcal G(\rho) = \int f\rho$  the distribution $\InvMeasG_n$ coincides with $\InvMeasf_n$ and $\hFreeEnergy^{\mathcal G}$ with $\hFreeEnergy^f$;
\item For $\mathcal G(\rho) = \iint W(x-y)\rho(dx)\rho(dy)$ the distribution $\InvMeasG_n$ coincides with  $\InvMeas_n$ and $\hFreeEnergy^{\mathcal G}$ with $\hFreeEnergy$. 
\end{itemize}

\begin{theorem}[Large-deviation principle on path space, version with general initial distribution]
\label{th:LDP-path-strongestversion}
Assume that $V,W$ satisfy Assumption~\ref{ass:VW}. 
For each $n$, let the particle system $t\mapsto Y^n(t)\in\R^n$ be given by~\eqref{eq:SDE}, with initial positions drawn from the modified invariant measure $\InvMeasG_n$.

The random evolving empirical measures $\rho_n(t) =\frac1n\sum_{i=1}^n \delta_{Y^n_i(t)} $ then satisfy a large-deviation principle on $C\bigl([0,T];\P(\R)\bigr)$ with good rate function $\hat I^{\mathcal G}$. 
%\[
%\Prob\Bigl(\rho_n|_{t\in[0,T]} \,\approx \,\nu|_{t\in[0,T]}\Bigr) \sim 
%e^{-n\hat I^{\mathcal G}(\nu)}\qquad\text{as }n\to\infty.
%\]
%
If in addition $\rho$ satisfies  $\rho(0)\in \P_2(\R)$ and  $\hFreeEnergy(\rho(0)) + \hat I^{\mathcal G}(\rho) < \infty$, then we have $\rho\in C([0,T];\P_2(\R))$ and $\hat I^{\mathcal G}(\rho)$ can be characterized as
\begin{align}
\hat I^{\mathcal G}(\rho) := 2\, &\hFreeEnergy^{\mathcal G}(\rho(0)) + \hFreeEnergy( \rho(T))-  \hFreeEnergy(\rho(0))
+\frac12\int_0^T |\dot {\rho}|^2(t)\, dt  
+ \frac12\int_0^T |\partial \hFreeEnergy|^2(\rho(t))\, dt.
\label{def:RF-ind-part-strongest}
\end{align}
Here $|\dot {\rho}|$ and $|\partial \hFreeEnergy|$ are the metric derivative and the local slope defined in Definition~\ref{def:metric-GFs}, for the Wasserstein metric space $\mathcal X=\P_2(\R)$.
\end{theorem}

Theorem~\ref{th:LDP-path} is Theorem~\ref{th:LDP-path-strongestversion} for the special case $\mathcal G(\rho) = \int f\rho$.

\begin{proof}[Proof of Theorem~\ref{th:LDP-path-strongestversion}]
As in the theorem, let $Y^n$ be solutions of equation~\eqref{eq:SDE} (or equivalently of Definition~\ref{def:particle-systems}\eqref{i:def:particle-systems:Y}) with initial data drawn from $\InvMeasG_n$. Let $\widetilde Y^n$ be solutions of the same system, but with initial data $\widetilde Y^n(0)$ drawn from $\InvMeasWiszero_n$. Let $\rho_n,\, \widetilde \rho_n \in C([0,T];\P(\spY))$ be the corresponding empirical measures, and write $\bbP_{\rho_n}$ and $\bbP_{\widetilde \rho_n}$ for their laws. By Theorem~\ref{th:LDP-pathwise-weakversion}, $\widetilde \rho_n$ satisfies a large-deviation principle in $C([0,T];\P(\spY))$ with the rate function $\hat I$ defined in~\eqref{def:hatI}.

Since the evolution of the two particle systems is the same, conditioned on their initial positions, we have as in the proof of Theorem~\ref{th:LDP-dynamic-X} that for any $\rho\in C([0,T];\P(\spY))$
\[
\frac{d\bbP_{\rho_n}}{d \bbP_{\widetilde\rho_n}}(\rho) = 
\frac{d(\eta_n)_\# \InvMeasG_n}{d (\eta_n)_\# \InvMeasWiszero_n}(\rho_0)
= C_n \exp\bigl( - n\mathcal G(\rho_0) \bigr)
\]
for some normalization constants $C_n>0$. Since the exponent is a bounded and narrowly continuous function of $\rho_0$, Varadhan's Lemma (e.g.~\cite[Th.~4.3.1]{DemboZeitouni98}) implies that $\rho_n$ satisfies a large-deviation principle in $C([0,T];\P(\spY))$ with rate function
\[
\hat I^{\mathcal G} (\rho) := \hat I(\rho) + \mathcal G(\rho_0) + C,
\]
where  the constant $C$ is chosen such that $\inf \hat I^{\mathcal G}=0$.

We now show the formula~\eqref{def:RF-ind-part-strongest}. Set 
\[
\mathfrak F^{\mathcal G}(\nu) := \inf\Bigl\{ \hat I^{\mathcal G}(\rho): \rho_0 = \nu\Bigr\}.
\]
We prove that $\mathfrak F^{\mathcal G} = 2\hFreeEnergy^{\mathcal G}$. 
Taking any $\mu$ with $I(\mu)<\infty$, by Lemma~\ref{l:char-RF-path} there exists a {unique} function $u$ such that
\[
I(\mu) = \mathfrak I_{\bbQ^{\mu_0}}(\mu) + \gamma(\mu_0) =
\RelEnt(\mu_0|\bbQ^{\mu_0}) + \gamma(\mu_0) + 
\frac12 \int_0^T \int_{\spX} u^2(t,x)\, \mu_t(dx) dt.
\]
By repeating this identity for a sequence $T\downarrow 0$ and using the uniqueness of $u$ we find that $\inf\{I(\mu) : \mu_0 = \xi\} = \RelEnt(\xi|\bbQ^{\xi}) + \gamma(\xi)=2\Ent_V(\xi)$.
We then observe that 
\begin{align*}
\mathfrak F^{\mathcal G}(\nu) &= \inf\Bigl\{ \hat I(\rho) + \mathcal G(\rho_0) + C: \rho_0 = \nu\Bigr\}\\
&= \inf\Bigl\{ I(\mu) : A\mu_0 = \nu\Bigr\} + \mathcal G(\nu) + C\\
&= 2\hEnt_V(\nu) + \mathcal G(\nu) + C\\
& \leftstackrel{\eqref{def:FE-G}}= 2\hFreeEnergy^{\mathcal G}(\nu) + C - 2C_{\mathcal G}.
\end{align*}
Since both $\mathfrak F^{\mathcal G}$ and $2\hFreeEnergy^{\mathcal G}$ have zero infimum, $C=2C_{\mathcal G}$. 
It follows that 
\[
\hat I^{\mathcal G}(\rho) = \hat I(\rho) + \mathcal G(\rho_0) + 2C_{\mathcal G}
= \hat I(\rho) - 2\hEnt_V(\rho_0) + 2\hFreeEnergy^{\mathcal G}(\rho_0),
\]
and~\eqref{def:RF-ind-part} then implies  the characterization~\eqref{def:RF-ind-part-strongest}.
\end{proof}

\subsection{Proof of Theorem~\ref{th:LDP-invmeas}}

We deduce the  large-deviation principle for the invariant measures $\InvMeas_n$ from Theorem~\ref{th:LDP-path-strongestversion} by applying the contraction principle. Let $Y^n(0)$ be drawn from $\InvMeas_n$, and let $Y^n$ be solutions of Definition~\ref{def:particle-systems}\eqref{i:def:particle-systems:Y} over a time interval $[0,T]$. Applying Theorem~\ref{th:LDP-path-strongestversion} with $\mathcal G(\rho) := \iint W(x-y)\rho(dx)\rho(dy)$ we find that $\rho_n = \eta_n(Y^n)$ satisfies a large-deviation principle on $C([0,T];\P(\spY))$ with good rate function $\hat I^{\mathcal G}$ given by~\eqref{def:RF-ind-part}. Since the evaluation map 
\[
C([0,T];\P(\spY)) \to \P(\spY), \qquad \mu \mapsto \mu(0)
\]
is continuous, the contraction principle (see e.g.~\cite[Th.~4.21]{DemboZeitouni98}) implies that $\rho_n(0)$ satisfies a large-deviation principle with good rate function
\[
\mathfrak F^{\mathcal G}(\nu) := \inf\Bigl\{ \hat I^{\mathcal G}(\rho): \rho_0 = \nu\Bigr\}.
\]
In the proof of Theorem~\ref{th:LDP-path-strongestversion} above we showed that $\mathfrak F^{\mathcal G}$ equals $2\hFreeEnergy^{\mathcal G}$, which coincides with $2\hFreeEnergy$ for this choice of $\mathcal G$; the characterization of Lemma~\ref{lemma:properties-of-the-functionals} then concludes the proof. 

\appendix

\section{Proof of Lemma~\ref{l:ex-un-SDEs}}

The existence and uniqueness of weak solutions  $X^n$ follows from arguments similar to the proof of Lemma~\ref{l:char-RF-path}, and we omit this proof. 

For the stochastic process $Y^n$ we prove the existence and uniqueness as follows. The main step is to show the unique existence of a fundamental solution of the parabolic partial differential equation $\partial_t \rho - \mathcal L_Y^*\rho=0$ on $\Omega_n$:
\begin{lemma}
\label{l:ex-un-fund-sol}
Fix $T>0$. 
For each $y_0\in \Omega_n$ there exists a unique fundamental solution $(t,y)\mapsto p(t,y;y_0)$ of the equation $\partial_t -\mathcal L_Y^* =0$ on $\Omega_n$ with Neumann boundary conditions, i.e., a function $p = p_{y_0}\in L^1((0,T)\times \Omega_n)$ satisfying 
\begin{equation}
\label{eq:weak-fund-sol}
0 = \int_0^T \int_{\Omega_n} p(t,y) \big[\partial_t \varphi(t,y) +(\mathcal L_Y \varphi)(t,y)\big]\, dy dt + \varphi(y_0)
\end{equation}
%\label{eq:fund-sol}
%\begin{alignat}2
%\partial_tp &= \mathcal L_Y^* p , &\qquad & \text{on }(0,T)\times \Omega_n,\\
%\partial_n p &= 0 && \text{on }(0,T)\times \partial\Omega_n,\\
%p(t=0,\cdot) &= \delta_{y_0}.
%\end{alignat}
for all $\varphi\in C^{1,2}_b([0,T]\times\Omega_n)$. 
In addition, for each $t>0$ the function $y\mapsto p(t,\cdot;y_0)$ is non-negative and  continuous, and satisfies $\int_{\Omega_n} p(t,y;y_0) \, dy= 1$.
\end{lemma}

Informally, the fundamental solution $p$  satisfies
\begin{alignat}2
\partial_tp &= \mathcal L_Y^* p , &\qquad & \text{on }(0,T)\times \Omega_n,\\
\partial_n p &= 0 && \text{on }(0,T)\times \partial\Omega_n,\\
p(t=0,\cdot) &= \delta_{y_0}.
\end{alignat}
%For a finite measure $p$ on $[0,T]\times \Omega_n$ (including a Borel measurable family of probability measures $t\mapsto p(t,dy)$, or  an element of $L^1((0,T)\times\Omega_n)$ such as given by the lemma above) the weak formulation of the equations~\eqref{eq:fund-sol} is the following. For all $\varphi\in C^{1,2}([0,T]\times\Omega_n)$ with $\varphi=0$ at $t=T$, 
%
Given this fundamental solution,   the proof of Lemma~\ref{l:ex-un-SDEs} proceeds along classical lines. We  construct a consistent family of finite-dimensional distributions $P_{t_1,\dots,t_k}\in \P((\Omega_n)^k)$ in the usual way, by daisy-chaining copies of the fundamental solution $p({t_k-t_{k-1}}, \,\cdot\;;\,\cdot\,)$ (see e.g.~\cite[Th.~2.2.2]{StroockVaradhan97}). By applying the maximum-principle method of~\cite[Cor.~3.1.3]{StroockVaradhan97} we show that this consistent family satisfies the conditions of Kolmogorov's continuity theorem, and Theorem~2.1.6 of~\cite{StroockVaradhan97} then implies that $P_{t_1,\dots,t_k}$ is generated by a unique probability measure~$\mathbb P$ on the space $C([0,\infty);\Omega_n)$. This concludes the argument. 

\medskip

The main step therefore is the proof of Lemma~\ref{l:ex-un-fund-sol}, which we now give.

\begin{proof}
Define for $L>\alpha$ the truncated state space 
\[
\Omega_n^L := \big\{ y\in [-L,L]^n: |y_i-y_j|\geq \alpha/n \text{ for all }i\not=j\big\}.
\]
Fix an initial datum $\phi\in C_b(\Omega_n^L)$, $\phi\geq0$; by {classical methods} there exists a non-negative $C^{1,2}_b$ solution $u$ of the equation $\partial_t u - \mathcal L_Y^* u = 0$ on $(0,T)\times \Omega_n^L$, $\partial_n u = 0$ on $(0,\infty)\times \partial\Omega_n^L$, and $u(t=0) = \phi$. By integrating the equation over $[0,t]\times \Omega_n^L$ we find
\[
\int_{\Omega_n^L} u(t,y)\, dy = \int_{\Omega_n^L} \phi(y)\, dy.
\]

We now take a sequence  $L\to\infty$ and choose  $\phi^L\geq 0$ with $\int \phi^L = 1$ such that $\phi^L$  converges narrowly on $\Omega_n$ to $\delta_{y_0}$. For each $T>0$, the corresponding solution $(t,y)\mapsto u^L(t,y)$ is non-negative and has integral over $[0,T]\times \Omega_n$ equal to $T$ (where we extend $u^L$ on $\Omega_n\setminus \Omega_n^L$ by zero); by taking  a subsequence we can therefore assume that $u^L$ converges weakly, in duality with $C_c([0,T]\times \Omega_n)$, to a non-negative limit measure $p$ with $p([0,T]\times \Omega_n)\leq T$. 
By e.g.~\cite[Th.~6.4.1]{BogachevKrylovRocknerShaposhnikov15} the measure $p$ has a continuous Lebesgue density on $(0,T)\times \Omega_n$, implying that for $0<t<T$ we can write it as $p(dtdt) = p(t,y)dtdy$. 
%Since $\int_{\Omega_n} u^L(t, y)\, dy =1$ for all $t$, and since narrow convergence implies narrow convergence of marginals, the measure $p(\cdot,\Omega_n)$ on $[0,T]$ coincides with Lebesgue measure on $[0,T]$, and therefore $\int_{\Omega_n}p(t,y)dy = 1$ for all $t\in (0,T)$.

By e.g.~\cite[Th.~6.4.1]{BogachevKrylovRocknerShaposhnikov15} the measure $p$ has a continuous Lebesgue density on $(0,T)\times \Omega_n$, implying that for $0<t<T$ we can write it as $p(dtdt) = p(t,y)dtdy$. Since $\int_{\Omega_n} u^L(t, y)\, dy =1$ for all $t$, and since narrow convergence implies narrow convergence of marginals, the measure $p(\cdot,\Omega_n)$ on $[0,T]$ coincides with Lebesgue measure on $[0,T]$, and therefore $\int_{\Omega_n}p(t,y)dy = 1$ for all $t\in (0,T)$.

We now show that the function $p$ satisfies~\eqref{eq:weak-fund-sol}. Take a function $\varphi\in C^{1,2}_c([0,T)\times \Omega_n)$ satisfying $\partial_n\varphi=0$ on $\partial\Omega_n$. Then $\varphi\in C^{1,2}_c([0,T)\times \Omega_n^L)$ for sufficiently large $L$, and 
in the weak form of the equation $\partial_t u^L - \mathcal L_Y u^L = 0$ with initial datum $\phi$, 
\[
0 = \int_0^T \int_{\Omega_n^L} u^L(t,y) \big[\partial_t \varphi(t,y) +(\mathcal L_Y \varphi)(t,y)\big]\, dy dt + \int_{\Omega_n^L} \phi^L(y)\varphi(y)\, dy,
\]
we can replace the domain of integration $\Omega_n^L$ by $\Omega_n$. By taking the limit $L\to\infty$ we find for all such  $\varphi$ the property
\begin{equation*}
%\label{eq:weak-fund-sol}
0 = \int_0^T \int_{\Omega_n} p(t,y) \big[\partial_t \varphi(t,y) +(\mathcal L_Y \varphi)(t,y)\big]\, dy dt + \varphi(y_0).
\end{equation*}
By a standard approximation argument, using the fact that the total mass of $p$ is finite, this identity can be shown to hold for all $\varphi\in C^{1,2}_b([0,T]\times \Omega_n)$ with $\varphi(t=T) = 0$. This proves~\eqref{eq:weak-fund-sol}.
By taking $\varphi$ in~\eqref{eq:weak-fund-sol} to be a function only of $t$, we also find $\partial_t \int_{\Omega_n} p(t,y)\, dy = 0$ in distributional sense, and therefore $p(t,\cdot)$ has unit mass for all time $t$.

\medskip
Finally, we prove the uniqueness of $p$, which also implies that the final time $T$ can be taken equal to $\infty$. Let $p$ be a finite measure on $[0,T]\times \Omega_n$ that satisfies the weak equation~\eqref{eq:weak-fund-sol} with initial datum equal to zero, i.e. assume that for all $\varphi\in C^{1,2}([0,T]\times\Omega_n)$ with $\varphi(t=T)=0$, 
\begin{equation}
\label{eq:fund-sol-zero-id}
0 = \int_0^T \int_{\Omega_n} p(t,y) \big[\partial_t \varphi(t,y) +(\mathcal L_Y \varphi)(t,y)\big]\, dy dt.
\end{equation}
Fix $\chi\in C_b([0,T]\times\Omega_n)$. By arguments very similar to those above we can find a solution $\varphi\in C^{1,2}([0,T]\times\Omega_n)$ of the equation
\begin{alignat*}2
&\partial_t \varphi + \mathcal L_Y\varphi =\chi, &\qquad & \text{on }(0,\infty)\times \Omega_n,\\
&\partial_n \varphi = 0 && \text{on }(0,\infty)\times \partial\Omega_n,\\
&\varphi(t=T) = 0.
\end{alignat*}
By substituting this $\varphi$ in~\eqref{eq:fund-sol-zero-id} we find 
\[
\int_0^T \int_{\Omega_n} p(t,y)\chi(t,y)\, dydt= 0\qquad\text{for all }\chi\in C_b([0,T]\times\Omega_n).
\]
This implies that $p$ is the zero measure on $[0,T]\times \Omega_n$, and proves the uniqueness of solutions of~\eqref{eq:weak-fund-sol}.
\end{proof}

\bigskip
\textbf{Acknowledgements. }The authors would like to thank Jim Portegies, Oliver Tse, Jasper Hoeksema, Georg Prokert, and Frank Redig for several interesting discussions and insightful remarks. This work was partially supported by NWO grant 613.009.101.

\bibliographystyle{alphainitials}
\bibliography{ref-steric}

\end{document}